\pdfoutput=1
\documentclass[11pt]{article}

\usepackage{threeparttable}
\usepackage{array,tabularx,booktabs,threeparttable}
\usepackage{fullpage}

\usepackage[utf8]{inputenc}
\DeclareUnicodeCharacter{202F}{~}
\usepackage[english]{babel}
\usepackage[shortlabels]{enumitem}
\usepackage{mathtools,
  booktabs,
  braket,
  mathrsfs,
  latexsym,
  epsfig,
  xcolor,
  url,
  setspace,
  graphics,
  lipsum,
  float,
  lineno,
  placeins
}
\usepackage{subfigure}
\usepackage{float}
\usepackage{pdflscape}
\newcommand{\bland}{\begin{landscape}}
\newcommand{\eland}{\end{landscape}}
\usepackage{amsmath}
\usepackage{amsfonts}
\usepackage{amssymb}
\usepackage{amsthm}
\usepackage[margin=0.95in]{geometry}

\usepackage{thmtools} 
\theoremstyle{plain}
\newtheorem{lemma}{Lemma}
\usepackage[flushmargin]{footmisc}
\definecolor{CardinalRed}{cmyk}{0,1,0.65,0.34}
\definecolor{NavyBlue}{RGB}{0,60,180}
\usepackage[colorlinks,
            linkcolor=NavyBlue,
            citecolor=NavyBlue,
            urlcolor=NavyBlue]{hyperref}
\usepackage{hyperref}

\makeatletter
\def\maxwidth{\ifdim\Gin@nat@width>\linewidth\linewidth\else\Gin@nat@width\fi}
\def\maxheight{\ifdim\Gin@nat@height>\textheight\textheight\else\Gin@nat@height\fi}
\makeatother

\setkeys{Gin}{width=\maxwidth,height=\maxheight,keepaspectratio}

\newcommand{\burl}[1]{\textcolor{blue}{\url{#1}}}

\usepackage[round]{natbib}

\usepackage[OT1]{fontenc}

\usepackage{todonotes}

\makeatletter
\providecommand\@dotsep{5}
\def\listtodoname{List of Todos}
\def\listoftodos{\@starttoc{tdo}\listtodoname}
\makeatother

\graphicspath{{./graphs/}}%

\usepackage[labelsep=period, font=sc, skip=0pt,justification=centering]{caption}
\usepackage[figuresright]{rotating}
\usepackage{titlesec}
\titlelabel{\thetitle.\ }

\usepackage{titling}
\usepackage{titlesec}
\titleformat{\section}
{\normalfont\fontsize{15}{15}\bfseries}{\thesection.}{0.5em}{}

\newenvironment{itemize*}%
  {\begin{itemize}%
    \setlength{\itemsep}{0pt}%
    \setlength{\parskip}{0pt}}%
  {\end{itemize}}
\newenvironment{enumerate*}%
  {\begin{enumerate}%
    \setlength{\itemsep}{0pt}%
    \setlength{\parskip}{0pt}}%
  {\end{enumerate}}

\usepackage{
  amscd,
  amsfonts,
  amsmath,
  amssymb,
  amsbsy,
  amsthm,
  bm,
  dsfont,
  cancel,
  latexsym,
  mathtools,
  graphicx,
  xcolor,
  xargs
}

\usepackage{longtable}

\newcommand{\beq}{\begin{equation}}
\newcommand{\eeq}{\end{equation}}

\newcommand*\Bigpar[1]{\left( #1 \right )}

\newcommand{\ba}{\begin{array}}
\newcommand{\ea}{\end{array}}
\newcommand{\be}{\begin{enumerate}}
\newcommand{\ee}{\end{enumerate}}
\newcommand{\bi}{\begin{itemize}}
\newcommand{\ei}{\end{itemize}}
\newcommand{\bs}{\begin{align}\begin{split}\nonumber}
\newcommand{\bsnumber}{\begin{align}\begin{split}}
\newcommand{\es}{\end{split}\end{align}}

\newcommandx{\deriv}[2][1=x,2=f]{\nabla \, #2 \Bigpar{ #1 } }

\renewcommand{\det}{\mathrm{det}}

\renewcommand{\to}{{\rightarrow}}

\newcommand\frakfamily{\usefont{U}{yfrak}{m}{n}}
\DeclareTextFontCommand{\textfrak}{\frakfamily}

\renewcommand{\det}{\mathrm{det}}

\renewcommand{\to}{{\rightarrow}}

\newcommand{\E}{\mathbb{E}}

\newcommand{\hyp}[2]{
\ensuremath{H_0:#1 \ifhmode\quad\text{versus}\quad\fi\text{ vs. } H_1:#2}}

\newcommandx{\uniff}[1][1={a,b}]{\textrm{Unif}\left({#1}\right)}
\newcommandx{\unifd}[1][1={a,\ldots,b}]{\textrm{Unif}\left\{{#1}\right\}}
\newcommandx{\dunif}[3][1=x,2=a,3=b]{\frac{I(#2<#1<#3)}{#3-#2}}
\newcommandx{\dunifd}[3][1=x,2=a,3=b]{\frac{I(#2\le#1\le#3)}{#3-#2+1}}
\newcommandx{\punif}[3][1=x,2=a,3=b]{
\begin{cases} 0 & #1 < #2 \\ \frac{#1-#2}{#3-#2} & #2 < #1 < #3 \\ 1 & #1 > #3\\\end{cases}}
\newcommandx{\punifd}[3][1=x,2=a,3=b]{
\begin{cases} 0 & #1 < #2\\ \frac{\lfloor#1\rfloor-#2+1}{#3-#2} & #2 \le #1 \le #3 \\ 1 & #1 > #3\\ \end{cases}}

\newcommandx\bern[1][1=p]{\textrm{Bern}\left({#1}\right)}
\newcommandx\dbern[2][1=x,2=p]{#2^{#1} \left(1-#2\right)^{1-#1}}
\newcommandx\pbern[2][1=x,2=p]{\left(1-#2\right)^{1-#1}}

\newcommandx\bin[1][1={n,p}]{\textrm{Bin}\left(#1\right)}
\newcommandx\dbin[3][1=x,2=n,3=p]{\binom{#2}{#1}#3^#1\left(1-#3\right)^{#2-#1}}

\newcommandx\mult[1][1={n,p}]{\textrm{Mult}\left(#1\right)}
\newcommandx\dmult[3][1=x,2=n,3=p]{\frac{#2!}{#1_1!\ldots#1_k!}#3_1^{#1_1}\cdots#3_k^{#1_k}}

\newcommandx\hyper[1][1={N,m,n}]{\textrm{Hyp}\left({#1}\right)}
\newcommandx\dhyper[4][1=x,2=N,3=m,4=n]{\frac{\binom{#3}{#1}\binom{#2-#3}{#4-#1}}{\binom{#2}{#4}}}

\newcommandx\nbin[1][1={r,p}]{\textrm{NBin}\left({#1}\right)}
\newcommandx\dnbin[3][1=x,2=r,3=p]{\binom{#1+#2-1}{#2-1}#3^#2(1-#3)^#1}
\newcommandx\pnbin[3][1=x,2=r,3=p]{I_#3(#2,#1+1)}

\newcommandx\geo[1][1=p]{\textrm{Geo}\left(#1\right)}
\newcommandx\dgeo[2][1=x,2=p]{#2(1-#2)^{#1-1}}
\newcommandx\pgeo[2][1=x,2=p]{1-(1-#2)^#1}

\newcommandx\pois[1][1=\lambda]{\textrm{Po}\left({#1}\right)}
\newcommandx\dpois[2][1=x,2=\lambda]{\frac{#2^#1 e^{-#2}}{#1!}}
\newcommandx\ppois[2][1=x,2=\lambda]{e^{-#2}\sum_{i=0}^#1\frac{#2^i}{i!}}

\newcommandx\normall[1][1={\mu,\sigma^2}]{\mathcal{N}\left({#1}\right)}
\newcommandx\dnormall[3][1=x,2=\mu,3=\sigma]%
  {\frac{1}{#3\sqrt{2\pi}}\exp \Bigpar{-\frac{\left(#1-#2\right)^2}{2 #3^2}}}
\newcommandx\pnormall[1][1=x]{\Phi\left({#1}\right)}
\newcommandx\qnormall[1]{\Phi^{-1}\left({#1}\right)}

\newcommandx\mvn[1][1={\mu,\Sigma}]{\mathrm{MVN}\left({#1}\right)}

\newcommandx\ex[1][1=\lambda]{\textrm{Exp}\left(#1\right)}
\newcommandx\dex[2][1=x,2=\lambda]{#2e^{-#1 #2}}
\newcommandx\pex[2][1=x,2=\lambda]{1-e^{-#1 #2}}

\newcommandx\gam[1][1={\alpha,\lambda}]{\textrm{Gamma}\left({#1}\right)}
\newcommandx\dgamma[3][1=x,2=\alpha,3=\lambda]%
  {\frac{#3^{#2}}{\Gamma\left( #2 \right)} #1^{#2-1}e^{-#3#1}}

\newcommandx\invgamma[1][1={\alpha,\beta}]{\textrm{InvGamma}\left({#1}\right)}
\newcommandx\dinvgamma[3][1=x,2=\alpha,3=\beta]%
{\frac{#3^{#2}}{\Gamma\left(#2\right)}#1^{-#2-1}e^{-#3/#1}}
\newcommandx\pinvgamma[3][1=x,2=\alpha,3=\beta]%
{\frac{\Gamma\left(#2,\frac{#3}{#1}\right)}{\Gamma\left(#2\right)}}

\newcommandx\bet[1][1={\alpha,\beta}]{\textrm{Beta}\left(#1\right)}
\newcommandx\dbeta[3][1=x,2=\alpha,3=\beta]
{\frac{\Gamma\left(#2+#3\right)}{\Gamma\left(#2\right)\Gamma\left(#3\right)}#1^{#2-1}\left(1-#1\right)^{#3-1}}

\newcommandx\dir[1][1={\alpha}]{\textrm{Dir}\left(#1\right)}
\newcommandx\ddir[3][1=x,2=\alpha]{\frac{\Gamma\left(\sum_{i=1}^k #2_i\right)}{\prod_{i=1}^k\Gamma\left(#2_i\right)}\prod_{i=1}^k #1_i^{#2_i-1}}

\newcommandx\weibull[1][1={\alpha}]{\textrm{Dir}\left(#1\right)}
\newcommandx\dweibull[3][1=x,2=\lambda,3=k]{\frac{#3}{#2}
\left(\frac{#1}{#2}\right)^{#3-1} e^{-(#1/#2)^k}}

\newcommandx\chisq[1][1=k]{\chi_{#1}^2}

\newcommandx\zet[1][1=s]{\textrm{Zeta}\left(#1\right)}
\newcommandx\dzeta[2][1=x,2=s]{\frac{#1^{-#2}}{\zeta\left(#2\right)}}

\newtheoremstyle{mystyle}%
  {12pt}%
  {12pt}%
  {}%
  {}%
  {\sffamily \bfseries }%
  {.}%
  {0.5em}%
  {\thmname{#1}\thmnumber{ #2}\thmnote{ (#3)}}%

\theoremstyle{mystyle}

\newtheorem{assumption}{Assumption}

\renewenvironment{proof}{\noindent{\bf Proof}\hspace*{1em}}{\qed\bigskip\\}
\newenvironment{proof-sketch}{\noindent{\bf Sketch of Proof}
  \hspace*{1em}}{\qed\bigskip\\}
\newenvironment{proof-idea}{\noindent{\bf Proof Idea}
  \hspace*{1em}}{\qed\bigskip\\}
\newenvironment{proof-of-lemma}[1][{}]{\noindent{\bf Proof of Lemma {#1}}
  \hspace*{1em}}{\qed\bigskip\\}
\newenvironment{proof-of-proposition}[1][{}]{\noindent{\bf
    Proof of Proposition {#1}}
  \hspace*{1em}}{\qed\bigskip\\}
\newenvironment{proof-of-theorem}[1][{}]{\noindent{\bf Proof of Theorem {#1}}
  \hspace*{1em}}{\qed\bigskip\\}
\newenvironment{inner-proof}{\noindent{\bf Proof}\hspace{1em}}{
  $\bigtriangledown$\medskip\\}
\newenvironment{proof-attempt}{\noindent{\bf Proof Attempt}
  \hspace*{1em}}{\qed\bigskip\\}

\usepackage{array}
\newcolumntype{L}[1]{>{\raggedright\let\newline\\\arraybackslash\hspace{0pt}}m{#1}}
\newcolumntype{C}[1]{>{\centering\let\newline\\\arraybackslash\hspace{0pt}}m{#1}}
\newcolumntype{R}[1]{>{\raggedleft\let\newline\\\arraybackslash\hspace{0pt}}m{#1}}

\usepackage{multirow}

\def\E{\mathbb{E}}

\usepackage{makecell}

\newcommand{\RR}{RR}
\newcommand{\Imp}{\text{Imp}}
\newcommand{\CSNYT}{\textit{CS-NYT}}
\newcommand{\CSNT}{\textit{CS-NT}}
\newcommand{\pre}{\text{pre}}
\newcommand{\post}{\text{post}}
\newcommand{\cvq}{\mathrm{cv}}
\ifcsname proposition\endcsname\else\newtheorem{proposition}{Proposition}\fi

\begin{document}
\onehalfspacing

\title{Cohort-Anchored HonestDiD\Large\bf
\thanks{Ziyi Liu (Corresponding Author), Ph.D.\ student, Haas School of Business, University of California, Berkeley. Email: \url{zyliu2023@berkeley.edu}. An earlier version circulated under the title ``Cohort-Anchored Robust Inference for Event-Study with Staggered Adoption.'' I thank Kirill Borusyak, Jonathan Roth and Yiqing Xu for invaluable discussions and constructive feedback, and the participants at the 42nd PolMeth Poster Session and UC Berkeley’s BPP Student Seminar for their helpful comments. I thank the associate editor and two anonymous reviewers at the Journal of Econometrics for their constructive feedback. I use Claude (Anthropic) and Refine.ink to assist with the preparation of this manuscript.}
\\\bigskip}

\author{Ziyi Liu\\(UC Berkeley)}

\date{
  \today
  \vspace{2em}
}

\maketitle

\vspace{-2em}
\begin{abstract}
\noindent HonestDiD sensitivity analysis \citep{rambachan2023more} is usually applied to event-study coefficients aggregated across cohorts. Under staggered adoption, the treated cohorts contributing to these coefficients change with event time, as do their control groups with estimators that use not-yet-treated units. Therefore, the restriction on the aggregated coefficients and the corresponding inference does not have a clear interpretation. This paper proposes to do the HonestDiD sensitivity analysis on a cohort basis. The restriction is placed on each cohort’s block bias, which captures its violation of parallel trends relative to the units untreated when it first adopts treatment. The restriction on the block bias maps to the restriction on the estimation bias through an invertible linear decomposition. In simulations, the proposed cohort inference keeps nominal coverage when parallel trend violations differ across cohorts, where the inference based on aggregated coefficients can fail. Applying the method to \citet{callaway2021-did} and \citet{bailey2015publicmedicine}, I highlight the trade-off between interpretable imposed restrictions and the precision demands of cohort inference.

\end{abstract}

\thispagestyle{empty}
\clearpage

\doublespace
\setcounter{page}{1}
\abovedisplayskip=5pt
\belowdisplayskip=5pt

\section{Introduction}
\label{sec:introduction}

Difference-in-differences (DiD) methods identify treatment effects by assuming that treated and control groups would have followed parallel trends in the absence of treatment \citep{baker2025difference}. Researchers commonly assess this parallel trends (PT) assumption by inspecting pre-treatment trends or testing whether pre-treatment event-study coefficients are jointly zero. Neither exercise can test the assumption after treatment.

Pre-treatment tests have limited power against substantively meaningful violations of PT, and conditioning the post-treatment analysis on passing a pre-test distorts inference \citep{roth2022pretest}. The robust inference framework of \citet{rambachan2023more} (hereafter \RR) addresses this problem by allowing bounded departures from PT. Implemented in the \texttt{HonestDiD} package \citep{honestdid026} and referred to here as HonestDiD, it uses pre-treatment PT deviations to restrict post-treatment PT violations. The treatment effect under the imposed restrictions is then set-identified, and the robust inference framework delivers valid \textit{confidence sets} under specified restrictions.

Under staggered adoption, cohorts adopt treatment in different periods and remain treated. Modern DiD estimators allow for heterogeneous treatment effects across cohorts and periods relative to treatment \citep[e.g.,][]{sun2021-event,callaway2021-did, borusyak2024revisiting, de2024difference, wooldridge2021two}. However, researchers commonly apply HonestDiD after aggregating the cohort estimates into an aggregate event-study coefficient path. Two changes in composition complicate the application of HonestDiD. First, the composition of treated units behind the aggregate coefficients differs between pre-treatment and post-treatment periods. For the aggregate event-study coefficients, the observed pre-treatment PT violation is therefore not a proper benchmark for the potential PT violation in post-treatment periods.

Second, even for a single treated cohort, the control group may also change over time for estimators relying on not-yet-treated controls. For example, the not-yet-treated version of the \citet{callaway2021-did} estimator compares each cohort with units not yet treated, which is a pool that shrinks as adoption proceeds. The \citet{borusyak2024revisiting} imputation estimator fits the untreated counterfactual on untreated observations, which a cohort leaves once it adopts. Applying HonestDiD to each treated cohort separately still faces this complexity, as comparisons between cohort-period coefficients have different meanings.

To address these two complications under staggered adoption, I propose a cohort-anchored robust inference framework (hereafter \emph{cohort inference}). It imposes restrictions on cohort-level PT violations and uses cohort-period coefficients to construct the confidence set, rather than applying \RR{} to aggregate event-study coefficients (hereafter \emph{aggregate inference}). To address the second complication of changing control groups, I introduce the concept of \emph{block bias}. Block bias is the difference between a treated cohort's untreated trend and that of the units still untreated when the cohort first adopts treatment. Because this comparison population is fixed by construction, pre- and post-treatment block biases describe the same cohort relative to the same control group. Observable pre-treatment block biases can therefore benchmark unobserved post-treatment ones. Building on \citet{arkhangelsky2024sequential} and \citet{aguilar2025comparative}, I write the bias of every post-treatment \Imp{} and \CSNYT{} coefficient as a known linear combination of block biases. A restriction on block biases then translates into a restriction on the estimators' biases, and \RR's framework applies directly.\footnote{I study these two estimators throughout the paper, but the framework itself can be extended to other modern DiD estimators.}

I use the two restrictions proposed by \RR{}. Relative Magnitudes (RM) bounds each post-treatment change in a cohort's block bias by a multiple of its largest pre-treatment change. Second Differences (SD) bounds the change in the slope of the block-bias path between consecutive periods. The block biases of \Imp{} and \CSNYT{} differ only by a normalization term, which cancels out under either restriction. Therefore, the same bound imposes an identical restriction on untreated potential outcomes for both estimators. Consequently, they share the same identified set for any common estimand and, when their covariances are estimated by the same rule, the same confidence set.

\RR{} proposes the conditional procedure (Conditional) and conditional least-favorable hybrid (C-LF) of \citet{andrews2023inference} to construct confidence sets under the RM and SD restrictions. It also proposes the fixed-length confidence interval (FLCI) of \citet{armstrong2018optimal} under the SD restriction. This machinery applies seamlessly to cohort inference. In the default setting of the \texttt{HonestDiD} package, there is a difference in the SD restriction across these procedures. Conditional and C-LF operate under an SD restriction bounding slope changes from treatment onward, leaving pure pre-treatment slope changes unrestricted. In contrast, the FLCI uses an SD restriction bounding slope changes along the full path. It uses these additional restrictions to balance worst-case bias against sampling variance by constructing an affine estimator that incorporates pre-treatment coefficients. I characterize the weights placed on pre-treatment coefficients for this affine estimator. For cohort inference, I show via simulation both the reduction in interval length when the full-path restriction holds and the loss of coverage when it fails.

The primary advantage of cohort inference is the transparency and interpretability of the identifying restrictions placed on cohort-level PT violations. When PT violations are heterogeneous across cohorts, aggregating event-study coefficients conflates these distinct violations. The resulting aggregate pre-treatment coefficient path may exhibit misleading dynamics that serve as a distorted benchmark for post-treatment violations, shifting the constructed confidence sets away from values admitted by the underlying cohort-level restrictions. Using calibrated simulations, I show that cohort inference maintains nominal coverage under restrictions on cohort-level PT violations. Under the same data-generating process, aggregate inference calibrated to yield an identified set of the same length undercovers severely, especially when cohort-level PT violations diverge. In the empirical application to \citet{callaway2021-did}, the confidence set constructed via cohort inference excludes zero under tight SD bounds, whereas aggregate inference covers zero even when the slope of the PT violation is held fixed.

While offering a clear interpretation, cohort inference relies on the full vector of cohort-period coefficients and their estimated covariance matrix, which leads to three costs. The first is precision. Because each cohort is benchmarked against its own pre-treatment path, noise in its coefficients is not averaged away. The cohort inference confidence sets therefore tend to be wider than aggregate inference sets whose identified set has the same width. The second is exposure to finite-sample bias in the estimated covariance, for example, when a cohort has few independent units or clusters. This exposure is similar between cohort and aggregate inference under Conditional/C-LF methods, but far greater for the FLCI under cohort inference, whose length as well as center depends on the estimated covariance. The third cost concerns the full-path SD restriction. Because cohort-level estimates typically exhibit larger fluctuations, the fitted cohort paths satisfy this restriction only at a larger bound than the aggregate path does. The confidence set is wide at such a bound, and the restriction is harder to assess. Calibrated simulations and the application to the community-health-center study of \citet{bailey2015publicmedicine} illustrate these costs.

The choice between the approaches depends first on whether and how PT violations differ across cohorts. Suppose PT violations are homogeneous across cohorts and each adopting cohort is small relative to the units still untreated, so that the change of comparison group matters little. Then the aggregate inference confidence set tends to have a center near that of cohort inference and, at a bound matched to the width of the cohort identified set, tends to be narrower. However, when PT violations across cohorts are heterogeneous, cohort inference is a more ``honest'' HonestDiD implementation. In this case, it is helpful to use cohorts' pre-treatment block biases as a diagnostic. Second, when the number of independent units or clusters in some cohorts is small, a bias-reduced covariance estimator is advisable. Third, the FLCI method with the full-path restriction should be used with caution, especially under cohort inference when some cohorts show large fluctuations in pre-trends. A small sensitivity bound may violate the full-path restriction, so that its coverage guarantee no longer applies, while a large sensitivity bound will inflate the confidence set. Lastly, restricting aggregate inference to an event-time window observed for every cohort eliminates changes in treated-cohort composition. This restriction can \textit{partially} improve the coverage of the aggregate inference confidence set by switching to another estimand, but it still suffers from undercoverage when cohort-level PT violations are heterogeneous.

My central contribution is to place the sensitivity restrictions of \RR{} on each cohort's block bias, its untreated trend relative to the units still untreated when it first adopts treatment. Building on the bias decompositions of \citet{arkhangelsky2024sequential} and \citet{aguilar2025comparative}, I map these restrictions one-to-one into restrictions on the estimation biases of \Imp{} and \CSNYT{}. Doing so anchors pre- and post-treatment comparisons to a fixed population, avoiding distortions from shifts in both treated-cohort composition and control groups. The resulting framework applies the inference procedures of \RR{} directly, providing an interpretable sensitivity analysis for staggered adoption without requiring new inference machinery. Moreover, under the RM and SD restrictions, it yields identical identified sets for the two estimators and, under a common covariance rule, identical confidence sets.

Section~\ref{sec:illustration} illustrates the composition changes. Section~\ref{sec:decomposition} defines block biases and gives the bias decompositions for \Imp{} and \CSNYT{}. Section~\ref{sec:inference} develops cohort inference and compares its restrictions with aggregate inference. Section~\ref{sec:flci} characterizes the FLCI under SD. Sections~\ref{sec:simulations} and~\ref{sec:applications} report the simulations and applications, and Section~\ref{sec:conclusion} concludes. Appendix~\ref{app:covariates} extends cohort inference to the covariate-adjusted \CSNYT{} estimators of \citet{callaway2021-did}.

\section{The minimum-wage application and two composition changes}
\label{sec:illustration}

In this section, I introduce the running example used throughout the paper. I use this example to illustrate how the composition of treated cohorts changes at different periods relative to treatment timing, and how the comparison group of not-yet-treated units shifts for a given treated cohort over time. Both dynamics threaten robust inference based on aggregate event-study coefficients.

The federal minimum wage remained fixed at \$5.15 between 2001 and 2007, while several states raised their own minimum wages above this threshold at different times. \citet{callaway2021-did} exploit this variation in treatment timing to estimate average treatment effects of minimum wage increases on teen employment. A county is defined as treated once its state's minimum wage exceeds \$5.15 and remains treated thereafter. The balanced panel contains 2,284 counties: 1,377 never-treated counties and three cohorts of 100, 223, and 584 counties first treated in 2004, 2006, and 2007, respectively. I index the years 2001 through 2007 by $t \in \{1, \dots, 7\}$ and label the three treated cohorts as $G_4$, $G_6$, and $G_7$, where the subscript $g \in \{4, 6, 7\}$ denotes the period of initial treatment. The treatment status is shown in Figure~\ref{fig:minwage}(a). I consider the case without covariates here and defer the discussion of covariates to Appendix~\ref{app:covariates}.

Let $\mathrm{ATT}(g,t)$ denote the average treatment effect for cohort $g$ in period $t$. I consider two estimands. The first is the average treatment effect on the treated, defined as a cohort-size-weighted average across all post-treatment $\mathrm{ATT}(g,t)$. The second is the dynamic treatment effect at relative period $s = t - g + 1$, which averages $\mathrm{ATT}(g,t)$ across cohorts at each $s$. 

I use the estimator from \citet{callaway2021-did} that uses not-yet-treated controls (\CSNYT{}) as a leading example. For cohort $g$ and year $t$, \CSNYT{} compares the cohort's outcome change from the base year $g-1$ against that of a comparison group. In post-treatment periods ($t \ge g$), the comparison group consists of never-treated counties alongside cohorts treated strictly after $t$. In pre-treatment periods ($t < g$), it includes all counties that remain untreated at the time cohort $g$ adopts the treatment. Figure~\ref{fig:minwage} (b) plots these cohort-specific estimates by relative period, with standard errors clustered at the county level. Pre-treatment estimates for $G_4$ and $G_6$ lie below their base-year benchmarks, whereas those for $G_7$ lie above. The aggregate event-study coefficients are a weighted average of these estimates at each relative period, exhibiting modest pre-treatment deviations followed by negative post-treatment estimates.

\begin{figure}[!h]
\centering
\subfigure[Treatment timing]{\includegraphics[width=0.48\textwidth]{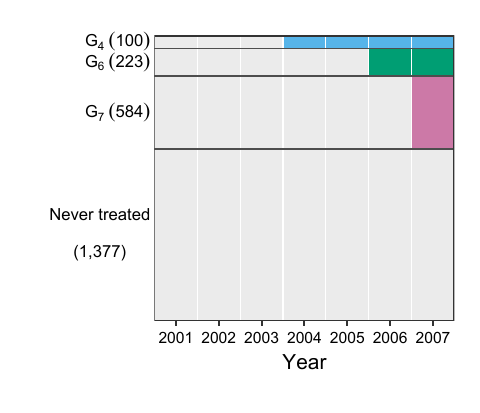}}
\subfigure[\CSNYT{} estimates by relative period]{\includegraphics[width=0.48\textwidth]{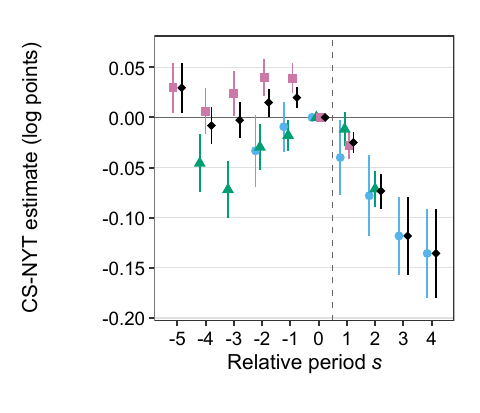}}\\
\subfigure[Cohort weights in the aggregate coefficients]{\includegraphics[width=0.48\textwidth]{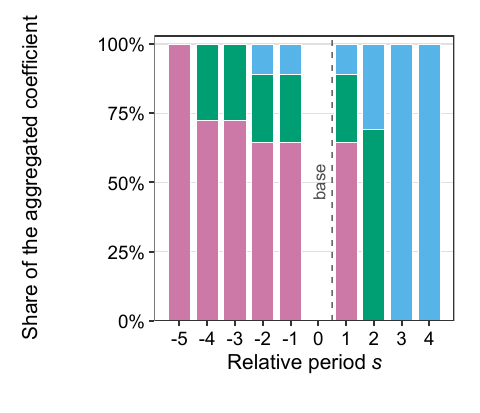}}
\subfigure[Comparison group of $G_4$]{\includegraphics[width=0.48\textwidth]{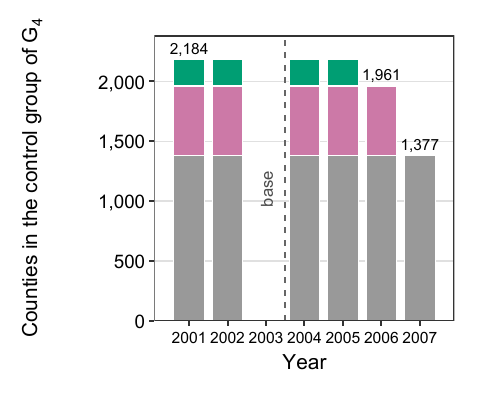}}\\
\includegraphics[width=0.9\textwidth]{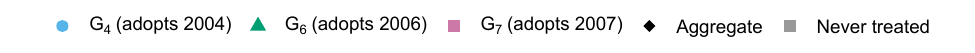}
\caption{The minimum-wage application: data, \CSNYT{} estimates, and the two composition changes}
\label{fig:minwage}
\vspace{0.5em}
\parbox{\linewidth}{\footnotesize
    \textit{Notes:} In (a), block heights are proportional to cohort size, and colored cells are treated. In (b), $s = 1$ is the adoption year and $s = 0$ the base year, and the intervals are 95 percent pointwise, clustered by county. In (d), counties are grouped by cohort. 
}
\end{figure}

Two compositional changes complicate the comparison between pre-treatment and post-treatment aggregate event-study coefficients. First, the set of contributing treated cohorts shifts over event time. Panel~(c) shows that $G_7$ receives 64 to 100 percent of the weight across all pre-treatment periods, yet drops out after $s = 1$. By contrast, the coefficients at $s = 3$ and $s = 4$ reflect $G_4$ alone. Consequently, an assessment of parallel trends based on the aggregate pre-treatment path primarily evaluates $G_7$, but the later post-treatment coefficients capture the effects for $G_4$.

Second, even for a single treated cohort, the comparison group shifts across relative periods. As panel~(d) shows, $G_4$ has 2,184 control counties through 2005. This number falls to 1,961 in 2006 once $G_6$ adopts the treatment, and further drops to only the 1,377 never-treated counties in 2007. Consequently, its pre-treatment coefficients and its final two post-treatment coefficients rely on distinct comparison groups.

The imputation estimator of \citet{borusyak2024revisiting} (\Imp{}) shares this issue, albeit implicitly, because later-treated cohorts drop out of the sample used to estimate the untreated potential outcomes. In the next section, I formalize the connection between \CSNYT{} and \Imp{}. In contrast, the estimator from \citet{callaway2021-did} that relies exclusively on never-treated units as controls (\CSNT{}) suffers only from the first compositional change.
\section{Block bias and the bias decomposition}
\label{sec:decomposition}

To address these two types of compositional change—in particular, the shifting control group for a given treated cohort—I introduce the concept of \emph{block bias}. A block bias compares a cohort's untreated trend with that of a fixed control group. Below, I formalize this concept and show that the post-treatment estimation biases of both \Imp{} and \CSNYT{} are linear mappings of these block biases. This decomposition provides the foundation for the cohort inference.

\subsection{Setup}
\label{sec:setup}

Suppose we observe a balanced panel of $N$ units over periods $t = 1, \dots, T$. Each unit belongs either to one of $G$ treated cohorts, denoted $\mathcal{G}_1, \dots, \mathcal{G}_G$, or to a never-treated cohort, denoted $\mathcal{G}_\infty$. Cohort $\mathcal{G}_g$ adopts treatment in period $t_g$ and remains treated thereafter. The corresponding cohort sizes are $N_g$ and $N_\infty$. For cohort $\mathcal{G}_g$, relative period $s$ corresponds to calendar period $t = t_g + s - 1$, such that $s = 1$ is the first post-treatment period and $s = 0$ is the last pre-treatment period. The set $\pre_g = \{1, \dots, t_g - 1\}$ collects its pre-treatment periods. Finally, we assume there are no always-treated units, there is an always-control group, there is no missing data, and cohort sizes are non-stochastic.

\begin{assumption}[Design]
\label{ass:design}
(i) $1 < t_1 < \cdots < t_G \leq T$; (ii) $N_\infty > 0$; (iii) $Y_{it}$ is observed for every unit and every period; (iv) the cohort sizes $(N_1, \dots, N_G, N_\infty)$ are fixed features of the design.
\end{assumption}

Let $Y_{it}(g)$ denote the potential outcome of unit $i$ in period $t$ if it adopts treatment in period $t_g$, and $Y_{it}(\infty)$ its potential outcome if it never adopts. The observed outcome is $Y_{it} = Y_{it}(g)$ for $i \in \mathcal G_g$ and $Y_{it} = Y_{it}(\infty)$ for $i \in \mathcal G_\infty$. We impose the standard Stable Unit Treatment Value Assumption (SUTVA) on the potential outcomes.

\begin{assumption}[SUTVA]
\label{ass:nointerference}
(i) No spillovers: the potential outcomes of unit $i$ do not depend on the adoption dates of other units; (ii) No anticipation: $Y_{it}(g) = Y_{it}(\infty)$ for all $i \in \mathcal G_g$ and $t < t_g$.
\end{assumption}

The treatment effect for cohort $\mathcal G_g$ in relative period $s \geq 1$ is defined as
\begin{equation}
\tau_{\mathcal G_g,s} = \E\bigl[Y_{it}(g) - Y_{it}(\infty) \mid i \in \mathcal G_g\bigr], \qquad t = t_g + s - 1.
\label{eq:att}
\end{equation}
For $s \leq 0$, this expression equals zero by no anticipation. Stacking the post-treatment cohort-period treatment effects into a vector $\vec\tau_{\post}$, we express the estimand of interest as a linear combination $\theta = \ell'\vec\tau_{\post}$. For example, if the weights in $\ell$ are proportional to cohort size and sum to one, $\theta$ is the average treatment effect on the treated. Other choices of $\ell$ can select effects at a specific relative period $s$ or restrict attention to a single cohort.

\subsection{Two estimators and estimation bias}
\label{sec:estimators}

I focus primarily on two popular ``modern DiD'' estimators: the imputation estimator of \citet{borusyak2024revisiting} (\Imp{}) and the \citet{callaway2021-did} estimator with not-yet-treated controls (\CSNYT{}). Both share the feature of relying on time-varying control groups for treated cohorts.

The imputation estimator \citep{borusyak2024revisiting} fits a two-way fixed effects model, $Y_{it} = \alpha_i + \phi_t + \varepsilon_{it}$, on the untreated observations only. Using the estimated unit and time fixed effects, it imputes the untreated potential outcome for each treated observation as $\hat Y_{it}(\infty) = \hat\alpha_i + \hat\phi_t$. The cohort-period estimate is then the average gap between observed and imputed outcomes:\footnote{Numerically equivalent estimators are documented in \citet{wooldridge2021two}, \citet{gardner2022two}, \citet{liu2024practical} and \citet{gardnerone}.}
\begin{equation}
\hat\tau^{\Imp}_{\mathcal G_g,s} = \frac{1}{N_g}\sum_{i \in \mathcal G_g}\bigl[Y_{it} - \hat Y_{it}(\infty)\bigr], \qquad t = t_g + s - 1.
\label{eq:imp}
\end{equation}

The \CSNYT{} estimator of \citet{callaway2021-did} compares long differences from the base period $t_g - 1$. Its comparison group in relative period $s \geq 1$ is the set of units not yet treated at $t = t_g + s - 1$, given by $\mathcal C_{g,s} = (\bigcup_{k : t_k > t}\mathcal G_k) \cup \mathcal G_\infty$, which shrinks as later cohorts adopt treatment:
\begin{equation}
\hat\tau^{\CSNYT}_{\mathcal G_g,s} = \frac{1}{N_g}\sum_{i \in \mathcal G_g}\bigl(Y_{it} - Y_{i,t_g-1}\bigr) - \frac{1}{N_{\mathcal C_{g,s}}}\sum_{i \in \mathcal C_{g,s}}\bigl(Y_{it} - Y_{i,t_g-1}\bigr), \qquad t = t_g + s - 1.
\label{eq:csnyt}
\end{equation}
In the first post-treatment period, the comparison group is $\mathcal C_{g,1} = (\bigcup_{k : t_k > t_g}\mathcal G_k) \cup \mathcal G_\infty$. As shown later, this initial comparison group for cohort $\mathcal G_g$ serves as the benchmark control group in my definition of block bias.\footnote{The never-treated variant (\CSNT{}) replaces $\mathcal C_{g,s}$ with $\mathcal G_\infty$ in every period. Without covariates, it coincides with the estimator of \citet{sun2021-event}.}

The estimation bias of a post-treatment coefficient is $\delta_{\mathcal G_g,s} = \E[\hat\tau_{\mathcal G_g,s}] - \tau_{\mathcal G_g,s}$ for $s \geq 1$. For the imputation estimator it is the gap between the true and the imputed untreated potential outcomes of the cohort, and for \CSNYT{} the difference in untreated trends between the cohort and its not-yet-treated control group $\mathcal C_{g,s}$:
\begin{align}
\delta^{\Imp}_{\mathcal G_g,s} &= \E\bigl[Y_{it}(\infty) \mid i \in \mathcal G_g\bigr] - \E\bigl[\hat Y_{it}(\infty) \mid i \in \mathcal G_g\bigr],
\label{eq:delta-imp}\\
\delta^{\CSNYT}_{\mathcal G_g,s} &= \E\bigl[Y_{it}(\infty) - Y_{i,t_g-1}(\infty) \mid i \in \mathcal G_g\bigr] - \E\bigl[Y_{it}(\infty) - Y_{i,t_g-1}(\infty) \mid i \in \mathcal C_{g,s}\bigr],
\label{eq:delta-csnyt}
\end{align}
with $t = t_g + s - 1$. 

For \CSNYT{}, $\delta^{\CSNYT}_{\mathcal G_g,s} = 0$ is the parallel trend assumption in its familiar form. It states that, between $t_g - 1$ and $t$, untreated outcomes change on average by the same amount in the cohort and among the units not yet treated in period $t$. For the imputation estimator, \citet{borusyak2024revisiting} assume that the two-way model holds for the untreated potential outcome in every period, $Y_{it}(\infty) = \alpha_i + \phi_t + \varepsilon_{it}$ with $\E[\varepsilon_{it} \mid i \in \mathcal G_g] = 0$ for every cohort and every $t$. The imputation is then unbiased, so $\delta^{\Imp}_{\mathcal G_g,s} = 0$.

\subsection{Block bias}
\label{sec:blockbias}

Sensitivity analysis needs to bound estimation bias $\delta$, but that bias is difficult to benchmark directly with pre-treatment data. For \CSNYT{}, the comparison group in \eqref{eq:delta-csnyt} changes with $s$, so one pre-treatment comparison cannot match every post-treatment bias. For imputation, the pre-treatment residuals are in-sample fits of the counterfactual model and are not valid benchmarks \citep{li2025benchmarking}. Block bias provides a comparison that remains fixed over time.

The block bias of cohort $\mathcal G_g$ in relative period $s$ is its untreated trend minus the untreated trend of its initial control group $\mathcal C_{g,1}$. Both trends are in untreated potential outcomes and are measured from the estimator's own reference period(s), the last pre-treatment period for \CSNYT{} and the average over $\pre_g$ for the imputation estimator:
\begin{align}
\Delta^{\CSNYT}_{\mathcal G_g,s} &= \E\bigl[Y_{it}(\infty) - Y_{i,t_g-1}(\infty) \mid i \in \mathcal G_g\bigr] - \E\bigl[Y_{it}(\infty) - Y_{i,t_g-1}(\infty) \mid i \in \mathcal C_{g,1}\bigr],
\label{eq:block-csnyt}\\
\Delta^{\Imp}_{\mathcal G_g,s} &= \E\bigl[Y_{it}(\infty) - \bar Y_{i,\pre_g}(\infty) \mid i \in \mathcal G_g\bigr] - \E\bigl[Y_{it}(\infty) - \bar Y_{i,\pre_g}(\infty) \mid i \in \mathcal C_{g,1}\bigr],
\label{eq:block-imp}
\end{align}
where $\bar Y_{i,\pre_g}(\infty) = \frac{1}{(t_g - 1)}\sum_{r \in \pre_g} Y_{ir}(\infty)$. The comparison group is the same in every period, and the members of $\mathcal C_{g,1}$ that adopt later still enter the definition through $Y_{it}(\infty)$. For $s \leq 0$ every outcome in \eqref{eq:block-csnyt} and \eqref{eq:block-imp} is observed, because neither the cohort nor its initial control group is treated before $t_g$. Pre-treatment block biases are therefore estimable. For $s \geq 1$ they are not, as we do not observe $Y_{it}(\infty)$ for the treated cohort and possibly some cohorts in $\mathcal C_{g,1}$. 

In the first treated period the two biases coincide, $\delta_{\mathcal G_g,1} = \Delta_{\mathcal G_g,1}$, for both estimators. For \CSNYT{} this is immediate, since $\mathcal C_{g,1}$ is the comparison group at $s = 1$. For the imputation estimator it follows from an identity of the untreated two-way fixed effects fit.

\begin{lemma}
\label{lemma:first-period}
For $i \in \mathcal G_g$, the imputed counterfactual in the first treated period is a block difference-in-differences, $\hat Y_{i,t_g}(\infty) = \bar Y_{i,\pre_g} + \bigl(\bar Y_{\mathcal C_{g,1},t_g} - \bar Y_{\mathcal C_{g,1},\pre_g}\bigr)$, where $\bar Y_{\mathcal C_{g,1},t_g}$ and $\bar Y_{\mathcal C_{g,1},\pre_g}$ are the initial control group's average outcome in period $t_g$ and over $\pre_g$.
\end{lemma}

Taking expectations gives $\delta^{\Imp}_{\mathcal G_g,1} = \Delta^{\Imp}_{\mathcal G_g,1}$. The proof (Appendix~\ref{app:lemma-first-period}) uses only the residual identities of the fixed effects fit.

\subsection{The bias decomposition}
\label{sec:biasdecomposition}

After the first treated period, the estimation bias of an earlier-treated cohort can depend on the block biases of the cohorts that adopt later. Let $\mathcal K_g(t) = \{k : t_g < t_k \leq t\}$ denote these later-treated cohorts at calendar period $t$. I term them \textit{adjustment cohorts}. Let $s_k(t) = t - t_k + 1$ be the relative period of cohort $\mathcal G_k$ at $t$, and let
\begin{equation}
w_k = \frac{N_k}{\sum_{j \geq k} N_j + N_\infty} = \frac{N_k}{N_k + N_{\mathcal C_{k,1}}}
\label{eq:weights}
\end{equation}
be the share of cohort $\mathcal G_k$ in the block formed by itself and its initial control group.

\begin{proposition}[Imputation estimator]
\label{prop:imputation}
Under Assumptions~\ref{ass:design} and~\ref{ass:nointerference}, for every cohort $\mathcal G_g$ and $s \geq 1$ with $t = t_g + s - 1$,
\begin{equation}
\delta^{\Imp}_{\mathcal G_g,s} = \Delta^{\Imp}_{\mathcal G_g,s} + \sum_{k \in \mathcal K_g(t)} w_k\, \Delta^{\Imp}_{\mathcal G_k, s_k(t)} .
\label{eq:decomp-imp}
\end{equation}
\end{proposition}

\begin{proposition}[\CSNYT{} estimator]
\label{prop:csnyt}
Under the same assumptions, for every cohort $\mathcal G_g$ and $s \geq 1$ with $t = t_g + s - 1$,
\begin{equation}
\delta^{\CSNYT}_{\mathcal G_g,s} = \Delta^{\CSNYT}_{\mathcal G_g,s} + \sum_{k \in \mathcal K_g(t)} w_k \bigl[\Delta^{\CSNYT}_{\mathcal G_k, s_k(t)} - \Delta^{\CSNYT}_{\mathcal G_k, s_k(t_g - 1)}\bigr].
\label{eq:decomp-csnyt}
\end{equation}
\end{proposition}

In both decompositions, the bias for cohort $\mathcal G_g$ in period $t$ equals its own block bias plus adjustments for cohorts that adopted between $t_g$ and $t$. Each adjustment uses that cohort's block bias at calendar time $t$, and its weight increases with the cohort's size relative to its initial control group.

The reference periods account for the difference between the formulas. Imputation measures each block bias from the cohort's own pre-treatment average. \CSNYT{} measures cohort $\mathcal G_g$ from $t_g - 1$ and cohort $\mathcal G_k$ from $t_k - 1$. It therefore subtracts $\Delta^{\CSNYT}_{\mathcal G_k, s_k(t_g - 1)}$ to express the adjustment relative to the target cohort's base period. Since $s_k(t_g - 1) = t_g - t_k < 0$, this correction is an observable pre-treatment block bias of the adjustment cohort.

Propositions~\ref{prop:imputation} and~\ref{prop:csnyt} restate Corollaries~1 and~3 of \citet{aguilar2025comparative} in the block-bias notation and under the reference-period conventions used here. The imputation case also follows from the sequential representation of \citet{arkhangelsky2024sequential}. Appendix~\ref{app:correspondence} records the notation correspondence and the telescoping argument that connects the two statements.

I extend the notation to pre-treatment cells by setting $\delta_{\mathcal G_g,s} = \Delta_{\mathcal G_g,s}$ for $s \leq 0$. I stack both objects with the pre-treatment cells first, ordered within each subvector by calendar period and then by adoption period. The decompositions then read
\begin{equation}
\vec\delta = \mathbf W\vec\Delta .
\label{eq:W}
\end{equation}
After regrouping cells by calendar period and then by adoption period within each period, each diagonal calendar block of $\mathbf W$ is unit upper triangular. An adjustment cohort always adopts after the cohort it adjusts. Imputation uses only current-period cells, so $\mathbf W$ is block diagonal across calendar periods. The pre-treatment corrections for \CSNYT{} enter earlier calendar blocks, making its $\mathbf W$ block lower triangular. In both cases $\det\mathbf W = 1$. A restriction on $\vec\Delta$ therefore maps one to one into a restriction on $\vec\delta$.

Figure~\ref{fig:w_rows} shows the rows of $\mathbf W$ with adjustment terms in the minimum-wage design, where cohorts are labeled by their adoption period as in Section~\ref{sec:illustration}. Each row writes the estimation bias of one post-treatment cell as a combination of block biases, with one tile per nonzero coefficient. For example, the bias of $\mathcal G_4$ in 2007 ($s = 4$) is its own block bias in 2007 plus $w_6$ and $w_7$ times the 2007 block biases of $\mathcal G_6$ and $\mathcal G_7$. The weights are $w_6 = 223/2{,}184 = 0.102$ and $w_7 = 584/1{,}961 = 0.298$. The imputation rows in panel (a) have entries only in the current calendar block. The \CSNYT{} rows in panel (b) also carry $-w_6$ and $-w_7$ on the same adjustment cohorts' cells in the target cohort's base year. Every other row of $\mathbf W$ is a row of the identity.

\begin{figure}[!h]
\centering
\subfigure[Imputation]{\includegraphics[width=\textwidth]{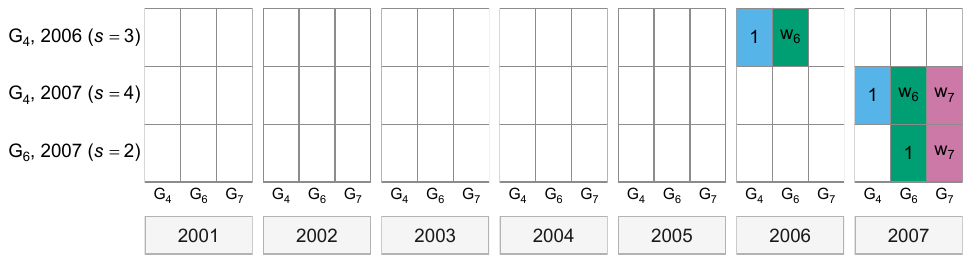}}\\[0.4em]
\subfigure[\CSNYT{}]{\includegraphics[width=\textwidth]{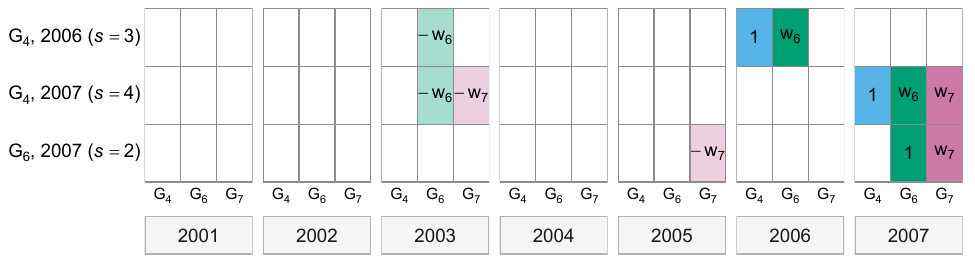}}
\caption{The rows of the bias map $\mathbf W$ with adjustment terms in the minimum-wage design}
\label{fig:w_rows}
\vspace{0.5em}
\parbox{\linewidth}{\footnotesize
    \textit{Notes:} Each panel shows the three rows of $\vec\delta = \mathbf W\vec\Delta$ with adjustment terms: the estimation biases of $G_4$ in 2006 ($s=3$), $G_4$ in 2007 ($s=4$), and $G_6$ in 2007 ($s=2$). Columns are the block-bias cells, ordered by calendar year and, within a year, by adoption year. A tile shows the coefficient on the block bias of that cell, colored by the cell's cohort; lighter tiles are negative, and blank tiles are zeros. The adjustment weights are $w_6 = 0.102$ and $w_7 = 0.298$. In (b), the adjustment terms also subtract each adjustment cohort's block bias in the target cohort's base year (2003 for $G_4$, 2005 for $G_6$).
}
\end{figure}

\subsection{The cohort coefficient vector}
\label{sec:coefvector}

The pre-treatment block biases are estimated by their sample analogues,
\begin{align}
\hat\Delta^{\CSNYT}_{\mathcal G_g,s} &= \frac{1}{N_g}\sum_{i \in \mathcal G_g}\bigl(Y_{it} - Y_{i,t_g-1}\bigr) - \frac{1}{N_{\mathcal C_{g,1}}}\sum_{i \in \mathcal C_{g,1}}\bigl(Y_{it} - Y_{i,t_g-1}\bigr),
\label{eq:pre-csnyt}\\
\hat\Delta^{\Imp}_{\mathcal G_g,s} &= \frac{1}{N_g}\sum_{i \in \mathcal G_g}\bigl(Y_{it} - \bar Y_{i,\pre_g}\bigr) - \frac{1}{N_{\mathcal C_{g,1}}}\sum_{i \in \mathcal C_{g,1}}\bigl(Y_{it} - \bar Y_{i,\pre_g}\bigr),
\label{eq:pre-imp}
\end{align}
for $s \leq 0$ and $t = t_g + s - 1$. For \CSNYT{} these are the pre-treatment coefficients that the \texttt{did} package reports with the universal base period, the coefficients of Figure~\ref{fig:minwage}(b). For the imputation estimator, it corresponds to the pre-treatment estimation that \citet{Roth2024interpret} calls ``symmetric'' under non-staggered adoptions.

Both estimators deliver one coefficient per cohort-period cell. I use $\hat\beta$ to denote these cohort-period estimates and $\beta$ for their population counterparts. For $s \geq 1$, $\hat\beta_{\mathcal G_g,s}$ is the estimate in \eqref{eq:imp} or \eqref{eq:csnyt}. For $s \leq 0$, it is the estimated pre-treatment block bias in \eqref{eq:pre-csnyt} or \eqref{eq:pre-imp}, a placebo comparison for the corresponding estimator. I stack these coefficients into $\vec{\hat\beta} = (\vec{\hat\beta}_{\pre}', \vec{\hat\beta}_{\post}')'$ in the cell order of Section~\ref{sec:biasdecomposition}. Its population counterpart is $\vec\beta$: block biases in pre-treatment cells and $\tau_{\mathcal G_g,s} + \delta_{\mathcal G_g,s}$ in post-treatment cells, with $\delta$ the estimation bias of Section~\ref{sec:estimators}.

Each cohort carries one normalization in its block bias. The \CSNYT{} coefficient at $s = 0$ is normalized to zero, whereas a cohort's pre-treatment imputation coefficients sum to zero because they are referenced against its pre-treatment mean. These normalizations render the full covariance matrix of $\vec{\hat\beta}$ singular, though the inference procedures introduced later are robust to this singularity. I assume that $\vec{\hat{\beta}}$ is asymptotically normal and the covariance matrix can be estimated consistently. 

\begin{assumption}[Sampling]\label{ass:normality}
With cohort shares fixed as $N \to \infty$, $\sqrt{N}(\vec{\hat\beta} - \vec\beta) \stackrel{d}{\to} \mathcal N(0, \Sigma^*)$, and an estimator $\hat\Sigma_N$ of $\Sigma_N = \Sigma^*/N$ satisfies $N\hat\Sigma_N \stackrel{p}{\to} \Sigma^*$.
\end{assumption}

Note that we can estimate $\Sigma_N$ for the joint vector $\vec{\hat\beta}$ under a common rule for both the \Imp{} and \CSNYT{} estimators, as they are linear in the stacked $NT \times 1$ outcome vector $\vec Y$. For each estimator, we can write $\vec{\hat\beta} = \mathbf L\vec Y$. The estimator-specific weight matrix $\mathbf L$ is non-stochastic under Assumption~\ref{ass:design} and assigns identical weights to all units within the same cohort-period cell $(\mathcal G_g, t)$. $\vec{\hat\beta}$ is thus a linear function of the sample cell means (Appendix~\ref{app:covariance}). I center outcomes at their cohort-period cell means, $\hat e_{it} = Y_{it} - \bar Y_{\mathcal G_{g(i)}, t}$, where $g(i)$ denotes unit $i$'s cohort. Let $\vec{\hat s}_i = \mathbf L_i\vec{\hat e}_i$ denote unit $i$'s score formed from its columns in $\mathbf L$ and its $T \times 1$ residual vector $\vec{\hat e}_i$. I estimate $\Sigma_N$ by the unit-clustered sandwich estimator $\hat\Sigma_N = a_N\sum_{i=1}^N \vec{\hat s}_i\vec{\hat s}_i'$, with finite-sample adjustment $a_N = \frac{N}{N-1}\frac{NT-1}{NT-p}$ for $p$ cohort-period cells. Appendix~\ref{app:covariance} also details the extension to coarser clusters.

\section{Cohort inference}
\label{sec:inference}

I introduce cohort inference in this section. The restriction is placed on the block biases defined in Section~\ref{sec:decomposition} and mapped to estimation bias through $\mathbf W$. I first define the identified set and the RM and SD restrictions, then describe inference and establish equivalence across estimators. The final subsection defines aggregate inference and compares the two approaches. I suppress estimator superscripts except when comparing the two estimators.

\subsection{Model and identified set}
\label{sec:model}

Following Assumption~\ref{ass:normality} and the decomposition of Section~\ref{sec:decomposition}, the estimated cohort-period coefficients follow
\begin{equation}
\vec{\hat\beta} \sim N\bigl(\vec\tau + \mathbf W\vec\Delta,\ \Sigma_N\bigr), \qquad \vec\tau_{\pre} = 0,
\label{eq:model}
\end{equation}
which is the model of \RR{} with the block-bias vector $\vec\Delta$ as the nuisance parameter. Finite-sample guarantees are stated for this model with known $\Sigma_N$. In practice $\hat\Sigma_N$ replaces $\Sigma_N$ under the asymptotic conditions stated in Section~\ref{sec:confidence}.

A restriction is a set $\Lambda_\Delta$ of admissible block-bias vectors. I require it to be either a polyhedron $\{\vec\Delta : \mathbf P\vec\Delta \leq \vec d\}$ for a known matrix $\mathbf P$ and vector $\vec d$, or a finite union of such polyhedra. Given the restriction, the identified set for $\theta$ collects the values of the target compatible with the population coefficients,
\begin{equation}
\mathcal S(\vec\beta, \Lambda_\Delta) = \Bigl\{\ell'\bigl(\vec\beta_{\post} - (\mathbf W\vec\Delta)_{\post}\bigr) : \vec\Delta \in \Lambda_\Delta,\ \vec\Delta_{\pre} = \vec\beta_{\pre}\Bigr\}.
\label{eq:identified-set}
\end{equation}
The pre-treatment block biases are pinned down by the pre-treatment coefficients, the post-treatment block biases range over what the restriction allows, and each admissible vector removes its implied estimation bias $(\mathbf W\vec\Delta)_{\post}$ from the post-treatment coefficients. Because $\mathbf W$ is invertible, a restriction on block biases is also a restriction on estimation biases. The map $\Lambda_\delta = \mathbf W\Lambda_\Delta$ turns the polyhedron $\{\mathbf P\vec\Delta \leq \vec d\}$ into the polyhedron $\{\mathbf P\mathbf W^{-1}\vec\delta \leq \vec d\}$ and a union into a union. The set $\mathcal S(\vec\beta, \Lambda_\Delta)$ thus equals $\mathcal S(\vec\beta, \Lambda_\delta)=\Bigl\{\ell'\bigl(\vec\beta_{\post} - \vec\delta_{\post}\bigr) : \vec\delta \in \Lambda_\delta,\ \vec\delta_{\pre} = \vec\beta_{\pre}\Bigr\}$, the expression in \RR{}. 

The population identified set \eqref{eq:identified-set} contains the values of the target that remain compatible with the restriction even if the population coefficients are known, so it describes identification uncertainty. Its plug-in version $\mathcal S(\vec{\hat\beta}, \Lambda_\Delta)$ replaces $\vec\beta$ by $\vec{\hat\beta}$. The confidence sets in Section~\ref{sec:confidence} also account for statistical uncertainty. I report the plug-in identified sets and the confidence sets in the applications.

\subsection{Restrictions on block biases}
\label{sec:restrictions}

The two restrictions considered in \RR{} apply here, assuming at least two pre-treatment observations per cohort ($t_1 \geq 3$). The Relative Magnitudes (RM) restriction bounds the change in a cohort's block bias from the preceding period at every post-treatment period, including the treatment boundary. The bound is a multiple $\bar M$ of the largest change between consecutive periods of that cohort's own pre-treatment path, 
\begin{equation}
\Lambda^{\mathrm{RM}}_\Delta(\bar M) = \Bigl\{\vec\Delta : |\Delta_{\mathcal G_g,s} - \Delta_{\mathcal G_g,s-1}| \leq \bar M \max_{3 - t_g \leq r \leq 0}|\Delta_{\mathcal G_g,r} - \Delta_{\mathcal G_g,r-1}| \ \text{for all } g \text{ and } s \geq 1\Bigr\}.
\label{eq:rm}
\end{equation}
The maximum on the right-hand side is the cohort-specific benchmark. At $\bar M = 0$ the block bias is constant after treatment and equal to its value in the last pre-treatment period. For \CSNYT{} that value is zero; for the imputation estimator the value is $\Delta_{\mathcal G_g,0}$.

The post-only version of the Second Differences (SD) restriction bounds the change in the slope of a cohort's block-bias path from one period to the next,
\begin{equation}
\Lambda^{\mathrm{SD}}_\Delta(M) = \Bigl\{\vec\Delta : |\nabla^2\Delta_{\mathcal G_g,s}| \leq M \ \text{for all } g \text{ and } s \geq 1\Bigr\}, \qquad \nabla^2\Delta_{\mathcal G_g,s} = \Delta_{\mathcal G_g,s} - 2\Delta_{\mathcal G_g,s-1} + \Delta_{\mathcal G_g,s-2}.
\label{eq:sd}
\end{equation}
The bound applies to every second difference that involves a post-treatment cell, from the treatment boundary onward. Second differences that lie entirely inside the pre-treatment window are unrestricted. At $M = 0$ each cohort's block bias continues the slope of its last two pre-treatment periods. A positive $M$ allows the slope to change by up to $M$ per period. Cohort-specific bounds $M_g$ are also allowed. Section~\ref{sec:flci} introduces the full-path version of \eqref{eq:sd}, which also bounds the second differences inside the pre-treatment window and is the version the fixed-length confidence interval uses. 

Both restrictions act on differences of block biases along a cohort's path, and differences do not depend on the reference period. For either estimator,
$\Delta_{\mathcal G_g,s} - \Delta_{\mathcal G_g,s-1} = \E[Y_{it}(\infty) - Y_{i,t-1}(\infty) \mid i \in \mathcal G_g] - \E[Y_{it}(\infty) - Y_{i,t-1}(\infty) \mid i \in \mathcal C_{g,1}]$ with $t = t_g + s - 1$. So, for the imputation estimator and \CSNYT{} alike, \eqref{eq:rm} restricts the gap between a cohort and its initial control group in one-period changes of untreated potential outcomes, and \eqref{eq:sd} restricts the change in that gap over time. 

\subsection{Confidence sets}
\label{sec:confidence}

Following \RR, a confidence set $\mathcal C_N(\vec{\hat\beta}, \Sigma_N)$ is valid if, in the model \eqref{eq:model}, it covers each admissible target value with probability at least $1 - \alpha$, uniformly over the identified set and the admissible block-bias vectors,
\begin{equation}
\inf_{\vec\Delta \in \Lambda_\Delta,\ \vec\tau}\ \inf_{\theta \in \mathcal S(\vec\tau + \mathbf W\vec\Delta, \Lambda_\Delta)} \mathbb{P}\bigl(\theta \in \mathcal C_N(\vec{\hat\beta}, \Sigma_N)\bigr) \geq 1 - \alpha.
\label{eq:coverage}
\end{equation}
Under the uniform normal approximation, covariance consistency, and regularity conditions of \RR{} (their Section~3.3), the guarantee extends to uniform asymptotic coverage with $\hat\Sigma_N$ in place of $\Sigma_N$.

For a polyhedral restriction, the confidence set inverts a test of each candidate value $\theta_0$ against the mapped restriction on $\vec\delta$. The conditional test of \citet{andrews2023inference} gives the Conditional confidence set, and its least-favorable hybrid gives the C-LF confidence set. The hybrid first applies a size-$\kappa$ least-favorable test with $\kappa = \alpha/10$ and then the conditional test at level $(\alpha - \kappa)/(1 - \kappa)$. Throughout, a confidence set named by a method is the confidence set that method constructs. The test statistics, the conditioning event, and the discussion of power properties are in \RR{} (their Section~3) and \citet{andrews2023inference}. I use C-LF as the default confidence set under the RM and post-only SD restrictions throughout the paper. Figure~\ref{fig:app-cal-length}(a) in Appendix~\ref{app:flci-designs} compares it with the Conditional method; the difference is tiny.

The RM restriction is a union of polyhedra, one for each combination of candidates $|\Delta_{\mathcal G_g,r} - \Delta_{\mathcal G_g,r-1}|$ across cohorts that attains the maximum in \eqref{eq:rm}. Because the data do not reveal which change is the largest in population, the confidence set, following Lemma~2.2 of \RR, is the union of the confidence sets computed for each combination. The union is valid but conservative when many components are far from binding, and it costs one inversion per component. Therefore, for the RM restriction with multiple cohorts, the confidence set can also be computationally intensive.

\subsection{Estimator equivalence}
\label{sec:equivalence}

The block biases of the two estimators differ only by a cohort-specific normalization term (Section~\ref{sec:blockbias}), and the two restrictions act on differences and apply to the post-treatment untreated counterfactuals. The identified set is therefore the same for both estimators.

\begin{proposition}[Identified-set equivalence]
\label{prop:equivalence}
Under Assumptions~\ref{ass:design} and~\ref{ass:nointerference}, for any target weights $\ell$ and any $\bar M \geq 0$ (respectively $M \geq 0$), the identified set \eqref{eq:identified-set} under $\Lambda^{\mathrm{RM}}_\Delta(\bar M)$ (respectively $\Lambda^{\mathrm{SD}}_\Delta(M)$) is the same whether $\vec\beta$, $\mathbf W$, and $\vec\Delta$ are those of the imputation estimator or those of \CSNYT.
\end{proposition}

The proof (Appendix~\ref{app:equivalence}) rests on two observations. First, a block-bias path satisfies the restriction and matches the imputation estimator's pre-treatment coefficients exactly when the path shifted by the cohort constants does both for \CSNYT{}. Second, along that correspondence the gap between the two estimators' population post-treatment coefficients equals the gap between their mapped biases. The corrected coefficients $\vec\beta_{\post} - (\mathbf W\vec\Delta)_{\post}$ therefore coincide, and so does every $\ell'$ combination of them.

The same argument goes through with sample cell means in place of population means (Appendix~\ref{app:equivalence}), so the plug-in identified sets coincide. The covariance estimator of Section~\ref{sec:coefvector} is built from the estimator's linear operator and from the residuals of the cohort-period cell means, with one finite-sample factor. Neither the residuals nor the factor depends on the estimator. With this common rule, the equivalence extends to the confidence sets of Section~\ref{sec:confidence}.

\begin{proposition}[Confidence-set equivalence]
\label{prop:ci-equivalence}
Let $\hat\Sigma_N$ be the estimator of Section~\ref{sec:coefvector} for both estimators. Consider the Conditional and C-LF confidence sets of Section~\ref{sec:confidence} under $\Lambda^{\mathrm{RM}}_\Delta(\bar M)$ or $\Lambda^{\mathrm{SD}}_\Delta(M)$. For every sample, every target weights $\ell$, and every bound, each of these confidence sets is the same whether the inputs $(\vec{\hat\beta}, \hat\Sigma_N, \mathbf W)$ are those of the imputation estimator or those of \CSNYT.
\end{proposition}

Appendix~\ref{app:ci-equivalence} proves the result.\footnote{It also proves the proposition for the FLCI confidence sets discussed in Section~\ref{sec:flci}, with a common rule for selecting among multiple solutions of the FLCI program.} With the same covariance rule, both estimators give the same sensitivity analysis under cohort inference. Equivalence requires restrictions on differences. A restriction on block-bias levels, such as a sign restriction, depends on the reference period, so equivalence need not hold under it. Equivalence also does not extend to aggregate inference.

\subsection{Aggregate inference}
\label{sec:aggregate}
\label{sec:aggregate-inference}

Aggregate inference applies the procedures of Section~\ref{sec:confidence} to the aggregate event-study coefficients of Section~\ref{sec:illustration}. Let the matrix $\mathbf A$ average the cells $(\mathcal G_g, s)$ observed at each relative period $s$, weighting cohorts by their shares of the units observed there. The aggregate coefficients are $\vec{\hat\beta}_A = \mathbf A\vec{\hat\beta}$, with covariance $\Sigma_A = \mathbf A\Sigma_N\mathbf A'$, and the aggregate effects by relative period are $\vec\tau_A = \mathbf A\vec\tau$. For the default target, $\ell_A$ weights post-treatment relative periods by their shares of treated unit-periods and puts zero weight before treatment, so that $\mathbf A'\ell_A = (\vec 0', \ell')'$ and $\theta = \ell_A'\vec\tau_A$. Cohort and aggregate inference then share the point estimate $\ell'\vec{\hat\beta}_{\post}$ and its standard error. In population,
\begin{equation}
\vec\beta_A = \mathbf A\vec\tau + \mathbf A\mathbf W\vec\Delta ,
\label{eq:aggregate-path}
\end{equation}
so the aggregate coefficients carry the aggregate bias $\mathbf A\mathbf W\vec\Delta$, whose pre-treatment part the pre-treatment coefficients identify. I use the post-only SD restriction as an example to show the difference between cohort and aggregate inference.

Aggregate inference places the restriction on $\vec\delta_A =\mathbf A\mathbf W\vec\Delta$, with $(\vec{\hat\beta}_A, \hat\Sigma_A, \ell_A)$ as the inputs. Under post-only SD the restriction bounds the second differences of the aggregate bias path by a sensitivity parameter $M_A$, and the identified set collects the values of the target compatible with the restriction,
\begin{equation}
\Lambda^{\mathrm{SD}}_A(M_A) = \{\vec\delta_A : |\nabla^2\delta_{A,s}| \leq M_A,\ s \geq 1\}, \quad \mathcal S_A = \{\ell_A'(\vec\beta_A - \vec\delta_A)_{\post} : \vec\delta_A \in \Lambda^{\mathrm{SD}}_A(M_A),\ \vec\delta_{A,\pre} = \vec\beta_{A,\pre}\}.
\label{eq:aggregate-sd}
\end{equation}
This is the form of \eqref{eq:identified-set} with the aggregate path in place of the cohort paths. The post-treatment part of the aggregate path is anchored by its own pre-treatment part rather than by the cohort paths. RM is defined on the aggregate path in the same way. Therefore, cohort and aggregate inference restrict different objects. At the same bound, $M_A = M$, their identified sets in general differ in location and width. The reason is the two composition changes discussed in Section~\ref{sec:illustration}. 

The composition change in treated cohorts corresponds to $\mathbf A$. Take $\mathbf W = \mathbf I$, which corresponds to an estimator whose comparison group is always the never-treated cohort, such as \CSNT{}. The aggregate bias $\mathbf A\vec\Delta$ then averages the cohort-level block biases at each relative period. When the composition of treated cohorts does not change with the relative period $s$, the second differences of the aggregate bias are just averages of the second differences of the cohort-level block biases on the same group of treated cohorts. The cohort-level restriction thus implies the aggregate-level restriction at the same bound, and the two identified sets coincide. The reason is that the cohort-level block biases that attain the endpoints (worst-case) of the cohort identified set also reach the endpoints (worst-case) of the aggregate identified set. However, when cohorts enter or leave the average as $s$ changes, the average can switch between block biases with different slopes or levels. The aggregate path can therefore show curvature even if every cohort-level block bias is linear. The cohort inference identified set then differs from the aggregate inference identified set. The exception is homogeneity of block bias in relative period, $\Delta_{\mathcal G_g,s} = \Delta_s$ for all $g$, under which the aggregate bias path equals the common cohort-level block bias path whatever the composition.\footnote{For the imputation estimator, the common block bias path must further have zero mean over every cohort's pre-treatment window.}

The composition change in the not-yet-treated control group works through $\mathbf W$ and arises under imputation and \CSNYT{} estimators. Consider the relative periods in which the composition of treated cohorts does not change and the cohort-level block biases are linear with homogeneous slopes. The later-treated cohorts' block biases can still enter the estimation bias of the earlier-treated cohorts through the adjustment term in Propositions \ref{prop:imputation} and \ref{prop:csnyt}, and enter the aggregate bias eventually. Consequently, even if we restrict the cohort-level block biases to have zero slope change and fix the composition of treated cohorts, the aggregate bias path can still show curvature. The only exception is when the cohorts' block biases are zero, as under parallel trends.

Aggregate inference therefore restricts an induced bias path whose curvatures need not correspond to curvatures in any cohort's own block bias path. At the same bound $M_A = M$, the identified set under aggregate inference generally differs from the one under cohort inference and need not contain it. I define the covering bound $M_A^C$ as the smallest $M_A$ at which, given the pre-treatment block biases, the aggregate restriction \eqref{eq:aggregate-sd} accommodates the cohort restriction with bound $M$, that is, admits every bias path that the cohort restriction admits. For every $M_A \geq M_A^C$ the identified set under aggregate inference then contains the one under cohort inference for every target formed from the post-treatment aggregate coefficients. $M_A^C$ depends only on the pre-treatment block biases, which the data identify, and on $\mathbf A$ and $\mathbf W$. It is never below $M$ and can be far above it. Appendix~\ref{app:covering-bound} derives it and reports it for the simulation designs.

One alternative is balanced aggregation. It applies aggregate inference only to the pre- and post-treatment relative periods observed for every cohort. It removes the change in treated composition but changes the estimand. The composition change of the not-yet-treated comparison group is also not addressed. Section~\ref{sec:sim-balanced} shows that aggregate inference on this balanced window comes closer to cohort inference on the same estimand, but a discrepancy remains.

Statistical uncertainty also differs. Cohort inference uses the cohort-period coefficients jointly, with their full covariance, and aggregate inference uses the aggregate coefficients. The confidence sets differ in length even if the two identified sets have the same width. In Section~\ref{sec:simulations} the confidence sets under aggregate inference are somewhat shorter when the two identified sets have the same width. Estimator equivalence (Section~\ref{sec:equivalence}) between \Imp{} and \CSNYT{} does not extend to aggregate inference, as Section~\ref{sec:minwage} will show.

\section{Fixed-length confidence intervals under SD}
\label{sec:flci}

Under SD, \RR{} also propose the fixed-length confidence interval (FLCI) of \citet{armstrong2018optimal}. The FLCI uses a full-path SD restriction, which bounds every second difference, including those entirely before treatment. This additional restriction can produce shorter confidence sets than the post-only C-LF and Conditional procedures of Section~\ref{sec:inference}. However, when the pure pre-treatment curvature of the population path exceeds the bound, its coverage guarantee may no longer apply.

\subsection{The FLCI in the cohort model}
\label{sec:flci-definition}

The target remains $\theta = \ell'\vec\tau_{\post}$. Consider an affine estimator that uses the pre-treatment coefficients to balance worst-case bias against sampling variance:
\begin{equation}
	\hat\theta_v = v'\vec{\hat\beta} = v_{\pre}'\vec{\hat\beta}_{\pre} + \ell'\vec{\hat\beta}_{\post}, \qquad v = (v_{\pre}, \ell).
	\label{eq:flci-estimator}
\end{equation}
The weights on the post-treatment coefficients must equal $\ell$, since any other weight would load on the unknown treatment effects. Under the model \eqref{eq:model}, the estimation error of $\hat\theta_v$ decomposes into a normally distributed sampling error with standard deviation $\sigma(v) = (v'\Sigma_N v)^{1/2}$ and a bias $v'\mathbf W\vec\Delta$.

Under a restriction $\Lambda_\Delta$, the absolute value of the bias is at most
\begin{equation}
	B(v; \Lambda_\Delta) = \sup_{\vec\Delta \in \Lambda_\Delta}\bigl|v'\mathbf W\vec\Delta\bigr|.
	\label{eq:flci-bias}
\end{equation}
For given weights, consider an interval centered at $\hat\theta_v$ with half-length equal to the standard deviation $\sigma(v)$ times the $1-\alpha$ quantile of $|N(B(v;\Lambda_\Delta)/\sigma(v), 1)|$. This interval guarantees coverage of at least $1-\alpha$ for every bias allowed by $\Lambda_\Delta$.

Let $\cvq_\alpha(r)$ denote the $1-\alpha$ quantile of $|N(r,1)|$ and $\chi(v)$ the resulting half-length. The FLCI chooses $v^*$ to minimize that half-length:
\begin{equation}
	\mathcal C^{\mathrm{FLCI}}_N = \bigl[\hat\theta_{v^*} - \chi(v^*),\ \hat\theta_{v^*} + \chi(v^*)\bigr], \qquad \chi(v) = \sigma(v)\,\cvq_\alpha\!\Bigl(\frac{B(v; \Lambda_\Delta)}{\sigma(v)}\Bigr), \qquad v^* \in \arg\min_{v = (v_{\pre}, \ell)} \chi(v).
	\label{eq:flci}
\end{equation}

Under the Gaussian model, $\hat{\theta}_{v^*} - \theta$ is normal with absolute mean at most $B(v^*; \Lambda_\Delta)$. Thus $\mathcal C^{\mathrm{FLCI}}_N$ covers $\theta$ with probability at least $1-\alpha$ for every $\vec\Delta \in \Lambda_\Delta$. With known $\Sigma_N$, its length $2\chi(v^*)$ is fixed across samples. Coverage is exact at the worst-case bias, and this is the shortest fixed-length interval among affine estimators \citep{armstrong2018optimal}. In practice, I replace $\Sigma_N$ with $\hat\Sigma_N$ in the weights $v^*$ and half-length $\chi$, so coverage holds asymptotically.\footnote{\RR{} discuss the asymptotic justification.}

\subsection{The two SD restrictions and the admissible weights}
\label{sec:flci-matrix}

The SD restriction of Section~\ref{sec:restrictions} bounds the second difference of each cohort's block-bias path at every post-treatment period $s \geq 1$. I call it the post-only restriction. The full-path restriction applies the same bound in every period for which three consecutive cohort-period cells are observed, including the periods that lie entirely before treatment:
\begin{equation}
\Lambda^{\mathrm{SD}}_\Delta(M) = \bigl\{\vec\Delta : |\nabla^2\Delta_{\mathcal G_g,s}| \leq M,\ s \geq 1\bigr\}, \qquad
\Lambda^{\mathrm{SD,F}}_\Delta(M) = \bigl\{\vec\Delta : |\nabla^2\Delta_{\mathcal G_g,s}| \leq M,\ \text{all such } s\bigr\}.
\label{eq:sd-classes}
\end{equation}
The full-path restriction is the one written in \RR{} and used by their FLCI, and it is the stricter set, $\Lambda^{\mathrm{SD,F}}_\Delta(M) \subseteq \Lambda^{\mathrm{SD}}_\Delta(M)$.

Order the cells with the pre-treatment cells first, $\vec{\hat\beta} = (\vec{\hat\beta}_{\pre}', \vec{\hat\beta}_{\post}')'$, and the block biases $\vec\Delta = (\vec\Delta_{\pre}', \vec\Delta_{\post}')'$ in the same way. By Propositions~\ref{prop:imputation} and~\ref{prop:csnyt} and the convention $\vec\delta_{\pre} = \vec\Delta_{\pre}$, the model \eqref{eq:model} has the block form
\begin{equation}
\E[\vec{\hat\beta}_{\pre}] = \vec\Delta_{\pre}, \qquad
\E[\vec{\hat\beta}_{\post}] = \vec\tau_{\post} + \mathbf W_{\post,\pre}\vec\Delta_{\pre} + \mathbf W_{\post,\post}\vec\Delta_{\post}.
\label{eq:flci-blocks}
\end{equation}
The blocks of $\mathbf W$ are those described in Section~\ref{sec:biasdecomposition}. $\mathbf W_{\post,\post}$ is unit upper triangular and common to the two estimators. $\mathbf W_{\post,\pre}$ is zero for the imputation estimator and, for \CSNYT, carries the not-yet-treated corrections, the entries $-w_k$.

Next write the second differences of the block-bias paths in matrix form, splitting the rows according to whether they involve a post-treatment cell:
\begin{equation}
\begin{pmatrix} \mathbf D_{\pre,\pre} & 0\\ \mathbf D_{\post,\pre} & \mathbf D_{\post,\post}\end{pmatrix}
\begin{pmatrix}\vec\Delta_{\pre}\\ \vec\Delta_{\post}\end{pmatrix}
= \begin{pmatrix} z\\ u\end{pmatrix}.
\label{eq:flci-rows}
\end{equation}
The rows of $\mathbf D_{\pre,\pre}$ are the pure pre-treatment second differences, those ending at $s \leq 0$, and $z$ collects them. The remaining rows end at $s \geq 1$, and $u$ collects these post-only second differences. In this notation $\Lambda^{\mathrm{SD}}_\Delta(M) = \{|u| \leq M\}$ and $\Lambda^{\mathrm{SD,F}}_\Delta(M) = \{|u| \leq M,\ |z| \leq M\}$. Each second difference puts coefficient one on its ending cell $s$ and loads otherwise only on earlier cells ($s-1$ and $s-2$) of the same cohort. Hence the square block $\mathbf D_{\post,\post}$ is unit lower triangular and $\mathbf D_{\pre,\pre}$ has full row rank.

For weights $v = (v_{\pre}, \ell)$, the bias of the estimator in \eqref{eq:flci-estimator} is
\[
b(v; \vec\Delta) = \E[\hat\theta_v] - \theta = v_{\pre}'\vec\Delta_{\pre} + \ell'\mathbf W_{\post,\pre}\vec\Delta_{\pre} + \ell'\mathbf W_{\post,\post}\vec\Delta_{\post}.
\]
Because the SD restrictions bound the post-treatment second differences $u$ rather than $\vec\Delta_{\post}$ directly, we use the invertibility of $\mathbf D_{\post,\post}$ in \eqref{eq:flci-rows} to write $\vec\Delta_{\post} = \mathbf D_{\post,\post}^{-1}(u - \mathbf D_{\post,\pre}\vec\Delta_{\pre})$. Substituting this expression into $b(v;\vec\Delta)$ and grouping terms by $\vec\Delta_{\pre}$ and $u$ yields
\begin{equation}
b(v; \vec\Delta) = (v_{\pre} - v^{\mathrm P}_{\pre})'\vec\Delta_{\pre} + \xi'u, \qquad
\xi = \bigl(\mathbf D_{\post,\post}'\bigr)^{-1}\mathbf W_{\post,\post}'\,\ell, \qquad
v^{\mathrm P}_{\pre} = \mathbf D_{\post,\pre}'\,\xi - \mathbf W_{\post,\pre}'\,\ell .
\label{eq:flci-center}
\end{equation}

The vector $\xi$ measures how each post-treatment slope change in $u$ propagates into the bias of the target estimand. A unit slope change in cohort $\mathcal G_g$'s block-bias path at relative period $r \ge 1$ raises the path's slope by one in every subsequent period. It therefore shifts cohort $\mathcal G_g$'s block bias at each period $s \ge r$ by $s - r + 1$. In calendar period $t = t_g + s - 1$, that shift $\Delta_{\mathcal G_g, s}$ enters cohort $\mathcal G_g$'s own estimation bias $\delta_{\mathcal G_g, s}$. As $\mathcal G_g$ can be an adjustment cohort for earlier-treated cohorts, it also transmits through $\mathbf W_{\post,\post}$ into the estimation bias $\delta_{\mathcal G_h, s_h(t)}$ of every earlier-treated cohort $\mathcal G_h$ ($t_h < t_g$) with weight $w_g$. For the target $\theta = \ell'\vec\tau_{\post}$, $\xi$ sums these exposures over the post-treatment cells with the weights $\ell$.

The vector $v^{\mathrm P}_{\pre} = \mathbf D_{\post,\pre}'\xi - \mathbf W_{\post,\pre}'\ell$ gives the weights on $\vec\Delta_{\pre}$ that eliminate the bias contribution of $\vec\Delta_{\pre}$ in $\ell'\vec{\hat\beta}_{\post}$. The first term $\mathbf D_{\post,\pre}'\xi$ assigns non-zero weights only to each cohort's last two pre-treatment cells ($s = -1$ and $s = 0$) to extrapolate its pre-treatment slope. The second term $-\mathbf W_{\post,\pre}'\ell$ is zero for the imputation estimator and, for \CSNYT, adds weights on each adjustment cohort's cell at $s_k(t_g - 1)$ to undo the base-period correction in Proposition~\ref{prop:csnyt}.

Because the second term of \eqref{eq:flci-center}, $\xi'u$, does not depend on $v_{\pre}$, its worst-case magnitude over $|u| \le M$ is $B^{\mathrm P} = M\|\xi\|_1$. When there are no further restrictions on $\vec\Delta_{\pre}$, as under the post-only SD restriction, the worst-case bias is unbounded whenever $v_{\pre} \neq v^{\mathrm P}_{\pre}$. The full-path SD restriction restricts the curvature of $\vec\Delta_{\pre}$. Under it, $v_{\pre}$ can depart from $v^{\mathrm P}_{\pre}$ to balance worst-case bias against sampling variance. Proposition~\ref{prop:flci-weights} formalizes both cases.

\begin{proposition}[Admissible FLCI weights]\label{prop:flci-weights}
The weights $v_{\pre}$ below are determined up to each cohort's normalization direction of Section~\ref{sec:coefvector}, along which $\hat\theta_v$ and its bias do not change.\footnote{Dropping each cohort's normalized $s = 0$ cell, as in Appendix~\ref{app:reduced}, makes $v_{\pre}$ unique.}
\begin{enumerate}[nosep,label=(\alph*)]
\item Under $\Lambda^{\mathrm{SD}}_\Delta(M)$, $B(v; \Lambda^{\mathrm{SD}}_\Delta) < \infty$ if and only if $v_{\pre} = v^{\mathrm P}_{\pre}$, and then $B(v; \Lambda^{\mathrm{SD}}_\Delta) = B^{\mathrm P}$.
\item Under $\Lambda^{\mathrm{SD,F}}_\Delta(M)$, $B(v; \Lambda^{\mathrm{SD,F}}_\Delta) < \infty$ if and only if $v_{\pre}$ is of the form $v_{\pre}= v^{\mathrm P}_{\pre} + \mathbf D_{\pre,\pre}'\gamma$ for some vector $\gamma \in \mathbb{R}^J$, where $J$ is the number of rows of $\mathbf D_{\pre,\pre}$, and then
\begin{equation}
B(\gamma) = B^{\mathrm P} + M\|\gamma\|_1 .
\label{eq:flci-bias-gamma}
\end{equation}
\end{enumerate}
\end{proposition}

Under the post-only restriction, part~(a) of Proposition~\ref{prop:flci-weights} pins down a single estimator, $c = v^{\mathrm P\prime}\vec{\hat\beta}$ with $v^{\mathrm P} = (v^{\mathrm P}_{\pre}, \ell)$. This estimator subtracts the linear extrapolation of each cohort's last two pre-treatment periods from $\ell'\vec{\hat\beta}_{\post}$ and attains the minimum possible worst-case bias $B^{\mathrm P}$. Under the full-path restriction, part~(b) shows that the admissible estimators form a family indexed by the vector $\gamma \in \mathbb{R}^J$ of the form $\hat\theta_{v(\gamma)} = c + \gamma'\hat z$. Here $\hat z = \mathbf D_{\pre,\pre}\vec{\hat\beta}_{\pre}=\mathbf D_{\pre,\pre}\vec{\hat\Delta}_{\pre}$ collects the $J$ observed pure pre-treatment second differences of block biases. Setting $\gamma = 0$ recovers the post-only estimator $c$ and minimizes worst-case bias at $B(0) = B^{\mathrm P}$. Choosing $\gamma \neq 0$ uses the pure pre-treatment second differences $\hat z$ to reduce sampling variance at the cost of adding $M\|\gamma\|_1$ to the worst-case bias.

\subsection{The centers and lengths of the identified and confidence sets}
\label{sec:flci-center}

Under the post-only SD restriction, Proposition~\ref{prop:flci-weights}(a) fixes the affine estimator with finite worst-case bias at $c = v^{\mathrm P\prime}\vec{\hat\beta}$. The FLCI confidence set under post-only SD is centered at $c$. The following proposition shows that the plug-in identified set has the  midpoint at $c$ and that the Conditional and C-LF confidence sets in Section~\ref{sec:confidence} are symmetric about $c$.

\begin{proposition}[The post-only center]
\label{prop:post-only-center}
Let $c = v^{\mathrm P\prime}\vec{\hat\beta}$.
\begin{enumerate}[nosep,label=(\alph*)]
\item The plug-in and population identified sets satisfy
\[
\mathcal S\bigl(\vec{\hat\beta}, \Lambda^{\mathrm{SD}}_\Delta(M)\bigr) = [\,c - B^{\mathrm P},\ c + B^{\mathrm P}\,], \qquad
\mathcal S\bigl(\vec\beta, \Lambda^{\mathrm{SD}}_\Delta(M)\bigr) = [\,v^{\mathrm P\prime}\vec\beta - B^{\mathrm P},\ v^{\mathrm P\prime}\vec\beta + B^{\mathrm P}\,].
\]
\item Under $\Lambda^{\mathrm{SD}}_\Delta(M)$, the Conditional and C-LF confidence sets are symmetric about $c$.
\item $c$ takes the same value whether $\vec{\hat\beta}$ and $\mathbf W$ are those of the imputation estimator or those of \CSNYT.
\end{enumerate}
\end{proposition}
Appendix~\ref{app:post-only-center} proves the proposition. Part~(a) gives the common length $2B^{\mathrm P} = 2M\|\xi\|_1$ of the population and plug-in identified sets under the post-only SD, which grows linearly in $M$ while their centers stay fixed.

The confidence sets also account for sampling uncertainty in the estimated coefficients. Under the post-only restriction, the FLCI shares the center $c$ with the Conditional and C-LF confidence sets but can differ in length. From \eqref{eq:flci}, fixing the weights at $v^{\mathrm P}$ gives the post-only FLCI length $2\sigma(v^{\mathrm P})\,\cvq_\alpha\!\left(\frac{B^{\mathrm P}}{\sigma(v^{\mathrm P})}\right)$.

The post-only center $c$ can also be written in block-bias coordinates, which shows which coefficients it uses. Let $\mathbf D_{\post} = (\mathbf D_{\post,\pre}\ \mathbf D_{\post,\post})$ collect the post-only rows of \eqref{eq:flci-rows} and $\vec{\hat\Delta} = \mathbf W^{-1}\vec{\hat\beta}$ denote the plug-in block-bias path. Then $v^{\mathrm P\prime} = \xi'\mathbf D_{\post}\mathbf W^{-1}$, so that $c = \xi'\mathbf D_{\post}\vec{\hat\Delta} = \xi'\bigl(\mathbf D_{\post,\pre}\vec{\hat\Delta}_{\pre} + \mathbf D_{\post,\post}\vec{\hat\Delta}_{\post}\bigr)$. The center $c$ uses only the last two pre-treatment block biases of each cohort ($s = -1$ and $s = 0$) together with the plug-in post-treatment path. For \CSNYT{} that path itself depends on earlier pre-treatment coefficients of the adjustment cohorts. The full-path FLCI departs from $c$. Under $\Lambda^{\mathrm{SD,F}}_\Delta(M)$ it may put weight on the pure pre-treatment second differences $\hat z = \mathbf D_{\pre,\pre}\vec{\hat\beta}_{\pre}$, shifting its center to $c + \gamma'\hat z$ with $\gamma$ chosen to minimize the interval length. Its center therefore varies with $M$ and the covariance. With $v(\gamma) = (v^{\mathrm P}_{\pre} + \mathbf D_{\pre,\pre}'\gamma, \ell)$, the full-path FLCI solves
\begin{equation}
\min_{\gamma}\ \sigma\bigl(v(\gamma)\bigr)\,\cvq_\alpha\!\Bigl(\frac{B^{\mathrm P} + M\|\gamma\|_1}{\sigma(v(\gamma))}\Bigr).
\label{eq:flci-program}
\end{equation}
The program balances sampling variance against the additional worst-case bias $M\|\gamma\|_1$ and is convex (Appendix~\ref{app:flci-program}). Since $\gamma = 0$ remains feasible, the full-path FLCI is never longer than the post-only FLCI at the same bound and covariance.

At $M = 0$ the program minimizes variance, giving $\gamma = -\mathrm{Var}(\hat z)^{-1}\mathrm{Cov}(\hat z,c)$. The pure pre-treatment second differences then act as control variates with zero means under the restriction. As $M$ increases, each unit of weight on $\hat z$ incurs more worst-case bias. Above a threshold depending on the target, design, and covariance, $\gamma = 0$ is the optimal solution and the full-path and post-only FLCIs coincide. Appendix~\ref{app:flci-program} gives the condition, and Figure~\ref{fig:app-minwage-gamma} shows the thresholds in the application.

The stronger restriction on the pre-treatment block biases does not narrow the identified set. The identified set under the full-path SD restriction coincides with the post-only interval in Proposition~\ref{prop:post-only-center}(a) when $|\mathbf D_{\pre,\pre}\vec\beta_{\pre}| \le M$ and is empty otherwise. If a population second difference in $\mathbf D_{\pre,\pre}\vec\beta_{\pre}$ exceeds $M$, the adjustment $\gamma'\hat z$ can introduce more bias than the allowance $M\|\gamma\|_1$ accounts for. The resulting confidence set may then fail to achieve nominal coverage. The exposure is limited to the second differences that receive weight. A component with $\gamma_j = 0$ contributes nothing to the bias, so the guarantee survives a violation on that component. By contrast, procedures constructed under the post-only SD restriction retain their coverage guarantee regardless of the pure pre-treatment curvature.

In the default setting of the \texttt{HonestDiD} package \citep{honestdid026}, the FLCI is implemented under the full-path restriction, whereas the Conditional and C-LF confidence sets are implemented under the post-only restriction. Following this convention, FLCI hereafter refers to the full-path interval paired with the post-only Conditional and C-LF confidence sets. Sections~\ref{sec:simulations} and~\ref{sec:minwage} compare their finite-sample performance and empirical behavior. Appendix~\ref{app:flci-designs} isolates the FLCI in two Gaussian designs calibrated to the minimum-wage example (Figure~\ref{fig:app-cal-profiles}). The full-path FLCI is noticeably shorter than the post-only procedures at small bounds $M$ (Figure~\ref{fig:app-cal-length}, panels (a), (c), and (d)). When the pure pre-treatment second differences exceed the specified bound, the full-path FLCI loses coverage steadily as the excess grows, whereas Conditional, C-LF, and the post-only FLCI stay at the nominal level (Appendix~\ref{app:flci-designs}).
\section{Cohort versus aggregate inference: simulation evidence}
\label{sec:simulations}

I use four simulation designs to compare cohort and aggregate inference. Designs~I to~III build on the minimum-wage treatment structure of Section~\ref{sec:illustration}. Design~I varies the heterogeneity of cohort block-bias paths, Design~II compares balanced aggregation against cohort inference, and Design~III has few treated clusters. Design~IV divides a fixed treated sample into an increasing number of smaller cohorts. All four designs report results for the imputation estimator, with aggregate estimates and covariances computed as $\mathbf A\vec{\hat\beta}$ and $\mathbf A\hat\Sigma_N\mathbf A'$ using the covariance estimator $\hat\Sigma_N$ of Section~\ref{sec:coefvector}.

\subsection{Designs and reported quantities}
\label{sec:sim-designs}

In all four designs, the true treatment effects are set to zero, so the population coefficient paths reflect only PT violations. The disturbances follow Gaussian $\mathrm{AR}(1)$ processes. Designs~I to~III build on the minimum-wage panel structure ($2{,}284$ units and three treated cohorts). Design~I uses the baseline $T = 7$ panel with unit-clustered covariances, while Design~II extends the horizon to $T = 10$ to create a wider balanced window across cohorts. Design~III introduces state-level clustering. I drop the cohort whose 100 units lie within a single state and retain the state assignments of the remaining units (28 states, with 3 and 9 treated states in the two remaining cohorts). I also add a common state shock. Design~IV moves to a balanced panel of 60 units over $T = 21$ periods, dividing the 30 treated units into $G \in \{3, 5, 10, 15\}$ equal-sized cohorts of 10, 6, 3, and 2 units. To vary precision while holding the identified set fixed, Designs~I to~III scale the disturbance standard deviation by a multiplier $\lambda \in \{2, 1, 0.5, 0.25, 0.1\}$, whereas Design~IV fixes the noise multiplier at $1/\sqrt 5$.

Designs~I and~III use the average treatment effect of all post-treatment periods. Design~II compares this estimand with the average treatment effect over the post-treatment relative periods observed for all treated cohorts (the balanced target), and Design~IV reports both.

For covariance estimation, Designs~I and~II use the unit-clustered sandwich estimator of Section~\ref{sec:coefvector}. To evaluate finite-sample covariance adjustments when cohorts or clusters are small, Designs~III and~IV compute each confidence set on the same simulated draws under three covariance matrices: the baseline sandwich estimator (CR1), the bias-reduced CR2 estimator of \citet{bell2002bias}, and the true population covariance (the oracle). All designs use $R = 1{,}000$ replications. 

Under SD, I normalize the sensitivity bound under cohort inference by a fixed reference curvature unit $m_0$, setting the sensitivity parameter under cohort inference to $M_C = 1$ in normalized units. At the baseline noise level ($\lambda = 1$ in Design~I), $M_C$ equals roughly $3.2$ times the standard error of the estimated treatment effect in the first post-treatment period, so the bound and the standard error are of comparable magnitude.

To compare aggregate inference against cohort inference, I calibrate the aggregate bound in two ways. The first is the same-value bound, which simply sets $M_A^S = M_C$. The second is the width-matched bound $M_A^E$, which chooses the aggregate bound so that the population identified set under aggregate inference and the post-only SD restriction has the same width as the identified set under cohort inference at $M_C$. The RM analogues set $\bar M_C = \bar M_A^S = 1$ and calibrate $\bar M_A^E$ to match the width of the identified set under cohort inference, which under RM varies with the bias path.

Finally, I report results for C-LF under the post-only restriction (for both SD and RM) and for the FLCI under the full-path restriction (for SD). Throughout, ``cohort'' or ``aggregate'' before a method's name denotes the inference type: the cohort FLCI is the FLCI under cohort inference, and the aggregate C-LF confidence set is the C-LF confidence set under aggregate inference. The primary performance metric is minimum coverage over the population identified set under cohort-level restriction, defined as the lowest empirical probability that the confidence set covers a target value admitted by the cohort restriction. Appendix~\ref{app:simulations} reports the mean centers and lengths of the confidence sets and the rejection curves as secondary metrics. I omit the Conditional procedure, which shares the center of C-LF and has nearly the same length (Figure~\ref{fig:app-cal-length}(a)).

\subsection{Heterogeneous cohort paths (Design I)}
\label{sec:sim-heterogeneous}

Design~I evaluates how heterogeneity across cohort bias paths affects cohort and aggregate inference under SD, with three DGPs. Under SD1, every block bias is zero, so parallel trends hold. Under SD2, all cohorts accumulate bias in the same direction with period-by-period second differences equal to $M_C = 1$, while the late-treated cohort follows a different initial slope. Under SD3, the early and middle-treated cohorts drift in the opposite direction of the late cohort. By construction, every cohort path in SD2 and SD3 satisfies the cohort SD restriction at $M_C = 1$. However, because the cohorts follow different slopes or levels, the aggregate coefficients $\mathbf A\mathbf W\vec\Delta$ exhibit sharp kinks at event times. Their second differences therefore exceed both the same-value bound $M_A^S$ and the width-matched bound $M_A^E$.

A dagger in Table~\ref{tab:sim-design1} marks the cells in which this happens. The aggregate restriction with bound $M_A$ then does not accommodate the aggregate bias path, although every cohort-level block bias satisfies the cohort restriction at $M_C$.\footnote{Appendix~\ref{app:covering-bound} derives the covering bound $M_A^C$, the smallest aggregate bound that can accommodate the cohort restriction. In Design~I it is $8.3$, $8.8$, and $13.8\,M_C$ under SD1 to SD3, against $M_A^E = 1.04\,M_C$. Table~\ref{tab:A-covering} reports all designs.} In such a cell the identified set under aggregate inference need not contain the true target value, the coverage guarantee of the confidence set under aggregate inference does not apply, and the reported coverage measures how far it fails. Under SD1 the aggregate coefficients are zero and no cell is marked.

\begin{table}[!h]
\centering
\caption{Design I: minimum coverage over the identified set under cohort inference, percent}
\label{tab:sim-design1}
\footnotesize
\begin{tabular}{@{}llrrrrrrrrrr@{}}
\toprule
 & & \multicolumn{6}{c}{Minimum coverage, percent} & \multicolumn{4}{c}{Mean length, aggregate/cohort} \\
\cmidrule(lr){3-8} \cmidrule(lr){9-12}
 & & \multicolumn{2}{c}{Cohort, $M_C$} & \multicolumn{2}{c}{Aggregate, $M_A^S$} & \multicolumn{2}{c}{Aggregate, $M_A^E$} & \multicolumn{2}{c}{$M_A^S$} & \multicolumn{2}{c}{$M_A^E$} \\
\cmidrule(lr){3-4} \cmidrule(lr){5-6} \cmidrule(lr){7-8} \cmidrule(lr){9-10} \cmidrule(lr){11-12}
DGP & $\lambda$ & C-LF & FLCI & C-LF & FLCI & C-LF & FLCI & C-LF & FLCI & C-LF & FLCI \\
\midrule
\multirow{5}{*}{SD1} & 2 & 95.8 & 96.1 & 95.5 & 96.2 & 96.2 & 96.4 & 0.868 & 0.906 & 0.883 & 0.926
\\
 & 1 & 95.4 & 96.3 & 94.2 & 93.4 & 95.3 & 94.8 & 0.895 & 0.902 & 0.917 & 0.925
\\
 & 0.5 & 95.1 & 94.6 & 92.3 & 92.1 & 94.9 & 94.8 & 0.917 & 0.917 & 0.945 & 0.945
\\
 & 0.25 & 95.1 & 94.6 & 87.6 & 87.2 & 94.9 & 94.8 & 0.936 & 0.936 & 0.969 & 0.969
\\
 & 0.1 & 95.1 & 94.6 & 65.9 & 65.0 & 94.9 & 94.8 & 0.951 & 0.951 & 0.986 & 0.986
\\
\midrule
\multirow{5}{*}{SD2} & 2 & 95.8 & 94.5 & 86.8$^{\dagger}$ & 79.7$^{\dagger}$ & 87.8$^{\dagger}$ & 82.4$^{\dagger}$ & 0.868 & 0.906 & 0.883 & 0.926
\\
 & 1 & 95.4 & 94.3 & 64.0$^{\dagger}$ & 61.9$^{\dagger}$ & 68.4$^{\dagger}$ & 66.2$^{\dagger}$ & 0.895 & 0.902 & 0.917 & 0.925
\\
 & 0.5 & 95.1 & 94.6 & 15.6$^{\dagger}$ & 14.8$^{\dagger}$ & 22.7$^{\dagger}$ & 22.3$^{\dagger}$ & 0.917 & 0.917 & 0.945 & 0.945
\\
 & 0.25 & 95.1 & 94.6 & 0.0$^{\dagger}$ & 0.0$^{\dagger}$ & 0.1$^{\dagger}$ & 0.0$^{\dagger}$ & 0.936 & 0.936 & 0.969 & 0.969
\\
 & 0.1 & 95.1 & 94.6 & 0.0$^{\dagger}$ & 0.0$^{\dagger}$ & 0.0$^{\dagger}$ & 0.0$^{\dagger}$ & 0.951 & 0.951 & 0.986 & 0.986
\\
\midrule
\multirow{5}{*}{SD3} & 2 & 95.8 & 94.5 & 11.8$^{\dagger}$ & 3.7$^{\dagger}$ & 13.2$^{\dagger}$ & 5.3$^{\dagger}$ & 0.868 & 0.906 & 0.883 & 0.926
\\
 & 1 & 95.4 & 94.3 & 0.0$^{\dagger}$ & 0.0$^{\dagger}$ & 0.0$^{\dagger}$ & 0.0$^{\dagger}$ & 0.895 & 0.902 & 0.917 & 0.925
\\
 & 0.5 & 95.1 & 94.6 & 0.0$^{\dagger}$ & 0.0$^{\dagger}$ & 0.0$^{\dagger}$ & 0.0$^{\dagger}$ & 0.917 & 0.917 & 0.945 & 0.945
\\
 & 0.25 & 95.1 & 94.6 & 0.0$^{\dagger}$ & 0.0$^{\dagger}$ & 0.0$^{\dagger}$ & 0.0$^{\dagger}$ & 0.936 & 0.936 & 0.969 & 0.969
\\
 & 0.1 & 95.1 & 94.6 & 0.0$^{\dagger}$ & 0.0$^{\dagger}$ & 0.0$^{\dagger}$ & 0.0$^{\dagger}$ & 0.951 & 0.951 & 0.986 & 0.986
\\
\bottomrule
\end{tabular}

\vspace{0.5em}
\parbox{\linewidth}{\footnotesize
    \textit{Notes:} Minimum pointwise coverage over the population identified set under cohort inference, in percent. Rows: DGP and noise multiplier $\lambda$. Columns: C-LF confidence sets and FLCIs under cohort inference at $M_C$ and under aggregate inference at $M_A^S = M_C$ and at the width-matched $M_A^E = 1.04\,M_C$, followed by the mean length of the confidence set under aggregate inference relative to the one under cohort inference. A dagger marks cells in which the aggregate bias path of the DGP violates the aggregate restriction. Nominal coverage 95 percent.
}
\end{table}

Table~\ref{tab:sim-design1} reports minimum empirical coverage across all noise levels $\lambda$ for both C-LF and the FLCI. Cohort inference achieves near-nominal coverage in every specification. Under SD1, where PT holds, width-matched aggregate inference ($M_A^E$) also maintains nominal coverage and yields confidence sets that are 1 to 12 percent shorter than those under cohort inference. The length difference narrows as noise falls. Even under parallel trends, the same-value bound ($M_A^S = M_C$) falls to 65 to 66 percent at $\lambda = 0.1$, because setting $M_A = M_C$ generates an identified set under aggregate inference that is narrower than the one under cohort inference.

Once cohort-level block bias paths differ, as in SD2 and SD3, coverage under aggregate inference deteriorates as sampling noise decreases under both $M_A^S$ and $M_A^E$. The breakdown comes from the location of the population identified set under aggregate inference. The identified set under cohort inference spans $[0, 5]$ in normalized units under both DGPs. The center of the confidence set under aggregate inference shifts upward by $0.8$ to $0.9$ units under SD2 and by $4.3$ to $4.5$ units under SD3 across C-LF and the FLCI. 

Under RM the C-LF confidence set under cohort inference is conservative. The three RM data-generating processes likewise vary the degree of cohort-level block-bias heterogeneity. Across them, cohort minimum coverage ranges from 99.7 to 100 percent in every cell, reflecting the union bound over the components of the RM restriction. Aggregate inference still under-covers at low noise levels: at $\bar M_A^E$ in all three data-generating processes and at $\bar M_A^S$ in RM1 and RM2. Table~\ref{tab:app-design1-rm} and Figure~\ref{fig:app-rm-profiles} report the coverage and rejection curves under RM.

Appendix~\ref{app:sim-design1} collects the Design~I exhibits. Figure~\ref{fig:sim-paths} plots the block biases of the SD and RM DGPs with the induced aggregate path. Figure~\ref{fig:sim-profiles} shows the C-LF rejection curves at two noise levels, and Figures~\ref{fig:app-clf-width} and~\ref{fig:app-flci-width} the mean widths and centers of the confidence sets.

\subsection{Balanced aggregation (Design II)}
\label{sec:sim-balanced}

Design~II compares aggregate and cohort inference on the balanced target and repeats the comparison for the average treatment effect of all post-treatment periods. Both use aggregate inference at $M_A^E$ against cohort inference at $M_C$ under the SD restriction. For the balanced target, the aggregate coefficients are formed on the window of pre- and post-treatment relative periods in which all treated cohorts are observed. Cohort inference uses every cohort-period coefficient as in Design~I, with the target weights only on $s = 1$ to $4$.

\begin{table}[!h]
\centering
\caption{Design II: minimum coverage over the identified set under cohort inference by target, percent}
\label{tab:sim-design2}
\footnotesize
\begin{tabular}{@{}llrrrrrrrr@{}}
\toprule
 & & \multicolumn{4}{c}{Average effect over all post-treatment periods} & \multicolumn{4}{c}{Average effect over the balanced window} \\
\cmidrule(lr){3-6} \cmidrule(lr){7-10}
 & & \multicolumn{2}{c}{Cohort, $M_C$} & \multicolumn{2}{c}{Aggregate, $M_A^E$} & \multicolumn{2}{c}{Cohort, $M_C$} & \multicolumn{2}{c}{Aggregate, $M_A^E$} \\
\cmidrule(lr){3-4} \cmidrule(lr){5-6} \cmidrule(lr){7-8} \cmidrule(lr){9-10}
DGP & $\lambda$ & C-LF & FLCI & C-LF & FLCI & C-LF & FLCI & C-LF & FLCI \\
\midrule
\multirow{5}{*}{SD1} & 2 & 94.8 & 98.0 & 95.7 & 98.3 & 95.5 & 98.0 & 95.3 & 98.0
\\
 & 1 & 94.8 & 94.8 & 95.3 & 95.2 & 95.1 & 95.0 & 95.3 & 95.2
\\
 & 0.5 & 94.8 & 94.7 & 95.3 & 95.2 & 95.1 & 95.0 & 95.3 & 95.2
\\
 & 0.25 & 94.8 & 94.7 & 95.3 & 95.2 & 95.1 & 95.0 & 95.3 & 95.2
\\
 & 0.1 & 94.8 & 94.7 & 95.3 & 95.2 & 95.1 & 95.0 & 95.3 & 95.2
\\
\midrule
\multirow{5}{*}{SD2} & 2 & 94.8 & 93.9 & 94.6$^{\dagger}$ & 90.9$^{\dagger}$ & 95.5 & 93.8 & 95.7$^{\dagger}$ & 95.5$^{\dagger}$
\\
 & 1 & 94.8 & 94.4 & 90.5$^{\dagger}$ & 90.0$^{\dagger}$ & 95.1 & 95.0 & 95.3$^{\dagger}$ & 94.9$^{\dagger}$
\\
 & 0.5 & 94.8 & 94.7 & 83.2$^{\dagger}$ & 82.7$^{\dagger}$ & 95.1 & 95.0 & 92.4$^{\dagger}$ & 91.8$^{\dagger}$
\\
 & 0.25 & 94.8 & 94.7 & 64.1$^{\dagger}$ & 63.6$^{\dagger}$ & 95.1 & 95.0 & 88.0$^{\dagger}$ & 87.6$^{\dagger}$
\\
 & 0.1 & 94.8 & 94.7 & 7.6$^{\dagger}$ & 6.9$^{\dagger}$ & 95.1 & 95.0 & 68.7$^{\dagger}$ & 68.0$^{\dagger}$
\\
\midrule
\multirow{5}{*}{SD3} & 2 & 94.8 & 93.9 & 83.2$^{\dagger}$ & 76.3$^{\dagger}$ & 95.5 & 93.8 & 90.4$^{\dagger}$ & 97.7$^{\dagger}$
\\
 & 1 & 94.8 & 94.4 & 59.5$^{\dagger}$ & 58.9$^{\dagger}$ & 95.1 & 95.0 & 84.2$^{\dagger}$ & 83.7$^{\dagger}$
\\
 & 0.5 & 94.8 & 94.7 & 13.9$^{\dagger}$ & 13.5$^{\dagger}$ & 95.1 & 95.0 & 62.1$^{\dagger}$ & 61.1$^{\dagger}$
\\
 & 0.25 & 94.8 & 94.7 & 0.0$^{\dagger}$ & 0.0$^{\dagger}$ & 95.1 & 95.0 & 16.3$^{\dagger}$ & 15.3$^{\dagger}$
\\
 & 0.1 & 94.8 & 94.7 & 0.0$^{\dagger}$ & 0.0$^{\dagger}$ & 95.1 & 95.0 & 0.0$^{\dagger}$ & 0.0$^{\dagger}$
\\
\bottomrule
\end{tabular}

\vspace{0.5em}
\parbox{\linewidth}{\footnotesize
    \textit{Notes:} Minimum pointwise coverage over the population identified set under cohort inference, in percent. Rows: DGP and noise multiplier $\lambda$. Columns: for the average treatment effect of all post-treatment periods and for the balanced target, C-LF confidence sets and FLCIs under cohort inference at $M_C$ and under aggregate inference at the width-matched $M_A^E$ ($1.08\,M_C$ and $1.04\,M_C$). A dagger marks cells in which the aggregate bias path of the DGP violates the aggregate restriction. Nominal coverage 95 percent.
}
\end{table}

Table~\ref{tab:sim-design2} shows that moving to the balanced estimand improves the coverage of aggregate inference at moderate noise levels, but fails to restore validity as sampling noise decreases. Under SD2, shifting from the average treatment effect of all post-treatment periods to the balanced target raises the coverage of the aggregate FLCI at $M_A^E$ from 90 to 95 percent at $\lambda = 1$ and from 7 to 68 percent at $\lambda = 0.1$. Under the more heterogeneous DGP SD3, aggregate FLCI coverage increases from 59 to 84 percent at $\lambda = 1$, yet still collapses to zero for both targets at $\lambda = 0.1$. Aggregate C-LF coverage lies within one percentage point of the FLCI figures at $\lambda \leq 1$, whereas cohort inference maintains near-nominal coverage throughout.

Balanced aggregation falls short because holding the treated composition fixed addresses only the first composition change of Section~\ref{sec:illustration}. For the imputation estimator the comparison group also changes inside the balanced window. When a later cohort adopts, its block bias enters the earlier cohorts' coefficients through $\mathbf W$, and the aggregate coefficients $\mathbf A\mathbf W\vec\Delta$ may have additional curvature at that date. Under SD2 and SD3 these curvatures place the aggregate population coefficients outside the aggregate restriction, and the identified set under aggregate inference lies $0.1$ and $0.7$ units below the one under cohort inference.\footnote{On the balanced window the covering bound $M_A^C$ of Appendix~\ref{app:covering-bound} is $1.1$, $1.3$, and $1.8\,M_C$ under SD1 to SD3, against $M_A^E = 1.04\,M_C$. With $\mathbf W$ replaced by the identity it equals $M_C$, so the excess comes from the change of comparison group alone.} Appendix~\ref{app:sim-design2} states the balanced window.

\subsection{Few independent clusters per cohort (Designs III and IV)}
\label{sec:sim-support}

Because cohort inference uses the full vector of cohort-period coefficients rather than the aggregate coefficients, its finite-sample performance depends on how precisely $\hat\Sigma_N$ is estimated. Designs~III and~IV estimate $\hat\Sigma_N$ from little variation within cohorts: few treated states in Design~III and few units per cohort in Design~IV. Comparing CR1 and CR2 with the population covariance isolates the role of covariance estimation error. Design~III uses the average treatment effect of all post-treatment periods, as in Design~I. Panel~B of Design~IV uses the balanced target of Design~II.

\begin{table}[!h]
\centering
\caption{Designs III and IV: coverage with few treated clusters and with many small cohorts}
\label{tab:sim-design34}
\footnotesize
\begin{tabular}{@{}llrrrrr@{}}
\toprule
\multicolumn{7}{@{}l}{\textit{Panel A. Design III: few treated clusters (SD1, state-clustered covariance)}} \\
\addlinespace[2pt]
 & & \multicolumn{2}{c}{Cohort, $M_C$} & \multicolumn{2}{c}{Aggregate, $M_A^E$} & Length ratio \\
\cmidrule(lr){3-4} \cmidrule(lr){5-6}
$\lambda$ & Covariance & C-LF & FLCI & C-LF & FLCI & aggregate/cohort \\
\midrule
\multirow{3}{*}{1} & CR1 & 92.0 & 89.7 & 93.0 & 93.3 & 0.981
\\
 & CR2 & 94.1 & 91.2 & 94.3 & 94.0 & 0.958
\\
 & Oracle & 96.5 & 96.4 & 95.7 & 96.0 & 0.953
\\
\midrule
\multirow{3}{*}{0.1} & CR1 & 90.8 & 90.5 & 92.1 & 91.8 & 0.997
\\
 & CR2 & 92.9 & 92.5 & 93.5 & 93.3 & 0.992
\\
 & Oracle & 95.4 & 95.1 & 94.9 & 94.6 & 0.990
\\
\midrule
\multicolumn{7}{@{}l}{\textit{Panel B. Design IV: more and smaller cohorts (SD1, balanced target, unit-clustered covariance)}} \\
\addlinespace[2pt]
 & & \multicolumn{2}{c}{Cohort, $M_C$} & \multicolumn{2}{c}{Aggregate, $M_A^E$} & Length ratio \\
\cmidrule(lr){3-4} \cmidrule(lr){5-6}
$G$ (units) & Covariance & C-LF & FLCI & C-LF & FLCI & aggregate/cohort \\
\midrule
\multirow{3}{*}{3 (10)} & CR1 & 94.6 & 95.4 & 94.5 & 96.2 & 0.992
\\
 & CR2 & 94.7 & 95.5 & 94.5 & 96.2 & 0.991
\\
 & Oracle & 95.6 & 97.5 & 95.6 & 97.5 & 0.999
\\
\addlinespace[3pt]
\multirow{3}{*}{5 (6)} & CR1 & 95.0 & 93.5 & 94.4 & 96.6 & 0.986
\\
 & CR2 & 95.4 & 93.4 & 94.9 & 96.7 & 0.985
\\
 & Oracle & 95.4 & 97.7 & 95.4 & 97.7 & 0.999
\\
\addlinespace[3pt]
\multirow{3}{*}{10 (3)} & CR1 & 94.4 & 79.7 & 93.6 & 95.1 & 0.966
\\
 & CR2 & 96.1 & 79.7 & 95.3 & 96.7 & 0.963
\\
 & Oracle & 96.4 & 97.7 & 96.3 & 97.7 & 0.988
\\
\addlinespace[3pt]
\multirow{3}{*}{15 (2)} & CR1 & 90.6 & 67.6 & 88.7 & 90.1 & 0.956
\\
 & CR2 & 93.9 & 67.8 & 92.5 & 94.1 & 0.949
\\
 & Oracle & 94.5 & 97.1 & 94.4 & 96.9 & 0.984
\\
\bottomrule
\end{tabular}

\vspace{0.5em}
\parbox{\linewidth}{\footnotesize
    \textit{Notes:} Minimum pointwise coverage over the population identified set under cohort inference, in percent, under SD1. Rows: Panel~A, Design~III by noise multiplier $\lambda$; Panel~B, Design~IV by number of cohorts $G$ (units per cohort in parentheses); within each block, the covariance supplied to the procedures, CR1 (clustered sandwich), CR2 (bias-reduced), or Oracle (population). Columns: C-LF confidence sets and FLCIs under cohort inference at $M_C$ and aggregate inference at the width-matched $M_A^E$ ($1.045\,M_C$ in Panel~A, $1.000$ to $1.021\,M_C$ in Panel~B), then the ratio of the mean length of the aggregate C-LF confidence set to that of the cohort C-LF confidence set. Nominal coverage 95 percent.
}
\end{table}

Panel~A of Table~\ref{tab:sim-design34} reports results for Design~III under parallel trends (SD1) at $\lambda = 1$ and $\lambda = 0.1$. With few treated states, using the baseline CR1 covariance yields coverage below the 95 percent nominal level for all procedures. Across the five noise levels, two of which Panel~A shows, the cohort C-LF confidence set covers 90.8 to 92.4 percent, the aggregate C-LF confidence set covers 92.1 to 93.0 percent, and the cohort FLCI drops furthest to between 86.3 and 90.9 percent. Replacing CR1 with the oracle covariance restores the nominal coverage for every procedure, confirming that the under-coverage stems entirely from covariance estimation. The bias-reduced CR2 estimator recovers part of the coverage loss for C-LF but provides less improvement for the cohort FLCI. At $\lambda = 1$, aggregate C-LF confidence sets with $M_{A}^{E}$ are 2 to 5 percent shorter than cohort C-LF confidence sets under the same covariance estimator. Table~\ref{tab:app-design3} shows that under the heterogeneous DGPs SD2 and SD3, the severe under-coverage of aggregate inference observed in Design~I reappears under all three covariances.

Panel~B of Table~\ref{tab:sim-design34} shows an even sharper contrast between C-LF and the cohort FLCI in Design~IV as the 30 treated units are partitioned into an increasing number of smaller cohorts. For C-LF, both cohort and aggregate inference remain near the 95 percent nominal level down to 3 units per cohort ($G = 10$). Even with only 2 units per cohort ($G = 15$), C-LF loses only 4 to 6 percentage points under CR1, and switching to CR2 recovers most of this shortfall. By contrast, the cohort FLCI deteriorates sharply as cohort size shrinks. Its coverage under CR1 falls to 79.7 percent at 3 units per cohort and 67.6 percent at 2 units per cohort (compared to 95.1 and 90.1 percent for the aggregate FLCI). Moreover, while the oracle covariance lifts all FLCI procedures to roughly 97 percent coverage, applying CR2 does not improve the coverage of the cohort FLCI. Table~\ref{tab:app-design4} shows that the average treatment effect of all post-treatment periods exhibits the same pattern.

With 3 and 2 units per cohort, the cohort FLCI using estimated covariance is 19.5 and 32.0 percent shorter than the cohort FLCI under the oracle covariance and undercovers. The cohort FLCI breaks down in small cohorts because the weight optimization in \eqref{eq:flci-program} ``overfits'' sampling noise in the estimated covariance matrix $\hat\Sigma_N$. As $G$ increases from 3 to 15 and cohort sizes shrink, $\hat z$ gains elements while $\hat\Sigma_N$ becomes noisier. Because $\gamma$ minimizes the interval length computed from $\hat\Sigma_N$, it finds weights whose estimated variance of $c + \gamma'\hat z$ is small by chance, so the reported standard error understates the true one. The aggregate FLCI is less exposed to this overfitting because it optimizes over a much smaller vector of aggregate pre-treatment coefficients.

Appendices~\ref{app:sim-design3} and~\ref{app:sim-design4} collect the remaining exhibits. Table~\ref{tab:app-design3} gives Design~III's minimum coverage for every DGP and noise level under the three covariances, and Figure~\ref{fig:app-oracle-clf} the rejection curves of its cohort C-LF confidence sets with the estimated and the oracle covariance. Table~\ref{tab:app-design4} gives Design~IV's coverage, calibrations, and length ratios for the average treatment effect of all post-treatment periods under the three covariances, and Figure~\ref{fig:app-c-clf} its rejection curves with the largest and the smallest cohorts.

\subsection{What the evidence implies for the choice}
\label{sec:sim-choice}

The simulations show a trade-off. When each cohort contains enough independent units or clusters, cohort inference maintains nominal coverage under homogeneous and heterogeneous PT violations. By contrast, aggregate inference can severely under-cover over the identified set defined by the cohort restriction, more so the more heterogeneous the cohort-level block biases. Balanced aggregation narrows the gap but does not close it. When cohorts contain few independent clusters or units, covariance estimation bias degrades coverage even when the restriction holds. For C-LF, this under-coverage is moderate, comparable between cohort and aggregate inference, and partially corrected by CR2. For the cohort FLCI, the under-coverage is severe and is not repaired by CR2.

In practice, I recommend cohort inference with the Conditional or C-LF procedure under the post-only SD or RM restriction as the default. I also recommend choosing between RM and SD based on the empirical context, and plotting the estimated pre-treatment block-bias paths cohort by cohort to assess heterogeneity across cohorts. When these paths differ across cohorts in magnitude or direction, restrictions on the aggregate coefficients conflate cross-cohort heterogeneity with cohort-level violations of PT, and cohort inference is the safer choice. Aggregate inference remains valid when the researcher intends to restrict the aggregate event-study coefficients themselves, though that restriction is less interpretable or transparent. When some cohorts contain only a few units or clusters, pairing the Conditional or C-LF procedure with the bias-reduced CR2 covariance reduces the covariance estimation bias, and the cohort FLCI should be used with caution.

Balanced aggregation is a reasonable alternative when the balanced target is the estimand of interest. Fixing the treated composition removes most of the gap between the two approaches. In Design~II, Table~\ref{tab:A-covering} puts the covering bound at 34 to 49 times $M_C$ for all post-treatment periods and at 1.1 to 1.8 times $M_C$ on the balanced window. The remaining gap comes from the change of comparison group. It grows with the share of later-adopting cohorts in the comparison groups and with the differences in their block biases. The aggregate FLCI is also less exposed to covariance estimation error in small cohorts. With heterogeneous cohort-level block biases, however, balanced aggregation still under-covers as sampling noise decreases.
\section{Empirical applications}
\label{sec:applications}

I apply cohort inference in two empirical settings. The first is the minimum-wage study of Section~\ref{sec:illustration}, with three large cohorts whose pre-treatment trends differ. The second is the community-health-center program of \citet{bailey2015publicmedicine}, with nine cohorts, two of which have only two and six counties. Both applications use the imputation and \CSNYT{} estimators with the county-clustered covariance of Section~\ref{sec:coefvector}. The two estimators produce identical confidence sets under cohort inference but can yield different confidence sets under aggregate inference. Results under aggregate and cohort inference are plotted at the same sensitivity bound $M$.\footnote{I report the width-matching factor $r = M_A^E/M_C$ of Section~\ref{sec:sim-designs} for each target. The width of the plug-in identified set under aggregate inference at $rM$ matches the width of the plug-in identified set under cohort inference at $M$.} The main text focuses on C-LF under the post-only SD and the FLCI under the full-path SD, while Appendix~\ref{app:empirical} collects additional results.

\subsection{Minimum wage and teen employment}
\label{sec:minwage}

The minimum-wage panel from Section~\ref{sec:illustration} consists of treated cohorts $\mathcal G_4$, $\mathcal G_6$, and $\mathcal G_7$ (with 100, 223, and 584 counties) and 1,377 never-treated counties. The target estimand is the average effect across all treated county-years. The sensitivity bound $M$ ranges over $\{0, 0.005, \dots, 0.03\}$ log points per period.\footnote{The width-matching factor is $r = 1.04$.} The conventional point estimates are $-0.045$ (standard error $0.0065$) for imputation and $-0.050$ ($0.0067$) for \CSNYT{}.

Figure~\ref{fig:minwage-sd} shows that cohort inference excludes zero at small SD bounds, whereas the aggregate C-LF confidence set never does. Under cohort inference, the center of the C-LF confidence set remains fixed at $-0.045$ across all $M$. The confidence set excludes zero at $M = 0$ and $M = 0.005$ and first includes zero at $M = 0.01$. The difference comes from the composition change across treated cohorts. Cohort $\mathcal G_7$ has a mild downward pre-treatment block-bias trend, while the earlier-treated cohorts $\mathcal G_4$ and $\mathcal G_6$ have upward pre-treatment block-bias trends. The aggregate coefficients thus conflate these opposing cohort trends as cohort composition shifts across event time. Extrapolating the aggregate pre-treatment slope in the last two pre-treatment periods shifts the center of the aggregate C-LF confidence set toward zero at every $M$.

Figure~\ref{fig:minwage-byperiod} overlays the $M = 0.01$ confidence sets for the effect at each relative period on each estimator's aggregate event-study coefficients. At $s = 1$ cohort and aggregate inference under imputation (panels (a) and (b)) nearly agree, because all three cohorts contribute at adoption and each extrapolation uses the cohorts' own last pre-treatment changes. From $s = 2$ the confidence sets under cohort inference are centered below the conventional estimates and exclude zero at $s = 2$ and $s = 3$, since the rising pre-treatment paths of $\mathcal G_4$ and $\mathcal G_6$ suggest higher counterfactuals. The confidence sets under aggregate inference extrapolate $\mathcal G_7$'s slope into periods that $\mathcal G_7$ never reaches, drift toward zero, and include zero at every $s$. Both confidence sets widen with the extrapolation horizon. The aggregate inference confidence sets also differ between the \CSNYT{} and \Imp{} estimators. 

\begin{figure}[!h]
\centering
\subfigure[Aggregate inference, imputation]{\includegraphics{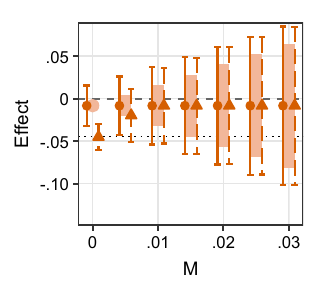}}\hfill
\subfigure[Aggregate inference, \CSNYT]{\includegraphics{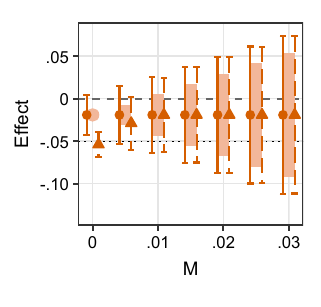}}\hfill
\subfigure[Cohort inference, both estimators]{\includegraphics{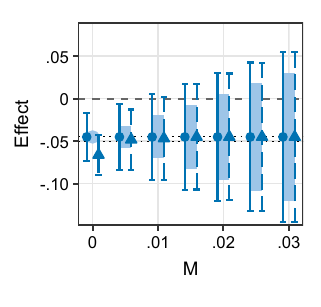}}\\[0.2em]
\includegraphics{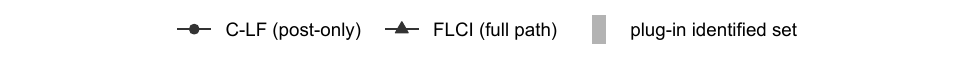}
\caption{Minimum wage: confidence sets under the SD restriction}
\label{fig:minwage-sd}
\vspace{0.5em}
\parbox{\linewidth}{\footnotesize
    \textit{Notes:} 95 percent confidence sets for the average effect over treated county-years, in log points, against the SD bound $M$ in log points per period. At each $M$, the post-only C-LF confidence set (solid whisker, circle) and the full-path FLCI (dashed whisker, triangle) are shown with markers at their centers; the shaded band is the post-only plug-in identified set; the dashed line marks zero and the dotted lines the conventional estimates.
}
\end{figure}

\begin{figure}[!h]
\centering
\subfigure[Aggregate inference, imputation]{\includegraphics{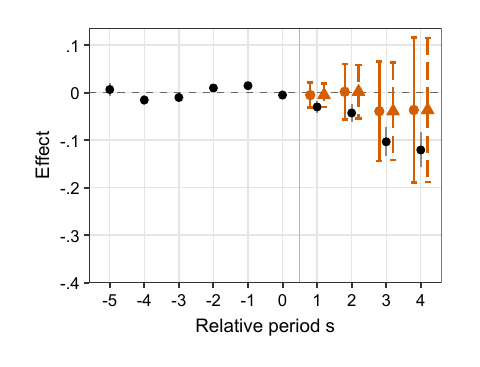}}\hfill
\subfigure[Cohort inference, imputation path]{\includegraphics{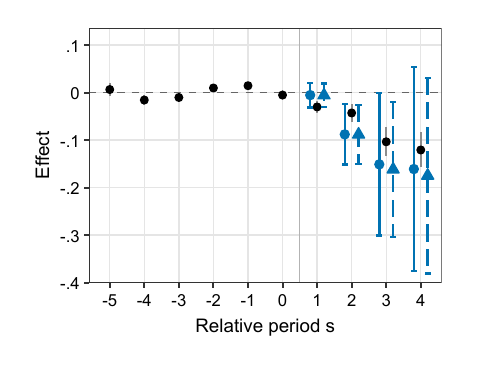}}\\[0.2em]
\subfigure[Aggregate inference, \CSNYT]{\includegraphics{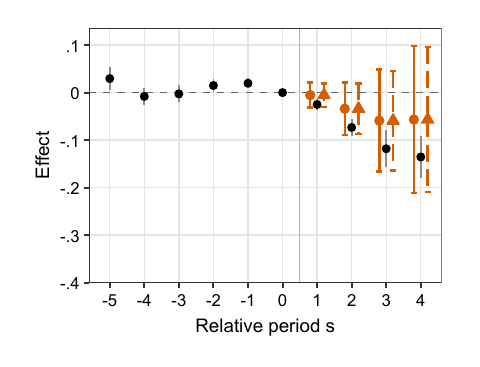}}\hfill
\subfigure[Cohort inference, \CSNYT{} path]{\includegraphics{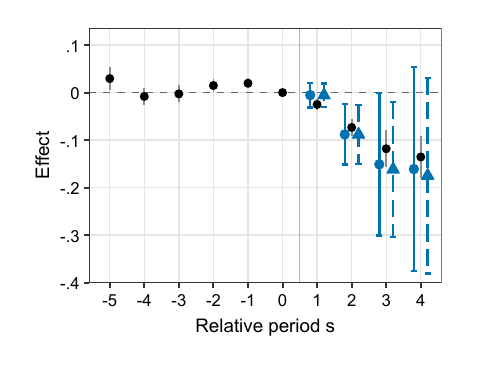}}\\[0.2em]
\includegraphics{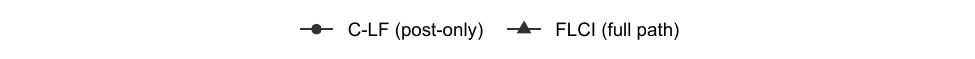}
\caption{Minimum wage: confidence sets by relative period at $M = 0.01$}
\label{fig:minwage-byperiod}
\vspace{0.5em}
\parbox{\linewidth}{\footnotesize
    \textit{Notes:} 95 percent confidence sets for the effect at relative periods $s = 1, \dots, 4$ under the SD restriction with $M = 0.01$ log points per period. Black points: the aggregate event-study coefficients of the estimator named in the panel, with 95 percent pointwise intervals; the vertical line separates pre-treatment from post-treatment periods; markers as in Figure~\ref{fig:minwage-sd}. Confidence sets are pointwise; those under cohort inference in (b) and (d) are the same.
}
\end{figure}

Under cohort inference the full-path FLCI is 18 percent shorter than C-LF at $M = 0$ and 9 percent shorter at $M = 0.005$. Its pre-treatment adjustment $\gamma'\hat z$ moves the center from the post-only center of $-0.045$ to $-0.066$ at $M = 0$. These intervals should be read with caution. The coverage guarantee of the full-path FLCI requires the full-path SD bound to hold for the second differences that carry weight $\gamma$. The pre-treatment data reject the full-path restriction jointly at every bound below $0.024$,\footnote{Across cohorts, the largest fitted pre-treatment second difference is $0.069$ log points per period and the mean absolute second difference is $0.027$. The joint test is the 5 percent max test over all pure pre-treatment second differences of Appendix~\ref{app:emp-minwage}.} so the question at each bound is whether the rejected constraints are among the pre-treatment second differences that carry weight.

At $M = 0$ they are: the adjustment uses all eight pure pre-treatment second differences. The narrower interval therefore reflects a rejected constraint rather than a tighter interval with a coverage guarantee. From $M = 0.005$ on, only the single second difference of $\mathcal G_4$ carries weight (Figure~\ref{fig:app-minwage-gamma}). Its fitted value, $-0.015$, exceeds the bounds $0.005$ and $0.01$ in absolute value by less than one standard error ($0.021$). The intervals at those bounds rest on one constraint that the data neither reject nor support. At those bounds the coverage guarantee is therefore not contradicted by the data, unlike at $M = 0$.

Under aggregate inference in panels~(a) and~(b), the FLCI excludes zero at $M = 0$, where every aggregate pre-treatment second difference carries weight. The aggregate paths are smoother than the cohort paths, with largest second differences of $0.027$ for imputation and $0.043$ for \CSNYT{}, but the joint test still rejects every bound below $0.004$ and $0.015$, respectively. The exclusion at $M = 0$ therefore rests on rejected constraints, as the cohort FLCI's does.

These results use the county-clustered covariance of the original study. Minimum wages change at the state level, and $\mathcal G_4$ lies in a single state, so a covariance that allows for common state shocks cannot be estimated for it. In Appendix~\ref{app:emp-minwage}, I drop $\mathcal G_4$ and conduct the cohort inference on the sub-sample using state-clustered covariance. Figure~\ref{fig:app-minwage-state} shows the results.

Appendix~\ref{app:emp-minwage} adds the cohort-specific confidence sets (Figure~\ref{fig:app-minwage-bycohort}), the C-LF confidence sets under RM (Figure~\ref{fig:minwage-rm}), the pure pre-treatment second differences behind the benchmarks (Table~\ref{tab:app-minwage-curvature}), and the FLCI weights on them (Figure~\ref{fig:app-minwage-gamma}).

\subsection{Community health centers}
\label{sec:chc}

The second application uses the county panel of \citet{bailey2015publicmedicine}, which tracks age-adjusted mortality (in deaths per 100,000) across 3,044 counties from 1959 to 1988 as community health centers open. The 114 treated counties form nine adoption cohorts in 1965 and 1967 to 1974 of two to 28 counties each, with only two counties in the 1965 cohort and six in the 1973 cohort. Both estimators use population weights, and the target estimand is the average effect over the balanced window $s = 1, \dots, 15$ observed for all nine cohorts. The sensitivity bound $M$ takes values in $\{0, 0.33, 0.83, 1.66, 3.32\}$ deaths per 100,000 per period.\footnote{The values are $0$ to $1$ times $3.32$, the aggregate imputation standard error at $s = 1$. For the aggregate path, the max test of Appendix~\ref{app:emp-minwage} does not reject the full-path restriction even at $M = 0$. The width-matching factor is $r = 1.07$. I omit RM, whose union has too many components to compute with nine cohorts and up to 14 pre-treatment periods.} The conventional point estimates are $-12.4$ ($4.1$) for imputation and $-7.2$ ($3.5$) for \CSNYT{}.

\begin{figure}[!h]
\centering
\subfigure[Aggregate inference, imputation]{\includegraphics{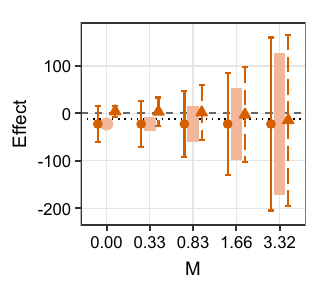}}\hfill
\subfigure[Aggregate inference, \CSNYT]{\includegraphics{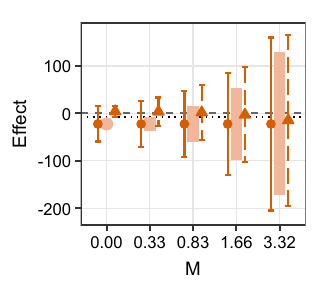}}\hfill
\subfigure[Cohort inference, both estimators]{\includegraphics{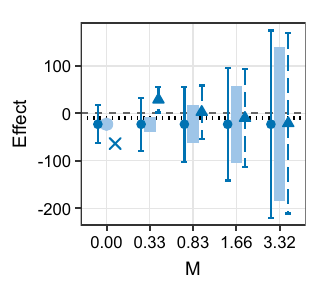}}\\[0.2em]
\includegraphics{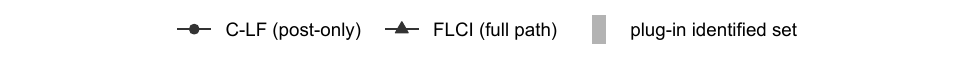}
\caption{Community health centers: confidence sets under the SD restriction}
\label{fig:chc}
\vspace{0.5em}
\parbox{\linewidth}{\footnotesize
    \textit{Notes:} 95 percent confidence sets for the average effect on age-adjusted mortality over the balanced window $s = 1, \dots, 15$, in deaths per 100,000, at equal spacing across $M \in \{0, 0.33, 0.83, 1.66, 3.32\}$ deaths per 100,000 per period. Panels, markers, band, and lines as in Figure~\ref{fig:minwage-sd}. In panel~(c) at $M = 0$, the cross marks the FLCI center under the adjustment with the smallest norm of $\gamma$, and its interval is omitted (Appendix~\ref{app:emp-chc}).
}
\end{figure}

Figure~\ref{fig:chc} shows that under the post-only SD restriction, the C-LF confidence sets include zero at every bound, even at $M = 0$. They are wide because averaging over 15 post-treatment periods extrapolates pre-treatment slopes over a long horizon. Unlike in the minimum-wage application, the centers of the cohort and aggregate C-LF confidence sets agree within one death per 100,000 in all three panels. One reason is that I am using the balanced window in which every cohort is observed.

The full-path cohort FLCI behaves erratically at small bounds because the two smallest cohorts drive the swings in its pre-treatment adjustment $\gamma'\hat z$. In panel~(c), the center of the cohort FLCI jumps to $28.8$ deaths per 100,000 at $M = 0.33$, with an interval entirely above zero. At $M = 0.83$, it drops back to 2.7, nearly matching the aggregate FLCI centers. The noisy fitted pre-treatment paths of the two-county 1965 cohort and the six-county 1973 cohort have maximum second differences of $85$ and $139$ deaths per 100,000. The fitted paths therefore violate the full-path restriction across the entire plotted grid. Yet the FLCI algorithm puts weight on them, which moves the center at small $M$ (Figure~\ref{fig:app-chc-mechanism}). Consistent with Design~IV in Section~\ref{sec:sim-support}, this instability of cohort FLCI shows the risk of the full-path cohort FLCI when some cohorts contain only a few units. Recomputing every set with the bias-reduced CR2 covariance widens the C-LF confidence sets by 2 to 11 percent but leaves their centers and the instability of the cohort FLCI unchanged.

Appendix~\ref{app:emp-chc} adds the adoption timing (Figure~\ref{fig:app-chc-timing}), the cohort-level mechanism behind the cohort FLCI (Figure~\ref{fig:app-chc-mechanism}), and the C-LF confidence sets by relative period (Figure~\ref{fig:app-chc-byperiod}).
\section{Conclusion}
\label{sec:conclusion}

This paper places the sensitivity restrictions of \RR{} on each cohort's block bias, its untreated trend relative to its initial control group, rather than on the aggregate coefficients. A known linear map carries block biases into the coefficient biases, so \RR's procedures apply directly. The imputation and \CSNYT{} estimators give the same confidence sets under cohort inference. A cohort-by-cohort restriction need not imply the same restriction on the aggregate path, and the two approaches disagree a lot when the cohorts' block biases are heterogeneous. In practice, I recommend cohort inference with the Conditional or C-LF procedure under post-only SD or RM, paired with a bias-reduced covariance when cohorts are small. The FLCI procedure requires a stronger restriction and should be used with caution. The accompanying R package \texttt{CohortHonestDiD} implements both approaches.

The framework assumes a balanced panel. With missing outcomes, the units behind a cohort's cell means change over time, and the block bias would have to be defined on the units observed in each period. With time-varying covariates, the block bias in each period averages the conditional violations over a different covariate distribution, so pre-treatment block biases can no longer benchmark post-treatment ones over the same population. These are left for future research.

\clearpage
\appendix
\renewcommand{\theequation}{\thesection.\arabic{equation}}
\renewcommand{\thefigure}{\thesection.\arabic{figure}}
\renewcommand{\thetable}{\thesection.\arabic{table}}
\renewcommand{\thelemma}{\thesection.\arabic{lemma}}
\setcounter{lemma}{0}
\section{Proofs and derivations}
\label{app:proofs}
\setcounter{equation}{0}
\setcounter{table}{0}

Throughout, $t = t_g + s - 1$ for a cell $(\mathcal G_g, s)$, $N_A$ is the number of units in a set $A$, and $\mathcal U_t = (\cup_{k : t_k > t}\mathcal G_k) \cup \mathcal G_\infty$ denotes the units not yet treated in period $t$, so that $\mathcal C_{g,s} = \mathcal U_t$ and $\mathcal C_{g,1} = \mathcal U_{t_g}$. For a set of units $A$ let
\begin{equation}
m_A(t) = \E\bigl[Y_{it}(\infty) \mid i \in A\bigr], \qquad \bar m_A(P) = \frac{1}{|P|}\sum_{r \in P} m_A(r),
\label{eq:A-means}
\end{equation}
be the mean untreated potential outcome in period $t$ and its average over a set of periods $P$. For a union of cohorts, $m_A(t)$ is the size-weighted average of the cohort means. The trend difference between two sets of units from a baseline $P$ is
\begin{equation}
\mathrm{TD}_{A,B}(t; P) = \bigl[m_A(t) - \bar m_A(P)\bigr] - \bigl[m_B(t) - \bar m_B(P)\bigr],
\label{eq:A-TD}
\end{equation}
and I write $\mathrm{TD}_{A,B}(t; r)$ when $P = \{r\}$. In this notation the block biases of Section~\ref{sec:blockbias} are $\Delta^{\Imp}_{\mathcal G_g,s} = \mathrm{TD}_{g,\mathcal C_{g,1}}(t; \pre_g)$ and $\Delta^{\CSNYT}_{\mathcal G_g,s} = \mathrm{TD}_{g,\mathcal C_{g,1}}(t; t_g - 1)$, where a cohort index in a subscript stands for the cohort.

\subsection{Proof of Lemma~\ref{lemma:first-period}}
\label{app:lemma-first-period}

\begin{proof}
The imputation fit minimizes $\sum (Y_{it} - \alpha_i - \phi_t)^2$ over the untreated observations. Its residuals $\hat\varepsilon_{it} = Y_{it} - \hat\alpha_i - \hat\phi_t$ satisfy the first-order conditions for the unit and period effects,
\begin{equation}
\sum_{t \in \pre_g} \hat\varepsilon_{it} = 0 \ \text{ for } i \in \mathcal G_g, \qquad \sum_{t = 1}^{T} \hat\varepsilon_{it} = 0 \ \text{ for } i \in \mathcal G_\infty, \qquad \sum_{i \in \mathcal U_t} \hat\varepsilon_{it} = 0 \ \text{ for } t = 1, \dots, T.
\label{eq:A-foc}
\end{equation}
The effects are determined only up to a common constant. The fitted values $\hat\alpha_i + \hat\phi_t$ are nevertheless unique, because every cohort has a pre-treatment period and the never-treated cohort is nonempty (Assumption~\ref{ass:design}). Fix a cohort $\mathcal G_g$ and a unit $i \in \mathcal G_g$. The unit condition gives $\hat\alpha_i = \bar Y_{i,\pre_g} - \bar{\hat\phi}_{\pre_g}$, with $\bar{\hat\phi}_{\pre_g}$ the average of $\hat\phi_t$ over $\pre_g$, so that
\begin{equation}
\hat Y_{i,t_g}(\infty) = \hat\alpha_i + \hat\phi_{t_g} = \bar Y_{i,\pre_g} + \bigl(\hat\phi_{t_g} - \bar{\hat\phi}_{\pre_g}\bigr).
\label{eq:A-imputed}
\end{equation}
In period $t_g$ the untreated units are exactly $\mathcal C_{g,1}$, so the period condition gives $\hat\phi_{t_g} = \bar Y_{\mathcal C_{g,1},t_g} - \bar{\hat\alpha}_{\mathcal C_{g,1}}$, with $\bar{\hat\alpha}_{\mathcal C_{g,1}}$ the average unit effect over $\mathcal C_{g,1}$. Now sum the period conditions over $t \in \pre_g$. Split the untreated units of each period into those in $\mathcal C_{g,1}$, which are untreated throughout $\pre_g$, and those in cohorts $\mathcal G_l$ with $t_l \leq t_g$, including $\mathcal G_g$ itself, which are untreated in a period $t < t_g$ exactly when $t < t_l$:
\begin{equation}
0 = \sum_{t \in \pre_g}\ \sum_{i \in \mathcal U_t} \hat\varepsilon_{it} = \sum_{i \in \mathcal C_{g,1}}\ \sum_{t \in \pre_g} \hat\varepsilon_{it} + \sum_{l : t_l \leq t_g}\ \sum_{i \in \mathcal G_l}\ \sum_{t \in \pre_l} \hat\varepsilon_{it}.
\label{eq:A-sum}
\end{equation}
Every inner sum of the last term is zero by the unit conditions, so $\sum_{i \in \mathcal C_{g,1}}\sum_{t \in \pre_g}\hat\varepsilon_{it} = 0$, that is, $\bar{\hat\phi}_{\pre_g} = \bar Y_{\mathcal C_{g,1},\pre_g} - \bar{\hat\alpha}_{\mathcal C_{g,1}}$. Subtracting this from the expression for $\hat\phi_{t_g}$ gives $\hat\phi_{t_g} - \bar{\hat\phi}_{\pre_g} = \bar Y_{\mathcal C_{g,1},t_g} - \bar Y_{\mathcal C_{g,1},\pre_g}$, and \eqref{eq:A-imputed} becomes the statement of the lemma.
\end{proof}

\subsection{Proofs of Propositions~\ref{prop:imputation} and~\ref{prop:csnyt}}
\label{app:correspondence}

\citet{aguilar2025comparative} indexes cohorts by adoption period, $E_i \in \{2, \dots, T, \infty\}$, and uses event time, written $\ell \geq 0$ there. In this subsection $\ell$ is that event time, not the target weights of Section~\ref{sec:setup}. I use the same letter for an adoption period and for the cohort adopting then, so that $t_h = h$. The cell $(g, \ell)$ there corresponds to $(\mathcal G_g, s)$ here, with $t_g = g$ and $s = \ell + 1$, and $ATT(g, \ell)$ is $\tau_{\mathcal G_g,s}$. Theorem~1 of that paper expresses the extended two-way fixed effects estimand as a weighted average of differences-in-differences of population means, one term built from post-treatment comparisons and one from pre-treatment comparisons. This estimand equals the imputation estimand \citep{wooldridge2021two, borusyak2024revisiting}. Its Corollaries~1 and~3 give the corresponding bias decompositions for imputation and \CSNYT{} under no anticipation.

Random sampling gives population-share weights, $\pi_{h,g+\ell} = \mathbb{P}(E_i = h \mid E_i > g + \ell)$ and $\tilde\pi_h = (h - 1)^{-1}\mathbb{P}(E_i = h \mid E_i \geq h)$. Under the fixed design of Assumption~\ref{ass:design}(iv), these become count shares $N_h / N_{\mathcal U_t}$ and $w_h / (t_h - 1)$, where $w_h$ is defined in \eqref{eq:weights}. \citet{aguilar2025comparative} obtains Theorem~1 in the same way, as a sample identity in cell means and count shares before passing to population shares.

The weighted common trend difference $WCTD(g, j, t^*, b)$ of that paper compares group $g$'s untreated trend with the share-weighted untreated trends of groups not yet treated at $g + j$, measured from the average over periods $1, \dots, b$ to period $t^*$. The two instances needed below have fourth argument $t_g - 1$, and in the second the third argument is the running period $r$ of the sum in \eqref{eq:A-cor1}. In the notation of \eqref{eq:A-TD}, $WCTD(g, \ell, t, t_g - 1) = \mathrm{TD}_{g, \mathcal U_t}(t; \pre_g)$ and $WCTD(k, 0, r, t_g - 1) = \mathrm{TD}_{k, \mathcal C_{k,1}}(r; \pre_g)$. Corollary~1 therefore becomes \eqref{eq:A-cor1}, and Corollary~3 becomes \eqref{eq:A-cor3}, which is \eqref{eq:delta-csnyt} with $\mathcal C_{g,s} = \mathcal U_t$:
\begin{align}
\delta^{\Imp}_{\mathcal G_g,s} &= \mathrm{TD}_{g,\mathcal U_t}(t; \pre_g) - \sum_{k \in \mathcal K_g(t)} \frac{w_k}{t_k - 1} \sum_{r = t_g}^{t_k - 1} \mathrm{TD}_{k, \mathcal C_{k,1}}(r; \pre_g),
\label{eq:A-cor1}\\
\delta^{\CSNYT}_{\mathcal G_g,s} &= \mathrm{TD}_{g,\mathcal U_t}(t; t_g - 1).
\label{eq:A-cor3}
\end{align}
The first term of \eqref{eq:A-cor1} compares cohort $\mathcal G_g$ with units not yet treated at $t$, using its baseline $\pre_g$. This is the post-treatment term of the theorem. The second term corrects for each adjustment cohort $\mathcal G_k$: it subtracts $w_k/(t_k - 1)$ times the sum of that cohort's trend differences from its initial control group over $t_g, \dots, t_k - 1$, when $\mathcal G_g$ is treated but $\mathcal G_k$ is not, still using the baseline $\pre_g$. This is the theorem's pre-treatment term.

Both statements, \eqref{eq:A-cor1} and \eqref{eq:A-cor3}, compare cohort $\mathcal G_g$ with the group not yet treated at $t$, from $\mathcal G_g$'s baseline throughout. Propositions~\ref{prop:imputation} and~\ref{prop:csnyt} instead compare every cohort with its initial control group from its own baseline. Two algebraic identities connect the two forms: a telescoping identity for the not-yet-treated group, stated as Lemma~\ref{lemma:telescoping}, and a change of baseline, stated in two versions in \eqref{eq:A-rebase} and \eqref{eq:A-rebase-cs}.

\begin{lemma}[Telescoping]
\label{lemma:telescoping}
Let $d$ assign a number $d_A$ to every cohort, extended to unions of cohorts by size-weighted averaging. For every cohort $\mathcal G_g$ and period $t \geq t_g$,
\begin{equation}
d_{\mathcal C_{g,1}} - d_{\mathcal U_t} = \sum_{k \in \mathcal K_g(t)} w_k \bigl(d_{\mathcal G_k} - d_{\mathcal C_{k,1}}\bigr).
\label{eq:A-telescope}
\end{equation}
\end{lemma}

\begin{proof}
Order the adjustment cohorts $\mathcal K_g(t) = \{k_1 < \cdots < k_m\}$ by adoption period and put $\mathcal C_{k_0,1} = \mathcal C_{g,1}$. Because $\mathcal K_g(t)$ contains every cohort that adopts in $(t_g, t]$ and adoption periods are distinct, no cohort adopts between consecutive members, so for $j = 1, \dots, m$ the initial control group of the preceding cohort is the disjoint union $\mathcal C_{k_{j-1},1} = \mathcal G_{k_j} \cup \mathcal C_{k_j,1}$. In this union $\mathcal G_{k_j}$ has the share $w_{k_j}$ of \eqref{eq:weights}, so size-weighted averaging gives $d_{\mathcal C_{k_{j-1},1}} - d_{\mathcal C_{k_j,1}} = w_{k_j}(d_{\mathcal G_{k_j}} - d_{\mathcal C_{k_j,1}})$. Since $\mathcal G_{k_m}$ is the last cohort to adopt by $t$, $\mathcal C_{k_m,1} = \mathcal U_t$, and summing over $j$ gives \eqref{eq:A-telescope}. If $\mathcal K_g(t)$ is empty, $\mathcal U_t = \mathcal C_{g,1}$ and both sides vanish.
\end{proof}

The second identity, the change of baseline, has one version for the multi-period baselines $\pre_g$ and $\pre_k$ in \eqref{eq:A-rebase} and one for the single-period baselines $t_g - 1$ and $t_k - 1$ in \eqref{eq:A-rebase-cs}: for a cohort $\mathcal G_k$ with $t_k > t_g$ and any sets of units $A$ and $B$,
\begin{align}
\mathrm{TD}_{A,B}(t; \pre_k) &= \mathrm{TD}_{A,B}(t; \pre_g) - \frac{1}{t_k - 1}\sum_{r = t_g}^{t_k - 1} \mathrm{TD}_{A,B}(r; \pre_g),
\label{eq:A-rebase}\\
\mathrm{TD}_{A,B}(t; t_g - 1) &= \mathrm{TD}_{A,B}(t; t_k - 1) - \mathrm{TD}_{A,B}(t_g - 1; t_k - 1).
\label{eq:A-rebase-cs}
\end{align}
The first holds because $\bar m_A(\pre_k) - \bar m_A(\pre_g) = (t_k - 1)^{-1}\sum_{r = t_g}^{t_k - 1}[m_A(r) - \bar m_A(\pre_g)]$, the periods before $t_g$ contributing zero to the sum; the second because $\mathrm{TD}_{A,B}(t; r) - \mathrm{TD}_{A,B}(t'; r) = \mathrm{TD}_{A,B}(t; t')$ for single-period baselines.

Proposition~\ref{prop:imputation} follows from \eqref{eq:A-cor1}. Lemma~\ref{lemma:telescoping} with $d_A = m_A(t) - \bar m_A(\pre_g)$ gives
\[
\mathrm{TD}_{g,\mathcal U_t}(t; \pre_g) = \mathrm{TD}_{g,\mathcal C_{g,1}}(t; \pre_g) + \sum_{k \in \mathcal K_g(t)} w_k\, \mathrm{TD}_{k,\mathcal C_{k,1}}(t; \pre_g).
\]
The first term is $\Delta^{\Imp}_{\mathcal G_g,s}$, and \eqref{eq:A-rebase} with $A = \mathcal G_k$ and $B = \mathcal C_{k,1}$ combines each adjustment term with its pre-treatment correction in \eqref{eq:A-cor1} into $w_k\, \mathrm{TD}_{k,\mathcal C_{k,1}}(t; \pre_k) = w_k \Delta^{\Imp}_{\mathcal G_k, s_k(t)}$. The result is \eqref{eq:decomp-imp}.

Proposition~\ref{prop:csnyt} follows from \eqref{eq:A-cor3} in the same way: Lemma~\ref{lemma:telescoping} with $d_A = m_A(t) - m_A(t_g - 1)$ gives
\[
\mathrm{TD}_{g,\mathcal U_t}(t; t_g - 1) = \Delta^{\CSNYT}_{\mathcal G_g,s} + \sum_{k \in \mathcal K_g(t)} w_k\, \mathrm{TD}_{k,\mathcal C_{k,1}}(t; t_g - 1),
\]
and \eqref{eq:A-rebase-cs} with the same $A$ and $B$ writes $\mathrm{TD}_{k,\mathcal C_{k,1}}(t; t_g - 1)$ as $\Delta^{\CSNYT}_{\mathcal G_k, s_k(t)} - \Delta^{\CSNYT}_{\mathcal G_k, s_k(t_g - 1)}$. The result is \eqref{eq:decomp-csnyt}.

In words, the adjustment terms of the two propositions are what remains of a comparison with the not-yet-treated control group once the comparison group is held fixed at $\mathcal C_{g,1}$, and the correction terms move each adjustment cohort's baseline to the one that defines its own block bias.

\paragraph*{Sample identities.}
Both sides of \eqref{eq:decomp-imp} and \eqref{eq:decomp-csnyt} are linear forms in the vector that collects the means $m_{\mathcal G_h}(r)$ of \eqref{eq:A-means} over all cohort-period cells: the right sides by construction, and the left sides because each estimator is linear in the cell means (Section~\ref{app:covariance}) and, apart from the cohort's own cell in period $t$, uses only untreated observations, so that $\delta$ is the same linear function of the population means. Under Assumptions~\ref{ass:design} and~\ref{ass:nointerference} nothing restricts these means, so the identities hold for every value of the mean vector and therefore coefficient by coefficient. The treated cells of the adjustment cohorts do not enter the left sides, so they have coefficient zero on the right sides as well. The identities then hold for any numbers substituted for the means, among them the sample cell means of observed outcomes; there the left side is the estimate $\hat\tau_{\mathcal G_g,s}$ and the right side its exact count-share decomposition into observed-outcome analogues of the block biases, under the fixed design of Assumption~\ref{ass:design}(iv). Section~\ref{app:equivalence} uses these sample identities to prove the plug-in identified-set equivalence stated in Section~\ref{sec:equivalence}.

\paragraph*{The sequential representation.}
\citet{arkhangelsky2024sequential} show that imputation-type estimators can be computed sequentially: the untreated outcomes are completed one relative period at a time, and the counterfactuals imputed in earlier rounds serve as data in later rounds. Read in this way, the imputed mean of cohort $\mathcal G_g$ in period $t$ is a block difference-in-differences between the cohort and its initial control group in which the adjustment cohorts $\mathcal K_g(t)$ enter through the cells already imputed for them, so that the estimation bias of cohort $\mathcal G_g$ inherits a share $N_k / N_{\mathcal C_{g,1}}$ of the estimation bias of each adjustment cohort in the same period, $\delta_{\mathcal G_g,s} = \Delta_{\mathcal G_g,s} + \sum_{k \in \mathcal K_g(t)} (N_k / N_{\mathcal C_{g,1}})\, \delta_{\mathcal G_k, s_k(t)}$. Solving this recursion from the latest adjustment cohort backward reproduces the weights $w_k$: the direct share $N_k / N_{\mathcal C_{g,1}}$ and the shares transmitted through the intermediate cohorts add up to $w_k$, which is the identity of Lemma~\ref{lemma:telescoping} written as a recursion. This recursion is another route to \eqref{eq:decomp-imp}.

\subsection{Influence functions and the clustered covariance}
\label{app:covariance}

Both estimators are linear in the outcomes. Stack the $NT$ observed outcomes in $\vec Y$ and write
\begin{equation}
\vec{\hat\beta} = \mathbf L\vec Y,
\label{eq:A-operator}
\end{equation}
where the rows of $\mathbf L$ depend only on the design: cohort membership, adoption periods, and cohort sizes. For \CSNYT{} the row of a post-treatment cell puts $1/N_g$ on the cohort's outcomes in period $t$ and $-1/N_g$ on those in period $t_g - 1$, and $-1/N_{\mathcal C_{g,s}}$ and $+1/N_{\mathcal C_{g,s}}$ on the comparison group's outcomes in the same two periods. The pre-treatment rows \eqref{eq:pre-csnyt} replace $\mathcal C_{g,s}$ by $\mathcal C_{g,1}$. For the imputation estimator the pre-treatment rows are \eqref{eq:pre-imp}. The untreated fit predicts $\hat Y_{it}(\infty) = \vec h_{it}'\vec Y_{\mathcal O}$, where $\vec Y_{\mathcal O}$ stacks the untreated observations and $\vec h_{it}$ is the row of the least-squares prediction operator for the observation $(i,t)$, a function of the design alone. The post-treatment row of cell $(\mathcal G_g, s)$ puts $1/N_g$ on the cohort's outcomes in period $t$ and subtracts the cohort average of these prediction weights. Permuting the units within a cohort changes neither the design nor any cohort-period estimate, because the untreated fit is symmetric in units with the same untreated periods. A linear function that is invariant to such permutations puts the same coefficient on every unit of a cell, so each row of $\mathbf L$ is constant within cohort-period cells, and $\vec{\hat\beta}$ is a linear function of the cohort-period cell means. The proofs below use this fact repeatedly.

I estimate the covariance by treating the cohort-period cell means as unrestricted parameters. Conditional on the cohort partition and sizes (Assumption~\ref{ass:design}), the mean outcome in cohort $\mathcal G_g$ is $m_{\mathcal G_g}(t)$ plus its cell treatment effect. The residual $\hat e_{it} = Y_{it} - \bar Y_{\mathcal G_{g(i)}, t}$, where $g(i)$ is unit $i$'s cohort, estimates the deviation from that mean. It includes unit-level heterogeneity. The sampling variation in Assumption~\ref{ass:normality} comes from variation in these deviations across independent clusters.

Group observations into $G_{\mathsf c}$ clusters: in the applications the clusters are counties, which are the units; in the simulations they are units or states. Let $\mathbf L_{\mathsf c}$ contain the columns of $\mathbf L$ for cluster $\mathsf c$, and let $\vec{\hat e}_{\mathsf c}$ collect its residuals. Its score is $\vec{\hat s}_{\mathsf c} = \mathbf L_{\mathsf c}\vec{\hat e}_{\mathsf c}$. The applications and panel simulations in Section~\ref{sec:simulations} use
\begin{equation}
\hat\Sigma_N = a_N\sum_{\mathsf c = 1}^{G_{\mathsf c}}\vec{\hat s}_{\mathsf c}\vec{\hat s}_{\mathsf c}', \qquad a_N = \frac{G_{\mathsf c}}{G_{\mathsf c} - 1}\cdot\frac{n - 1}{n - p},
\label{eq:A-vcov}
\end{equation}
with $n = NT$ observations and $p$ the number of outcome cells, counting all cohorts including the never-treated cohort and all periods including reference periods. Here $a_N$ is the usual finite-sample factor of the cluster-robust covariance applied to the cell-mean model. Since $\vec{\hat\beta} = \sum_{\mathsf c} \mathbf L_{\mathsf c}\vec Y_{\mathsf c}$ and clusters are independent, $\sum_{\mathsf c} \vec{\hat s}_{\mathsf c}\vec{\hat s}_{\mathsf c}'$ is the influence-function (sandwich) estimator of $\Sigma_N$. I apply it identically to both estimators and to the pre- and post-treatment coefficients jointly, as a joint restriction on $\vec\Delta$ requires.

Cell-mean centering has two consequences. First, each row of $\mathbf L$ is constant within cells and the residuals sum to zero within cells. The cluster scores therefore sum to zero, $\sum_{\mathsf c} \vec{\hat s}_{\mathsf c} = \mathbf L\vec{\hat e} = 0$, and $\hat\Sigma_N$ has rank at most $G_{\mathsf c} - 1$, which matters when there are few clusters.

Second, a cell contained entirely within one cluster contributes zero to that cluster's score. If a cluster is an entire cohort, as when one state forms a treatment cohort, its score is zero, and the state-clustered covariance omits that cohort's sampling variation.

\paragraph*{Relation to the did package.}
The analytical standard errors of the \texttt{did} package \citep{callaway2021-did}, computed from its influence functions, have the form of \eqref{eq:A-vcov} with the unit as the cluster and two differences. First, the long differences of the comparison group $\mathcal C_{g,s}$ are centered at the pooled mean of the comparison group rather than at the cohort-period cell means, so that differences in mean trends across the not-yet-treated cohorts enter the residuals. Second, no finite-sample factor is applied. The sensitivity analyses of Sections~\ref{sec:simulations} and~\ref{sec:applications} use \eqref{eq:A-vcov} for both estimators so that their confidence sets are comparable. For cohort inference they are then identical (Proposition~\ref{prop:ci-equivalence}).

\subsection{Proof of Proposition~\ref{prop:equivalence}}
\label{app:equivalence}

\begin{proof}
I write the proof on the full set of cells of Section~\ref{sec:coefvector}. The normalized cells satisfy $\Delta^{\CSNYT}_{\mathcal G_g,0} = 0$ and $\sum_{s \leq 0}\Delta^{\Imp}_{\mathcal G_g,s} = 0$ for the population paths and for their estimates alike. Let $\vec\Delta^{\Imp}$ and $\vec\Delta^{\CSNYT}$ be the true block-bias paths, $\vec\beta^{\Imp}$ and $\vec\beta^{\CSNYT}$ the population coefficient vectors, $\Lambda_\Delta$ either restriction, and let $F^{\Imp} = \{\vec\Delta \in \Lambda_\Delta : \vec\Delta_{\pre} = \vec\beta^{\Imp}_{\pre}\}$ and $F^{\CSNYT} = \{\vec\Delta \in \Lambda_\Delta : \vec\Delta_{\pre} = \vec\beta^{\CSNYT}_{\pre}\}$ be the feasible sets in \eqref{eq:identified-set}. Step~1 establishes the cohort constants that Section~\ref{sec:equivalence} takes from Section~\ref{sec:blockbias}; Steps~2 and~3 are the two observations stated there.

\emph{Step 1: the two block-bias paths differ by cohort constants.} In the notation of \eqref{eq:A-TD}, $\Delta^{\Imp}_{\mathcal G_g,s} - \Delta^{\CSNYT}_{\mathcal G_g,s} = \mathrm{TD}_{g,\mathcal C_{g,1}}(t; \pre_g) - \mathrm{TD}_{g,\mathcal C_{g,1}}(t; t_g - 1) = \mathrm{TD}_{g,\mathcal C_{g,1}}(t_g - 1; \pre_g) = \Delta^{\Imp}_{\mathcal G_g,0} =: \varkappa_g$ for every $s$, because the two baselines differ by a quantity that does not involve period $t$. Let $\vec\varkappa$ carry $\varkappa_g$ in every cell of cohort $\mathcal G_g$; then $\vec\Delta^{\CSNYT} = \vec\Delta^{\Imp} - \vec\varkappa$.

\emph{Step 2: the feasible sets correspond.} Both restrictions of Section~\ref{sec:restrictions} are defined through differences $\Delta_{\mathcal G_g,s} - \Delta_{\mathcal G_g,s-1}$ along a cohort's own path, including the benchmark in \eqref{eq:rm}, and adding a vector that is constant within cohorts leaves every such difference unchanged. Hence $\vec\Delta \in \Lambda_\Delta$ if and only if $\vec\Delta - \vec\varkappa \in \Lambda_\Delta$. Since $\vec\beta_{\pre} = \vec\Delta_{\pre}$ for each estimator, Step~1 gives $\vec\beta^{\Imp}_{\pre} - \vec\varkappa_{\pre} = \vec\beta^{\CSNYT}_{\pre}$, and $\vec\Delta \mapsto \vec\Delta - \vec\varkappa$ is a bijection from $F^{\Imp}$ onto $F^{\CSNYT}$.

\emph{Step 3: corresponding paths give the same corrected coefficients.} By Propositions~\ref{prop:imputation} and~\ref{prop:csnyt}, $\vec\beta_{\post} = \vec\tau_{\post} + (\mathbf W\vec\Delta)_{\post}$ for each estimator with its own map and true path; the two maps share $\mathbf W_{\post,\post}$, have identity pre-treatment blocks, and $\mathbf W^{\Imp}_{\post,\pre} = 0$ (Section~\ref{sec:biasdecomposition} and \eqref{eq:flci-blocks}). For $\vec\Delta \in F^{\Imp}$ and $\vec\Delta' = \vec\Delta - \vec\varkappa \in F^{\CSNYT}$,
\begin{align}
\vec\beta^{\Imp}_{\post} - (\mathbf W^{\Imp}\vec\Delta)_{\post} &= \vec\tau_{\post} + \mathbf W_{\post,\post}\bigl(\vec\Delta^{\Imp}_{\post} - \vec\Delta_{\post}\bigr),
\label{eq:A-corrected-imp}\\
\vec\beta^{\CSNYT}_{\post} - (\mathbf W^{\CSNYT}\vec\Delta')_{\post} &= \vec\tau_{\post} + \mathbf W^{\CSNYT}_{\post,\pre}\bigl(\vec\Delta^{\CSNYT}_{\pre} - \vec\Delta'_{\pre}\bigr) + \mathbf W_{\post,\post}\bigl(\vec\Delta^{\CSNYT}_{\post} - \vec\Delta'_{\post}\bigr).
\label{eq:A-corrected-cs}
\end{align}
In \eqref{eq:A-corrected-cs} the middle term vanishes because $\vec\Delta'_{\pre} = \vec\beta^{\CSNYT}_{\pre} = \vec\Delta^{\CSNYT}_{\pre}$, and $\vec\Delta^{\CSNYT}_{\post} - \vec\Delta'_{\post} = (\vec\Delta^{\Imp}_{\post} - \vec\varkappa_{\post}) - (\vec\Delta_{\post} - \vec\varkappa_{\post}) = \vec\Delta^{\Imp}_{\post} - \vec\Delta_{\post}$. The right-hand sides of \eqref{eq:A-corrected-imp} and \eqref{eq:A-corrected-cs} coincide, so every $\ell'$ combination does, and the two identified sets are equal.
\end{proof}

The same steps apply to the plug-in sets $\mathcal S(\vec{\hat\beta}, \Lambda_\Delta)$. Step~1 holds for the estimates because \eqref{eq:pre-csnyt} and \eqref{eq:pre-imp} compare the same two groups and differ only in the reference period, so that $\hat\beta^{\Imp}_{\mathcal G_g,s} - \hat\beta^{\CSNYT}_{\mathcal G_g,s} = \hat\beta^{\Imp}_{\mathcal G_g,0}$ for $s \leq 0$. Step~3 requires the sample counterpart of $\vec\beta^{\Imp}_{\post} - \vec\beta^{\CSNYT}_{\post} = (\mathbf W^{\Imp}\vec\Delta^{\Imp})_{\post} - (\mathbf W^{\CSNYT}\vec\Delta^{\CSNYT})_{\post}$. By Step~1 and the block structure of the two maps, this identity holds cell by cell in the form below with population coefficients, and its sample counterpart is
\begin{equation}
\hat\beta^{\Imp}_{\mathcal G_g,s} - \hat\beta^{\CSNYT}_{\mathcal G_g,s} = \hat\beta^{\Imp}_{\mathcal G_g,0} + \sum_{k \in \mathcal K_g(t)} w_k\, \hat\beta^{\Imp}_{\mathcal G_k, s_k(t_g - 1)}, \qquad s \geq 1.
\label{eq:A-sample-gap}
\end{equation}
Both sides are linear functions of the observed cohort-period cell means, the treated cell of $\mathcal G_g$ in period $t$ cancelling on the left, and the population version of \eqref{eq:A-sample-gap} holds for every configuration of the means by Propositions~\ref{prop:imputation} and~\ref{prop:csnyt}. By the argument for the sample identities in Section~\ref{app:correspondence}, \eqref{eq:A-sample-gap} therefore holds exactly for the estimates. The plug-in identified sets of the two estimators consequently coincide exactly, as Section~\ref{sec:equivalence} states.

\subsection{Proof of Proposition~\ref{prop:ci-equivalence}}
\label{app:ci-equivalence}

Three facts combine: the plug-in block-bias paths of the two estimators differ by cohort constants; every procedure reads the data only through differences of that path along cohorts, and the nuisance parameter only through a matrix that the two estimators share; and two quantities that are the same linear function of the outcomes receive the same estimated covariance under \eqref{eq:A-vcov}. Throughout, vectors are indexed by the full set of cells of Section~\ref{sec:coefvector}.

\begin{proof}
(i) Write $\vec{\hat\Delta} = \mathbf W^{-1}\vec{\hat\beta}$ for the plug-in path of either estimator, and let $\vec{\hat\varkappa}$ carry $\hat\varkappa_g = \hat\beta^{\Imp}_{\mathcal G_g,0}$ in every cell of cohort $\mathcal G_g$. Because
\[
\mathbf W = \begin{pmatrix}\mathbf I & 0\\ \mathbf W_{\post,\pre} & \mathbf W_{\post,\post}\end{pmatrix}, \qquad \mathbf W^{-1} = \begin{pmatrix}\mathbf I & 0\\ -\mathbf W_{\post,\post}^{-1}\mathbf W_{\post,\pre} & \mathbf W_{\post,\post}^{-1}\end{pmatrix},
\]
$\vec{\hat\Delta}_{\pre} = \vec{\hat\beta}_{\pre}$, and the sample version of Step~1 in Section~\ref{app:equivalence} gives $\vec{\hat\Delta}^{\Imp}_{\pre} - \vec{\hat\Delta}^{\CSNYT}_{\pre} = \vec{\hat\varkappa}_{\pre}$. In the post-treatment cells, Step~1 and the block structure of the two maps give $(\mathbf W^{\Imp}\vec\Delta^{\Imp})_{\post} - (\mathbf W^{\CSNYT}\vec\Delta^{\CSNYT})_{\post} = \mathbf W_{\post,\post}\vec\varkappa_{\post} - \mathbf W^{\CSNYT}_{\post,\pre}\vec\beta^{\CSNYT}_{\pre}$, whose cell-by-cell form is the population version of \eqref{eq:A-sample-gap}. Section~\ref{app:equivalence} shows that \eqref{eq:A-sample-gap} itself holds exactly for the estimates. Hence $\vec{\hat\beta}^{\Imp}_{\post} - \vec{\hat\beta}^{\CSNYT}_{\post} = \mathbf W_{\post,\post}\vec{\hat\varkappa}_{\post} - \mathbf W^{\CSNYT}_{\post,\pre}\vec{\hat\beta}^{\CSNYT}_{\pre}$. Since $\mathbf W^{\Imp}_{\post,\pre} = 0$, the post-treatment block of $\vec{\hat\Delta}^{\Imp}$ is $\mathbf W_{\post,\post}^{-1}\vec{\hat\beta}^{\Imp}_{\post}$, and
\begin{align*}
\vec{\hat\Delta}^{\CSNYT}_{\post} &= \mathbf W_{\post,\post}^{-1}\bigl(\vec{\hat\beta}^{\CSNYT}_{\post} - \mathbf W^{\CSNYT}_{\post,\pre}\vec{\hat\beta}^{\CSNYT}_{\pre}\bigr) \\
&= \mathbf W_{\post,\post}^{-1}\bigl(\vec{\hat\beta}^{\Imp}_{\post} - \mathbf W_{\post,\post}\vec{\hat\varkappa}_{\post}\bigr) = \vec{\hat\Delta}^{\Imp}_{\post} - \vec{\hat\varkappa}_{\post}.
\end{align*}
This gives $\vec{\hat\Delta}^{\CSNYT} = \vec{\hat\Delta}^{\Imp} - \vec{\hat\varkappa}$: the plug-in paths differ by cohort constants, as the population paths do.

(ii) For a polyhedral restriction $\{\mathbf P\vec\Delta \leq \vec d\}$, the tests of Section~\ref{sec:confidence} assess a candidate value $\theta_0$ through the moments $\mathbf P\mathbf W^{-1}(\vec{\hat\beta} - (\vec 0', \vec\tau_{\post}')')$, evaluated at hypothesized post-treatment effect vectors $\vec\tau_{\post}$ with $\ell'\vec\tau_{\post} = \theta_0$, their covariance $\mathbf P\mathbf W^{-1}\hat\Sigma_N(\mathbf W^{-1})'\mathbf P'$, and the bounds $\vec d$. Section~\ref{app:post-only-center} writes these vectors as $\theta_0\bar\ell + \tilde\tau$, with nuisance parameter $\tilde\tau$. By the block form of $\mathbf W^{-1}$ in (i), $\mathbf W^{-1}(\vec 0', \vec\tau_{\post}')' = (\vec 0', (\mathbf W_{\post,\post}^{-1}\vec\tau_{\post})')'$, and the moments equal $\mathbf P\vec{\hat\Delta} - \mathbf P_{\cdot,\post}\mathbf W_{\post,\post}^{-1}\vec\tau_{\post}$. The data therefore enter only through $\mathbf P\vec{\hat\Delta}$ and its covariance, and the nuisance parameter only through $\mathbf P_{\cdot,\post}\mathbf W_{\post,\post}^{-1}$, which is common to the two estimators because $\mathbf W_{\post,\post}$ is. Under $\Lambda^{\mathrm{SD}}_\Delta(M)$, $\Lambda^{\mathrm{SD,F}}_\Delta(M)$, and each component of $\Lambda^{\mathrm{RM}}_\Delta(\bar M)$, every row of $\mathbf P$ is a linear combination of first differences $\Delta_{\mathcal G_g,s} - \Delta_{\mathcal G_g,s-1}$ along cohort paths (Step~2 in Section~\ref{app:equivalence}), so $\mathbf P\vec{\hat\varkappa} = 0$ and, by (i), $\mathbf P\vec{\hat\Delta}^{\Imp} = \mathbf P\vec{\hat\Delta}^{\CSNYT}$. The components of the RM union are indexed by the location and sign of the binding pre-treatment change, an index set fixed by the design, so the two estimators' unions combine the same components. For the FLCI, $\xi$ and $B^{\mathrm P} = M\|\xi\|_1$ are common because $\mathbf W_{\post,\post}$ is. By the identity $v^{\mathrm P\prime} = \xi'\mathbf D_{\post}\mathbf W^{-1}$ of Section~\ref{sec:flci-center}, $c = v^{\mathrm P\prime}\vec{\hat\beta} = \xi'\mathbf D_{\post}\vec{\hat\Delta}$, while $\hat z = \mathbf D_{\pre,\pre}\vec{\hat\beta}_{\pre} = \mathbf D_{\pre,\pre}\vec{\hat\Delta}_{\pre}$. Both are differences of the plug-in path along cohorts and hence common as well.

(iii) Each quantity in (ii) is a linear function of the coefficients, $\mathbf Q\vec{\hat\beta} = \mathbf Q\mathbf L\vec Y$ with $\mathbf L$ the operator of \eqref{eq:A-operator}, and because (i) and (ii) hold for every outcome vector, $\mathbf Q^{\Imp}\mathbf L^{\Imp} = \mathbf Q^{\CSNYT}\mathbf L^{\CSNYT}$ as matrices. Under \eqref{eq:A-vcov} the estimated covariance of such a quantity is $\mathbf Q\hat\Sigma_N\mathbf Q' = a_N\sum_{\mathsf c} (\mathbf Q\mathbf L_{\mathsf c}\vec{\hat e}_{\mathsf c})(\mathbf Q\mathbf L_{\mathsf c}\vec{\hat e}_{\mathsf c})'$, where $a_N$ is the finite-sample factor of \eqref{eq:A-vcov} and $\vec{\hat e}_{\mathsf c}$ the cell-mean residuals of cluster $\mathsf c$, which are computed from the data alone. Neither depends on the estimator, and $\mathbf Q\mathbf L_{\mathsf c}$ is the same matrix for both. The moments and their covariance, the nuisance matrix, and $(c, \hat z)$ with its covariance are therefore identical for the two estimators.

The conditional test statistic, its conditioning event, and the least-favorable critical value are functions of the moments, their covariance, the bounds, and the nuisance matrix, so the Conditional and C-LF tests reach the same decision at every $\theta_0$ under either restriction. The program \eqref{eq:flci-program} depends on the data only through the covariance of $(c, \hat z)$, so its set of solutions coincides. When that set has more than one element, the common selection rule picks the same $\gamma$ for both estimators. The length of the interval and its center $c + \gamma'\hat z$ are then the same for both estimators. The post-only interval is the case $\gamma = 0$. Aggregate inference is not covered: $\mathbf A\vec{\hat\beta}$ is not a difference of the plug-in path, and the aggregate paths of the two estimators differ by more than cohort constants (Section~\ref{sec:aggregate-inference}).
\end{proof}

\subsection{Reduced coordinates and the proof of Proposition~\ref{prop:flci-weights}}
\label{app:reduced}

\paragraph*{Reduced coordinates.}
In the full pre-treatment vector of Section~\ref{sec:coefvector}, each cohort satisfies one exact normalization identity: $\Delta_{\mathcal G_g, 0} = 0$ for \CSNYT{} and $\sum_{s \le 0}\Delta_{\mathcal G_g, s} = 0$ (so that $\Delta_{\mathcal G_g, 0} = -\sum_{s < 0}\Delta_{\mathcal G_g, s}$) for the imputation estimator, with the same identities holding for $\vec{\hat\beta}_{\pre}$. Because of these identities, pre-treatment weights $v_{\pre}$ on the full cell vector are only unique up to each cohort's normalization: shifting $v_{\pre}$ along the normalization direction leaves both the estimator $v_{\pre}'\vec{\hat\beta}_{\pre}$ and its bias $v_{\pre}'\vec\Delta_{\pre}$ unchanged.

To make the admissible weights in Proposition~\ref{prop:flci-weights} unique, I eliminate this redundancy by dropping the $s = 0$ cell of each cohort, retaining the remaining $t_g - 2$ pre-treatment cells, which vary freely, and all post-treatment cells. Let $\mathbf J$ denote the linear map that reconstructs the full cell vector from the reduced coordinates by inserting $\Delta_{\mathcal G_g, 0} = 0$ for \CSNYT{} and $\Delta_{\mathcal G_g, 0} = -\sum_{s < 0}\Delta_{\mathcal G_g, s}$ for imputation. Right-multiplying any operator on the full vector (such as $\mathbf D$ in \eqref{eq:flci-rows} or $\mathbf W$, after dropping the $s = 0$ output rows of $\mathbf W$) by $\mathbf J$ expresses it in reduced coordinates; for simplicity, I keep the same symbols for the reduced vectors and matrices. For each cohort $\mathcal G_g$, the reduced matrix $\mathbf D_{\pre,\pre}$ has $t_g - 3$ rows and $t_g - 2$ columns, and its null space is spanned by the single linear pre-treatment path satisfying the cohort's normalization, so $\mathbf D_{\pre,\pre}$ retains full row rank. All other proofs in this appendix work directly on the full cell vector.

\paragraph*{Proof of Proposition~\ref{prop:flci-weights}.}
The proof works in the reduced coordinates and describes a block-bias path by its pre-treatment block and its post-only second differences $u$ of \eqref{eq:flci-rows}, written $(\vec\Delta_{\pre}, u)$.

\begin{proof}
By \eqref{eq:flci-center} the bias of $v'\vec{\hat\beta}$ is $b = (v_{\pre} - v^{\mathrm P}_{\pre})'\vec\Delta_{\pre} + \xi'u$, and because $\mathbf D_{\post,\post}$ is invertible, every pair $(\vec\Delta_{\pre}, u)$ corresponds to exactly one $\vec\Delta$. In the coordinates $(\vec\Delta_{\pre}, u)$ each restriction is a product set: $u$ ranges over the box $|u| \leq M$, while $\vec\Delta_{\pre}$ is unrestricted under $\Lambda^{\mathrm{SD}}_\Delta$ and restricted to $|\mathbf D_{\pre,\pre}\vec\Delta_{\pre}| \leq M$ under $\Lambda^{\mathrm{SD,F}}_\Delta$. Both factors are symmetric about zero, so the supremum of $|b|$ is the sum of the suprema of the two terms. The second is $M\|\xi\|_1$, attained at $u_j = M\,\mathrm{sign}(\xi_j)$.

(a) With $\vec\Delta_{\pre}$ unrestricted, the first supremum is zero if $v_{\pre} = v^{\mathrm P}_{\pre}$ and infinite otherwise.

(b) Decompose $v_{\pre} - v^{\mathrm P}_{\pre} = \mathbf D_{\pre,\pre}'\gamma + \zeta$ with $\zeta$ in the null space of $\mathbf D_{\pre,\pre}$, which is possible because the range of $\mathbf D_{\pre,\pre}'$ is the orthogonal complement of that null space. Paths $\vec\Delta_{\pre}$ proportional to $\zeta$ have zero pure pre-treatment curvature and are therefore feasible at every scale. Along them $(v_{\pre} - v^{\mathrm P}_{\pre})'\vec\Delta_{\pre}$ is proportional to $\|\zeta\|^2$ with an unbounded factor, so finiteness forces $\zeta = 0$. The vector $\gamma$ is unique because $\mathbf D_{\pre,\pre}$ has full row rank and $\mathbf D_{\pre,\pre}'$ is therefore injective. Then $(v_{\pre} - v^{\mathrm P}_{\pre})'\vec\Delta_{\pre} = \gamma'z$, and as $\vec\Delta_{\pre}$ ranges over the feasible set, $z$ ranges over the whole box $|z| \leq M$ because $\mathbf D_{\pre,\pre}$ is surjective. The first supremum is therefore $M\|\gamma\|_1$.
\end{proof}

\subsection{Proof of Proposition~\ref{prop:post-only-center}}
\label{app:post-only-center}

Part~(a) follows from the algebra of \eqref{eq:flci-center}, part~(b) from a reflection of the moment inequalities about $c$, and part~(c) from the linearity of both estimators in the cohort-period cell means.

\begin{proof}
(a) In population $\vec\beta_{\pre} = \vec\Delta_{\pre}$, and the algebra of \eqref{eq:flci-center} gives, for every admissible $\vec\Delta$ with $\vec\Delta_{\pre} = \vec\beta_{\pre}$, $\ell'(\vec\beta_{\post} - (\mathbf W\vec\Delta)_{\post}) = v^{\mathrm P\prime}\vec\beta - \xi'u$ with $u = \mathbf D_{\post}\vec\Delta$. Under $\Lambda^{\mathrm{SD}}_\Delta(M)$ every pair $(\vec\Delta_{\pre}, u)$ with $\vec\Delta_{\pre} = \vec\beta_{\pre}$ and $|u| \leq M$ is feasible, and because $\mathbf D_{\post,\post}$ is invertible each corresponds to exactly one $\vec\Delta$, so the identified set is the interval $v^{\mathrm P\prime}\vec\beta \pm M\|\xi\|_1$, its endpoints attained at $u_j = \mp M\,\mathrm{sign}(\xi_j)$. The plug-in set replaces $\vec\beta$ by $\vec{\hat\beta}$ in the same algebra. Under $\Lambda^{\mathrm{SD,F}}_\Delta(M)$ the additional constraint $|z| \leq M$ involves only $\vec\Delta_{\pre} = \vec\beta_{\pre}$, so it holds for every feasible $\vec\Delta$ or for none.

(b) Under $\Lambda^{\mathrm{SD}}_\Delta(M)$ the mapped restriction of Section~\ref{sec:model} is the pair of inequalities $\pm\mathbf D_{\post}\mathbf W^{-1}\vec\delta \leq M$, so the Conditional and C-LF tests of a value $\theta_0$ use the moments $\pm m(\theta_0, \tilde\tau) \leq M$ \citep[Section~3]{rambachan2023more} with
\begin{equation}
m(\theta_0, \tilde\tau) = a - \mathbf K(\theta_0\bar\ell + \tilde\tau), \qquad a = \mathbf D_{\post}\mathbf W^{-1}\vec{\hat\beta}, \qquad \mathbf K = \mathbf D_{\post,\post}\mathbf W_{\post,\post}^{-1},
\label{eq:A-moments}
\end{equation}
where $\bar\ell$ is a fixed vector with $\ell'\bar\ell = 1$ and the nuisance parameter $\tilde\tau$ ranges over the post-treatment effect vectors with $\ell'\tilde\tau = 0$, so that $\theta_0\bar\ell + \tilde\tau$ runs over every post-treatment effect vector whose target value is $\theta_0$. The matrix $\mathbf K$, the post-treatment block of $\mathbf D_{\post}\mathbf W^{-1}$, is invertible because $\mathbf D_{\post,\post}$ and $\mathbf W_{\post,\post}$ are.

Both tests profile $\tilde\tau$ out through the linear program $\min_{\tilde\tau}\max_j (\pm m_j - M)/\sigma_j$, the largest excess of a moment over its bound in units of the moment's standard error $\sigma_j$, and reject $\theta_0$ when its value exceeds a critical value computed from the covariance of the moments and the nuisance matrix and, for the conditional test, from the realized moments through which of them bind and through the truncation bounds \citep{andrews2023inference}. Because of this profiling, the statistics of both tests, the conditioning event and truncation bounds of the conditional critical value, and the least-favorable critical value depend on $(\theta_0, \tilde\tau)$ only through the set of moment vectors $\{m(\theta_0, \tilde\tau) : \ell'\tilde\tau = 0\}$. All of them are also equivariant under a joint permutation of rows with equal standard errors and bounds.

By the identity $v^{\mathrm P\prime} = \xi'\mathbf D_{\post}\mathbf W^{-1}$ of Section~\ref{sec:flci-center}, $\xi'a = c$, and by \eqref{eq:flci-center}, $\xi'\mathbf K = \ell'$, so $\xi'm(\theta_0, \tilde\tau) = c - \theta_0$ for every admissible $\tilde\tau$. Let $\theta_0' = 2c - \theta_0$ and $\tilde\tau' = 2\mathbf K^{-1}a - 2c\bar\ell - \tilde\tau$. Because $\ell'\mathbf K^{-1} = \xi'$, $\ell'\tilde\tau' = 2\xi'a - 2c = 0$, so $\tilde\tau'$ is admissible. Substituting $\theta_0'$ and $\tilde\tau'$ into \eqref{eq:A-moments} gives $m(\theta_0', \tilde\tau') = -m(\theta_0, \tilde\tau)$. The reflection maps the set of moment vectors at $\theta_0$ onto its negative at $2c - \theta_0$. Negating a moment vector interchanges the two rows of each $\pm$ pair, which have the same standard error and bound, so the tests reach the same decision at $\theta_0$ and at $2c - \theta_0$.

(c) $\xi$ is common to the imputation estimator and \CSNYT{} because $\mathbf W_{\post,\post}$ is; $v^{\mathrm P}_{\pre}$ differs through $\mathbf W_{\post,\pre}$. Both post-only centers are linear functions of the same cohort-period cell means (Section~\ref{app:covariance}), and by \eqref{eq:flci-center} both have expectation $\theta + \xi'u$ for every configuration of potential-outcome means and treatment effects; $u$ is the same for both estimators because second differences remove the cohort constants $\vec\varkappa$ by which the two block-bias paths differ (Step~1 of the proof in Section~\ref{app:equivalence}). Every vector of expected cell means arises from some configuration, so two linear functions of the cell means with the same expectation for every configuration are the same function; the two post-only centers are therefore the same function of the data.
\end{proof}

\subsection{Convexity and threshold of the full-path program}
\label{app:flci-program}

Let $\phi$ and $\Phi$ be the standard normal density and distribution function and $z_{1 - \alpha}$ the $1 - \alpha$ standard normal quantile. Since $\cvq_\alpha(r)$ is the $1 - \alpha$ quantile of $|N(r, 1)|$ (Section~\ref{sec:flci-definition}), $q = \cvq_\alpha(r)$ solves $\Phi(q - r) + \Phi(q + r) = 2 - \alpha$; differentiating this equation gives $\cvq_\alpha'(r) = [\phi(q - r) - \phi(q + r)]/[\phi(q - r) + \phi(q + r)] = \tanh(rq) \in (0, 1)$ for $r > 0$, and $\cvq_\alpha''(r) = (1 - \tanh^2(rq))(q + r\,\cvq_\alpha'(r)) > 0$. Hence $(B, \sigma) \mapsto \sigma\,\cvq_\alpha(B/\sigma)$ is the perspective of a convex function, the function obtained by scaling a convex function's argument and value by $\sigma$: it is convex on $\sigma > 0$ and extends continuously to $\sigma = 0$ with the value $B$. It is nondecreasing in $B$ because $\cvq_\alpha' > 0$. It is also nondecreasing in $\sigma$. Its derivative in $\sigma$ is $\cvq_\alpha(r) - r\,\cvq_\alpha'(r)$ with $r = B/\sigma$. This derivative decreases in $r$ toward the limit $z_{1 - \alpha}$, because its own derivative is $-r\,\cvq_\alpha''(r) < 0$, and $z_{1 - \alpha} > 0$ because $\alpha < 1/2$.

In \eqref{eq:flci-program}, $\sigma(v(\gamma))$ is convex in $\gamma$ as the norm of an affine function and $B^{\mathrm P} + M\|\gamma\|_1$ is convex, so the objective is the composition of a convex function, nondecreasing in each argument, with convex functions; it is therefore convex, with kinks at $\gamma_j = 0$.

At $M = 0$ the objective is $\sigma(v(\gamma))\,\cvq_\alpha(0)$, so the problem is to minimize $\mathrm{Var}(c + \gamma'\hat z)$. Let $\Sigma_{\pre}$ be the pre-treatment block of $\Sigma_N$, so that $\mathrm{Var}(\hat z) = \mathbf D_{\pre,\pre}\Sigma_{\pre}\mathbf D_{\pre,\pre}'$. On the full set of cells $\Sigma_{\pre}$ is singular along the normalization directions of Section~\ref{sec:coefvector}, one per cohort. Their span meets the range of $\mathbf D_{\pre,\pre}'$ only at zero, because within each cohort that range is orthogonal to the constant path while the normalization direction is not. Hence $\mathrm{Var}(\hat z)$ is positive definite whenever $\Sigma_{\pre}$ has no other null direction, in particular whenever the covariance of the reduced coordinates of Section~\ref{app:reduced} is positive definite. Suppose first that $\mathrm{Var}(\hat z)$ is positive definite. The solution is $\gamma = -\mathrm{Var}(\hat z)^{-1}\mathrm{Cov}(\hat z, c)$. If $\mathrm{Var}(\hat z)$ is singular, as with a rank-deficient clustered covariance, the normal equations $\mathrm{Var}(\hat z)\gamma = -\mathrm{Cov}(\hat z, c)$ remain consistent, because their right side lies in the column space of $\mathrm{Var}(\hat z)$, and every solution, for instance one obtained with a generalized inverse, attains the same minimum variance.

For $M > 0$ suppose $\sigma(v^{\mathrm P}) > 0$ and write $r = B^{\mathrm P}/\sigma(v^{\mathrm P})$, the argument of $\cvq_\alpha$ at $\gamma = 0$. Here $B^{\mathrm P} = M\|\xi\|_1 > 0$ for $\ell \neq 0$, because $\xi$ is the image of $\ell$ under an invertible matrix. At $\gamma = 0$ the map $(B, \sigma) \mapsto \sigma\,\cvq_\alpha(B/\sigma)$ has partial derivatives $\cvq_\alpha'(r)$ in $B$ and $\cvq_\alpha(r) - r\,\cvq_\alpha'(r)$ in $\sigma$, $\sigma(v(\gamma))$ has $\partial\sigma/\partial\gamma_j = \mathrm{Cov}(\hat z_j, c)/\sigma(v^{\mathrm P})$, and the subdifferential of $M\|\gamma\|_1$ is the box $M[-1, 1]^J$, with $J$ the number of pure pre-treatment second differences. The subdifferential of the objective at zero is therefore the product over $j$ of the intervals $\bigl(\cvq_\alpha(r) - r\,\cvq_\alpha'(r)\bigr)\,\partial\sigma/\partial\gamma_j + M\,\cvq_\alpha'(r)[-1, 1]$, and $\gamma = 0$ is optimal if and only if each contains zero, that is, if for every $j$
\begin{equation}
\Bigl|\frac{\partial\sigma(v(\gamma))}{\partial\gamma_j}\Bigr|_{\gamma = 0}\bigl(\cvq_\alpha(r) - r\,\cvq_\alpha'(r)\bigr) \leq M\,\cvq_\alpha'(r), \qquad r = \frac{B^{\mathrm P}}{\sigma(v^{\mathrm P})}.
\label{eq:A-threshold}
\end{equation}
As $M$ grows, $r$ grows, $\cvq_\alpha'(r)$ increases to one and $\cvq_\alpha(r) - r\,\cvq_\alpha'(r)$ decreases to $z_{1 - \alpha}$, so the left side of \eqref{eq:A-threshold} is nonincreasing and bounded while the right side increases without bound. Each inequality therefore holds exactly on a half-line $[M_j, \infty)$, where $M_j$ is the threshold for the $j$th pure pre-treatment second difference, and strictly for $M > M_j$. Once $M$ exceeds $\max_j M_j$, zero lies in the interior of the subdifferential, $\gamma = 0$ is the unique minimizer, and the full-path FLCI coincides with the post-only FLCI.

\subsection{The covering bound for aggregate inference}
\label{app:covering-bound}

Under the post-only SD restriction, two population bounds relate aggregate to cohort inference. The covering bound $M_A^C$ is the smallest aggregate bound at which the aggregate restriction $\Lambda^{\mathrm{SD}}_A(M_A)$ in \eqref{eq:aggregate-sd} accommodates the cohort restriction $\Lambda^{\mathrm{SD}}_\Delta(M_C)$ in \eqref{eq:sd}, that is, admits $\mathbf A\mathbf W\vec\Delta$ for every $\vec\Delta$ that the cohort restriction admits. The target-specific bound $M_A^\theta$ is the smallest aggregate bound at which the aggregate identified set contains the cohort identified set for a given target $\theta$. Throughout, I fix the pre-treatment block biases at their population values, $\vec\Delta_{\pre} = \vec\beta_{\pre}$, and, following Sections~\ref{app:reduced} and~\ref{app:post-only-center}, parameterize $\vec\Delta$ by $(\vec\Delta_{\pre}, u)$, where $u = \mathbf D_{\post}\vec\Delta$ collects the post-only second differences of \eqref{eq:flci-rows}.

Under $\Lambda^{\mathrm{SD}}_\Delta(M_C)$, every admissible block-bias path takes the form
\begin{equation}
\vec\Delta = \vec\Delta^0 + \mathbf E u, \qquad |u| \leq M_C.
\label{eq:A-admissible}
\end{equation}
Here $\vec\Delta^0$ is the linear continuation path at $M_C = 0$, under which each cohort continues the slope of its last two pre-treatment periods, and $\mathbf E$ has zeros in its pre-treatment rows and equals $\mathbf D_{\post,\post}^{-1}$ in its post-treatment rows, so that $\mathbf D_{\post}\mathbf E u = u$. Let $\mathbf D_A$ form the post-only second differences of the aggregate path at the relative periods $s \geq 1$ in the aggregation window of $\mathbf A$, so that the aggregate restriction $\Lambda^{\mathrm{SD}}_A(M_A)$ requires $|\mathbf D_A\mathbf A\mathbf W\vec\Delta| \leq M_A$. Substituting \eqref{eq:A-admissible} gives
\begin{equation}
\mathbf D_A\mathbf A\mathbf W\vec\Delta = \vec a + \mathbf C u, \qquad \vec a = \mathbf D_A\mathbf A\mathbf W\vec\Delta^0, \qquad \mathbf C = \mathbf D_A\mathbf A\mathbf W\mathbf E .
\label{eq:A-aggregate-sd}
\end{equation}
In \eqref{eq:A-aggregate-sd}, $a_s$ is the second difference of the aggregate path at relative period $s$ when every cohort continues its pre-treatment slope ($u = 0$), and row $s$ of $\mathbf C$ maps the cohorts' post-treatment slope changes $u$ into the aggregate second difference at $s$. Over the box $|u| \leq M_C$, the largest absolute value of entry $s$ of \eqref{eq:A-aggregate-sd} is $|a_s| + M_C\|\mathbf C_{s\cdot}\|_1$. The cohort restriction at $M_C$ therefore implies the aggregate restriction at $M_A$ if and only if
\begin{equation}
M_A \geq M_A^C = \max_{s \geq 1}\bigl(|a_s| + M_C\|\mathbf C_{s\cdot}\|_1\bigr).
\label{eq:A-covering}
\end{equation}
Whenever $M_A \geq M_A^C$, every admissible path $\vec\Delta \in \Lambda^{\mathrm{SD}}_\Delta(M_C)$ induces an aggregate bias path $\mathbf A\mathbf W\vec\Delta$ that satisfies $\Lambda^{\mathrm{SD}}_A(M_A)$ and, because $\vec\Delta_{\pre} = \vec\beta_{\pre}$, equals $\vec\beta_{A,\pre}$ before treatment, the two conditions in $\mathcal S_A$ of \eqref{eq:aggregate-sd}. For any post-treatment target $\theta = \ell_A'\vec\tau_A = \ell'\vec\tau_{\post}$ with $\mathbf A'\ell_A = (\vec 0', \ell')'$ (Section~\ref{sec:aggregate}), the induced path returns the same target value, $\ell_A'(\vec\beta_A - \mathbf A\mathbf W\vec\Delta)_{\post} = \ell'(\vec\beta - \mathbf W\vec\Delta)_{\post}$, so $\mathcal S_A$ then contains the cohort identified set in \eqref{eq:identified-set}.

Equation \eqref{eq:A-covering} separates the wedge between the aggregate and cohort bounds into two terms. The first, $a_s$, captures the kinks in the aggregate path that arise from continuing the pre-treatment cohort paths. Under parallel trends, $a_s = 0$. Otherwise $a_s$ reflects compositional shifts in $\mathbf A$ across cohorts with different levels or slopes. It also reflects the adjustment terms in $\mathbf W$, which enter even when the cohort pre-trends are homogeneous.

The second, $\|\mathbf C_{s\cdot}\|_1$, scales the post-treatment cohort curvature. A contemporaneous slope change $u_{\mathcal G_g,s}$ enters the estimation bias of cell $(\mathcal G_g,s)$ one for one. The adjustment terms of $\mathbf W$ carry it only into earlier-adopting cohorts' cells at later relative periods (Section~\ref{sec:biasdecomposition}), so no other cell in the second difference at $s$ moves, and its weight in row $s$ of $\mathbf C$ equals cohort $\mathcal G_g$'s share in the aggregate at $s$. Because these shares sum to one across the cohorts observed at $s$, $\|\mathbf C_{s\cdot}\|_1 \geq 1$ and $M_A^C \geq M_C$. When the treated cohort composition is fixed over the window and $\mathbf W = \mathbf I$, these contemporaneous shares are the only nonzero entries in row $s$ of $\mathbf C$, so $\|\mathbf C_{s\cdot}\|_1 = 1$; a fixed-weight average of linear continuations is linear, so $a_s = 0$; and $M_A^C = M_C$, the case in which the two identified sets coincide (Section~\ref{sec:aggregate}). When cohorts enter or leave the average, the aggregate second difference at $s$ contrasts different cohorts at different horizons. Worst-case paths then bend these cohorts in opposite directions, and the contrasts accumulate in $\|\mathbf C_{s\cdot}\|_1 > 1$.

Containing the cohort identified set for a single target $\theta$ needs at most the covering bound and usually less. By Proposition~\ref{prop:post-only-center}(a), the cohort identified set is the interval $c_C \pm M_C H_C$, with $c_C = v^{\mathrm P\prime}\vec\beta$ and $H_C = \|\xi\|_1$. The same proposition applied to the aggregate path, with $\vec\beta_A$, $\ell_A$, and $\mathbf D_A$ in place of $\vec\beta$, $\ell$, and $\mathbf D_{\post}$ and the identity in place of $\mathbf W$, gives the aggregate identified set $c_A \pm M_A H_A$: $c_A$ subtracts from $\ell_A'\vec\beta_{A,\post}$ the $\ell_A$-weighted linear continuation of the aggregate path from its last two pre-treatment periods, and $H_A$ is the $\ell_1$ norm of the analogue of $\xi$ in \eqref{eq:flci-center}. The aggregate interval contains the cohort interval if and only if $M_A H_A \geq M_C H_C + |c_A - c_C|$. Solving for $M_A$ gives the target-specific bound
\begin{equation}
M_A^{\theta} = M_A^E + \frac{|c_A - c_C|}{H_A}, \qquad M_A^E = \frac{M_C H_C}{H_A},
\label{eq:A-target-bound}
\end{equation}
where $M_A^E$ is the width-matched bound of Section~\ref{sec:sim-designs}. Because $\ell_A'\vec\beta_{A,\post} = \ell'\vec\beta_{\post}$, the post-treatment coefficients cancel in $c_A - c_C$. Both bounds depend on the data only through $\vec\beta_{\pre}$, $M_A^C$ through $\vec\Delta^0$ in $\vec a$ and $M_A^\theta$ through $c_A - c_C$; the remaining inputs are $\mathbf A$, $\mathbf W$, and, for $M_A^\theta$, the target. Under parallel trends, $c_A = c_C$ and $M_A^\theta = M_A^E$. Away from parallel trends the gap is in general nonzero, and containing the cohort identified set then requires an aggregate identified set strictly wider than the cohort set, $M_A^\theta > M_A^E$.

Table~\ref{tab:A-covering} reports these bounds for the simulation designs of Section~\ref{sec:simulations} at $M_C = 1$, under the DGPs SD1 to SD3 of Section~\ref{sec:sim-heterogeneous}. Even under parallel trends (SD1, where $\vec a = 0$), $M_A^C$ is $8.3\,M_C$ in Design~I, $34.4\,M_C$ in Design~II on the window of all post-treatment periods, and $3.6\,M_C$ in Design~III. With the identity in place of $\mathbf W$ these multipliers are $8.8$, $37.5$, and $3.9\,M_C$, so the composition change alone produces them and the adjustment terms of $\mathbf W$ partly offset the contrasts. Restricting Design~II to the balanced window fixes the treated cohort weights and reduces $M_A^C$ to $1.1\,M_C$ under SD1, where the remaining wedge comes from $\mathbf W$ alone. With the identity in place of $\mathbf W$ it equals $M_C$: with fixed weights the composition change contributes nothing, and the excess in Design~II comes from the cohorts adopting inside the window. Under the heterogeneous DGPs SD2 and SD3, nonzero kinks $\vec a$ raise $M_A^C$ further, to $8.8\,M_C$ and $13.8\,M_C$ in Design~I, and move the aggregate center away from the cohort center. In Design~I the center gap $c_A - c_C$ is $0.94$ and $4.52$ normalized units under SD2 and SD3, which corresponds to the upward shift of the aggregate confidence sets in Section~\ref{sec:sim-heterogeneous} and raises $M_A^\theta$ from $1.04\,M_C$ under SD1 to $1.42\,M_C$ and $2.90\,M_C$.

\begin{table}[H]
\centering
\caption{Aggregate bounds implied by the cohort restriction in the simulation designs}
\label{tab:A-covering}
\footnotesize
\begin{tabular}{@{}llrrrrrr@{}}
\toprule
 & & \multicolumn{3}{c}{Bound, multiple of $M_C$} & Center gap & \multicolumn{2}{c}{Width, aggregate/cohort} \\
\cmidrule(lr){3-5} \cmidrule(lr){6-6} \cmidrule(lr){7-8}
Design and aggregation window & DGP & $M_A^E$ & $M_A^C$ & $M_A^{\theta}$ & $c_A - c_C$ & at $M_A^C$ & at $M_A^{\theta}$ \\
\midrule
\multirow{3}{*}{I, all post-treatment periods} & SD1 & 1.039 & 8.280 & 1.039 & 0.000 & 7.97 & 1.00 \\
 & SD2 & 1.039 & 8.830 & 1.425 & 0.937 & 8.50 & 1.37 \\
 & SD3 & 1.039 & 13.829 & 2.900 & 4.522 & 13.30 & 2.79 \\
\addlinespace
\multirow{3}{*}{II, all post-treatment periods} & SD1 & 1.082 & 34.425 & 1.082 & 0.000 & 31.82 & 1.00 \\
 & SD2 & 1.082 & 36.743 & 1.137 & 0.374 & 33.96 & 1.05 \\
 & SD3 & 1.082 & 49.014 & 1.334 & 1.691 & 45.31 & 1.23 \\
\addlinespace
\multirow{3}{*}{II, balanced window} & SD1 & 1.041 & 1.117 & 1.041 & 0.000 & 1.07 & 1.00 \\
 & SD2 & 1.041 & 1.285 & 1.066 & $-$0.127 & 1.24 & 1.02 \\
 & SD3 & 1.041 & 1.843 & 1.189 & $-$0.742 & 1.77 & 1.14 \\
\addlinespace
\multirow{3}{*}{III, all post-treatment periods} & SD1 & 1.045 & 3.597 & 1.045 & 0.000 & 3.44 & 1.00 \\
 & SD2 & 1.045 & 5.554 & 1.341 & 0.424 & 5.32 & 1.28 \\
 & SD3 & 1.045 & 14.795 & 2.737 & 2.424 & 14.16 & 2.62 \\
\bottomrule
\end{tabular}

\vspace{0.5em}
\parbox{\linewidth}{\footnotesize
\textit{Notes:} Population quantities at $M_C = 1$. Rows: design, aggregation window, and DGP. Columns: the width-matched bound $M_A^E$, the covering bound $M_A^C$ of \eqref{eq:A-covering}, the target-specific bound $M_A^{\theta}$ of \eqref{eq:A-target-bound}, the gap $c_A - c_C$ between the centers of the aggregate and cohort identified sets in normalized units, and the ratio of aggregate to cohort identified-set widths at $M_A^C$ and at $M_A^{\theta}$. The target is the average treatment effect over all post-treatment periods or the balanced target of Section~\ref{sec:sim-balanced}.
}
\end{table}
\section{Simulation designs and additional results}
\label{app:simulations}
\setcounter{equation}{0}
\setcounter{figure}{0}
\setcounter{table}{0}

Section~\ref{app:flci-designs} reports the two designs specific to the FLCI, Section~\ref{app:sim-setup} the common setup of Designs~I to~IV, and Sections~\ref{app:sim-design1} to~\ref{app:sim-design4} their additional exhibits. Unless a design states otherwise, the treatment structure is that of the minimum-wage application with $T = 7$ periods, the estimator is the imputation estimator with its map $\mathbf W$ (Section~\ref{sec:biasdecomposition}), and treatment effects are zero. Bounds, widths, and centers in Section~\ref{app:flci-designs} are in units of $\sigma_0 = 0.0065$ log points, the standard error of that section's target. Sections~\ref{app:sim-design1} to~\ref{app:sim-design4} use the reference curvature unit $m_0$ of Section~\ref{sec:sim-designs}, equal to $0.131$ in the outcome units of the simulated panels; the Design~IV exhibits are instead in outcome units. The procedures are those of Sections~\ref{sec:confidence} and~\ref{sec:flci} at $\alpha = 0.05$, with the first-stage size $\kappa = \alpha/10$ of C-LF.

The primary quantity is the minimum coverage over the population identified set under cohort inference, $\mathcal S_C$. Each procedure is evaluated over the set implied by its own restriction: the post-only set for Conditional and C-LF, and the full-path set for the FLCI. The two sets coincide whenever the pure pre-treatment second differences of the population path are at most $M$ in absolute value; otherwise the full-path set is empty (Section~\ref{sec:flci-center}). The minimum coverage is
\begin{equation}
\widehat{\mathrm{Cov}}_{\min} = \inf_{\theta \in \mathcal S_C}\ \frac{1}{R}\sum_{r=1}^{R}\mathbf 1\bigl\{\theta \in \mathcal C^{(r)}\bigr\},
\label{eq:B-min-coverage}
\end{equation}
where $\mathcal C^{(r)}$ is the confidence set in replication $r$. I compute the infimum over the whole set from the saved endpoints of the $R$ sets, while the rejection curves plot $1 - R^{-1}\sum_r \mathbf 1\{\theta \in \mathcal C^{(r)}\}$ over a grid of $\theta$. Coverage of the true effect is secondary, because a set can contain the true effect in every replication while missing other values in the identified set. It is the quantity of interest only in the stress design, whose DGP violates the full-path restriction.

\subsection{Designs specific to the FLCI}
\label{app:flci-designs}

The calibrated design draws the coefficient vector from its Gaussian sampling distribution, $\vec{\hat\beta} \sim N(\vec\beta, \lambda^2\Sigma_0)$. Here $\lambda \in \{0.1, 0.5, 1, 2\}$, and $\Sigma_0$ is an earlier estimate of the imputation cohort covariance from the minimum-wage application. The target is the cohort-size-weighted first post-treatment effect, with weights $0.110$, $0.246$, and $0.644$ for the cohorts adopting in periods 4, 6, and 7, and $\sigma_0$ is its standard error at $\lambda = 1$. The procedures use a covariance estimated from Gaussian unit scores, drawn independently of the coefficient vector. Every unit's score has covariance $\lambda^2\Sigma_0/N$, with $N$ the number of units, so each cohort contributes to $\lambda^2\Sigma_0$ in proportion to its size. The scores are centered within cohorts and their outer products summed with the factor $N_g/(N_g - 1)$ for each cohort, the Gaussian-score analogue of the estimator in Appendix~\ref{app:covariance}. The bound takes the values $M \in \sigma_0 \times \{0, 0.1, 0.25, 0.5, 1, 2\}$. Each cell $(\lambda, M)$ uses 200 covariance draws and ten coefficient draws per covariance, $2{,}000$ confidence sets in all.

I use two DGPs, and both satisfy the post-only and the full-path restrictions. Under parallel trends, $\vec\Delta = 0$ and the true effect, zero, is the center of the identified set $[-M, M]$. Its half-width equals $M$ because $\|\xi\|_1 = 1$ for this target. Under the binding directional DGP, every second difference equals $+M$, the identified set is $[0, 2M]$, and the true effect is its lower endpoint.

Figure~\ref{fig:app-cal-profiles} shows the rejection curves under parallel trends at $\lambda = 1$ and $\lambda = 0.5$. Each procedure rejects at about the nominal rate at the endpoints of the identified set and rarely inside it. Across all cells and both DGPs, the minimum coverage over the identified set lies between 94.3 and 96.6 percent.

At small bounds the full-path FLCI has a noticeably shorter mean total length than Conditional and C-LF (Figure~\ref{fig:app-cal-length}(a)), and the gap narrows as $M$ grows. Conditional and C-LF are nearly equal in length everywhere, so Section~\ref{sec:simulations} reports only C-LF.

The stress design keeps the post-only restriction and violates only the extra inequalities of the full-path restriction, those on the pure pre-treatment second differences. For each cell $(\lambda, M)$ with $M/\sigma_0 \in \{0, 0.05, 0.1\}$ and each excess $\nu \in \{0, 0.025, 0.05, 0.1, 0.25, 0.5, 1, 2\}$, I first compute the full-path FLCI weights with the true covariance $\lambda^2\Sigma_0$. I then construct a population block-bias path whose post-treatment second differences equal $\pm M$ and whose pure pre-treatment second differences with nonzero FLCI weight equal $\pm(M + \nu\sigma_0)$, with all signs chosen to maximize the bias of that FLCI's center. The true effect is zero. Each cell uses $2{,}000$ coefficient draws from $N(\vec\beta, \lambda^2\Sigma_0)$, and every procedure uses this true covariance, so the coverage of the FLCI is also known exactly: with center bias $b$, standard error $\sigma$, and half-length $\chi$, it equals $\Phi((\chi - b)/\sigma) + \Phi((\chi + b)/\sigma) - 1$.

At $M = 0$ and $\lambda = 1$ (Figure~\ref{fig:app-cal-length}(b)), the full-path FLCI covers the true effect at the nominal level when $\nu = 0$ and loses coverage steadily as $\nu$ grows. Conditional and C-LF impose only the post-only restriction and stay at the nominal level throughout.

In the calibrated design under parallel trends, panels (c) and (d) of Figure~\ref{fig:app-cal-length} compare the lengths of the full-path FLCI and the post-only FLCI, the interval that fixes the pre-treatment weights at $v^{\mathrm P}$ (Section~\ref{sec:flci-center}). Below the threshold of Appendix~\ref{app:flci-program} the full-path interval is shorter, and the two lengths coincide once $M$ exceeds it. In the stress design at $M = 0$, the post-only FLCI also keeps nominal coverage, because by construction its center does not use the pure pre-treatment second differences.

\begin{figure}[!h]
\centering
\subfigure[$\lambda = 1$, $M = 0.1\,\sigma_0$]{\includegraphics{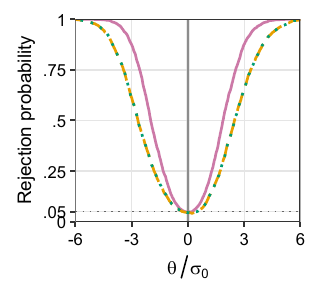}}\hfill
\subfigure[$\lambda = 1$, $M = 0.5\,\sigma_0$]{\includegraphics{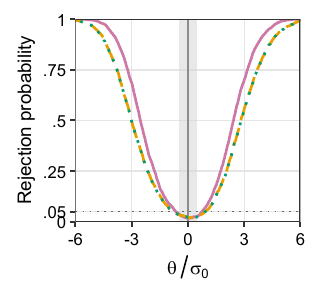}}\hfill
\subfigure[$\lambda = 1$, $M = 2\,\sigma_0$]{\includegraphics{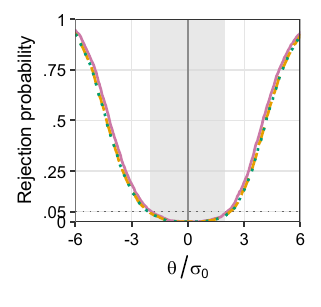}}\\[0.2em]
\subfigure[$\lambda = 0.5$, $M = 0.1\,\sigma_0$]{\includegraphics{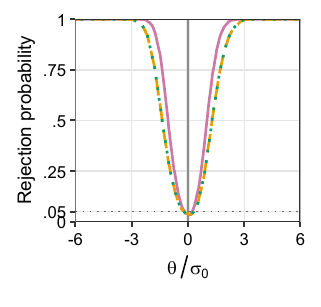}}\hfill
\subfigure[$\lambda = 0.5$, $M = 0.5\,\sigma_0$]{\includegraphics{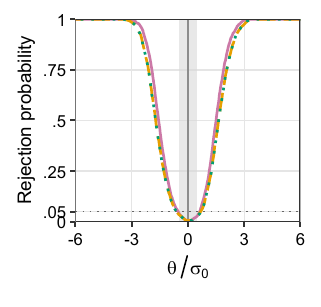}}\hfill
\subfigure[$\lambda = 0.5$, $M = 2\,\sigma_0$]{\includegraphics{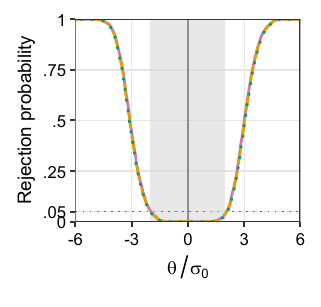}}\\[0.2em]
\includegraphics{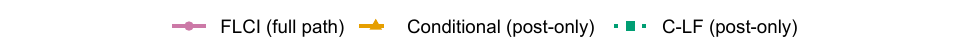}
\caption{Calibrated FLCI design: rejection curves under parallel trends}
\label{fig:app-cal-profiles}
\vspace{0.5em}
\parbox{\linewidth}{\footnotesize
    \textit{Notes:} Probability of rejecting the hypothesized effect $\theta$, in $\sigma_0$ units with the true effect at zero. The gray band is the population identified set $[-M, M]$ and the dotted line is 5 percent; $2{,}000$ confidence sets per panel.
}
\end{figure}

\begin{figure}[!h]
\centering
\subfigure[Mean length, $\lambda = 1$]{\includegraphics{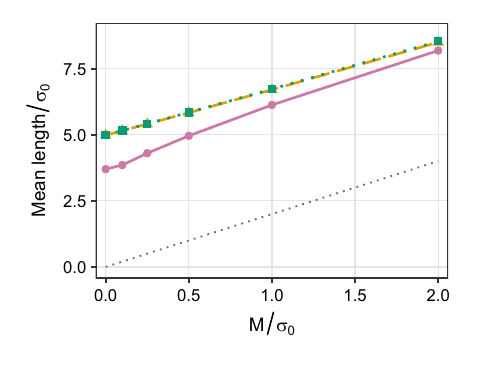}}\hfill
\subfigure[Coverage in the stress design, $\lambda = 1$]{\includegraphics{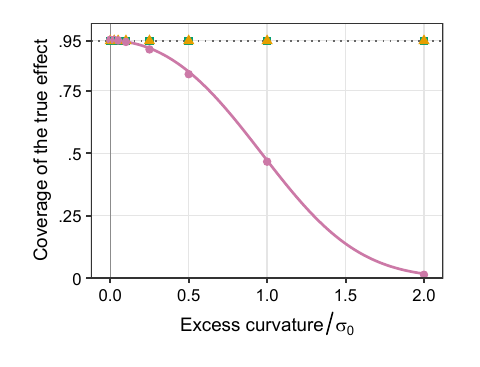}}\\[0.2em]
\includegraphics{figures/pdf/figB_flci_legend.pdf}\\[0.6em]
\subfigure[Full-path and post-only FLCI, $\lambda = 1$]{\includegraphics{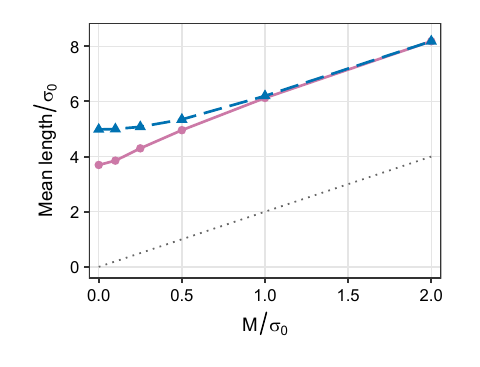}}\hfill
\subfigure[Full-path and post-only FLCI, $\lambda = 0.1$]{\includegraphics{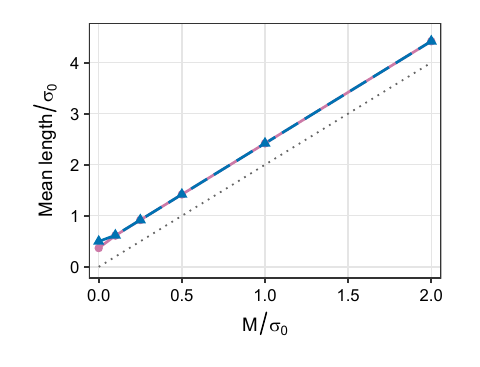}}\\[0.2em]
\includegraphics{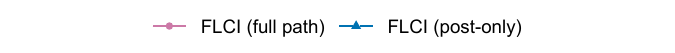}
\caption{Calibrated and stress designs: length and coverage of the FLCI}
\label{fig:app-cal-length}
\vspace{0.5em}
\parbox{\linewidth}{\footnotesize
    \textit{Notes:} Panels (a), (c), and (d): mean total length in $\sigma_0$ units against $M/\sigma_0$ under parallel trends in the calibrated design; the dotted line is the width $2M/\sigma_0$ of the population identified set; in (a) Conditional and C-LF overlap. Panel (b): coverage of the true effect at $M = 0$ against the excess of the pure pre-treatment curvature over $M$, in $\sigma_0$ units; the solid curve is the exact coverage of the full-path FLCI, markers are empirical coverage, and the dotted line marks 95 percent. $2{,}000$ confidence sets or draws per point.
}
\end{figure}

\subsection{Setup of Designs~I to~IV}
\label{app:sim-setup}

The disturbances $\varepsilon_{it}$ follow a stationary Gaussian first-order autoregression within each unit, with coefficient $0.25$ and unit variance, independently across units. The same draws are scaled by $\lambda$ at every noise level. Design~III replaces $\varepsilon_{it}$ by $\sqrt{\pi_S}\,\eta_{\mathrm{state}(i),t} + \sqrt{1 - \pi_S}\,\varepsilon_{it}$, where the state component $\eta$ has the same form and is shared by the units of a state, so that it carries a share $\pi_S = 0.05$ of the disturbance variance.

The covariance, labeled CR1, is the cluster-robust estimate of Appendix~\ref{app:covariance}, with residuals centered at cohort-by-period cell means and the common finite-sample factor of \eqref{eq:A-vcov}. It is clustered by unit in Designs~I, II, and~IV and by state in Design~III, and it is used as computed, without a ridge. Designs~III and~IV also compute every confidence set with two further covariances from the same draws. CR2 drops the CR1 factor and instead divides each cluster's score by $\sqrt{1 - p_c}$, with $p_c$ the cluster's share of the units of its cohort \citep{bell2002bias}. The oracle is the population covariance of the coefficient vector implied by the disturbance process.

\subsection{Design I: heterogeneous cohort paths}
\label{app:sim-design1}

Under SD2 and SD3 the cohort block-bias paths are quadratic in relative period $s$. Cohort $\mathcal G_g$'s raw path is
\begin{equation}
\Delta^{\mathrm{raw}}_{\mathcal G_g,s} = M_C\Bigl[\frac{s(s-1)}{2} + a_g\,s\Bigr],
\label{eq:B-sd-path}
\end{equation}
which is then normalized to a zero mean over the cohort's pre-treatment periods, the imputation reference of Section~\ref{sec:blockbias}. Every second difference equals $M_C$, so each cohort path lies on the boundary of the cohort SD restriction. The slopes are $a = (1.90, 1.90, 1.28)$ under SD2 and $a = (3.31, 3.31, -1.05)$ under SD3 for the cohorts adopting in periods 4, 6, and 7.

In the RM DGPs, each cohort's block bias alternates in sign from one period to the next. Each cohort has a loading, and its RM benchmark is the loading times $0.256$ outcome units. In absolute value, every pre-treatment change equals the benchmark and every post-treatment change equals $\bar M_C$ times the benchmark, so each cohort path lies on the boundary of the cohort RM restriction at $\bar M_C = 1$. RM1 uses loadings $(1, 1, 0.55)$ with the same phase for all cohorts. RM2 keeps the loadings and reverses the phase of the late cohort, so cohort deviations offset in the aggregate. RM3 uses loadings $(0.25, 1, 3)$, so the late cohort dominates. Figure~\ref{fig:sim-paths} shows the block-bias paths of the six DGPs with the induced aggregate path.

\begin{figure}[!h]
\centering
\subfigure[SD1]{\includegraphics{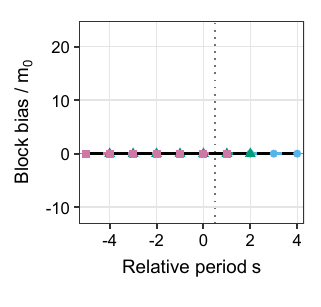}}\hfill
\subfigure[SD2]{\includegraphics{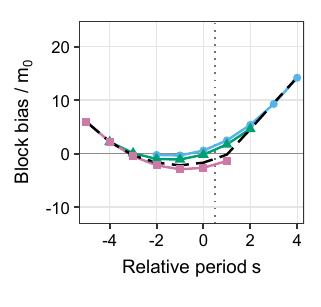}}\hfill
\subfigure[SD3]{\includegraphics{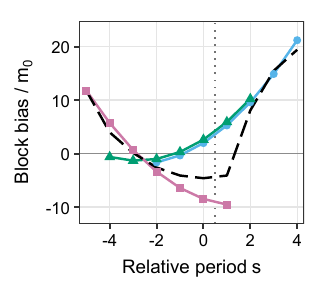}}\\[0.2em]
\subfigure[RM1]{\includegraphics{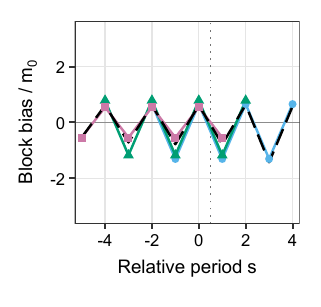}}\hfill
\subfigure[RM2]{\includegraphics{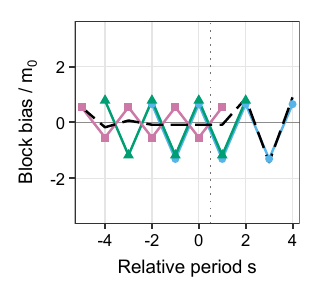}}\hfill
\subfigure[RM3]{\includegraphics{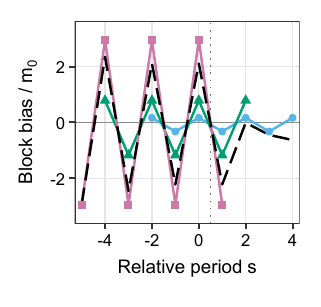}}\\[0.2em]
\includegraphics{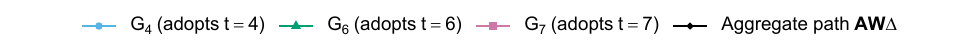}
\caption{Population block biases of the Design I DGPs}
\label{fig:sim-paths}
\vspace{0.5em}
\parbox{\linewidth}{\footnotesize
    \textit{Notes:} Population block biases $\Delta_{\mathcal G_g,s}$ of the imputation estimator for the three cohorts, in $m_0$ units against the relative period $s$, and the induced aggregate bias path $\mathbf A\mathbf W\vec\Delta$ (black dashed). Panels (a) to (c): DGPs SD1 to SD3; panels (d) to (f): DGPs RM1 to RM3. The dotted vertical line separates $s = 0$ from $s = 1$.
}
\end{figure}

Figure~\ref{fig:sim-profiles} shows the C-LF rejection curves behind Table~\ref{tab:sim-design1} at $\lambda = 1$ and $\lambda = 0.25$. Under SD1 the curve under cohort inference rejects at about the nominal rate at both endpoints of its identified set and rarely inside it. Under SD2 and SD3 the curves under aggregate inference shift to the right of the curve under cohort inference, and the shift persists as the noise falls.

Figures~\ref{fig:app-clf-width} and~\ref{fig:app-flci-width} report the mean widths and centers of the C-LF confidence sets and the FLCIs. The width-matched aggregate confidence sets are modestly shorter than the cohort ones, and the same-value aggregate sets are shorter still. Under SD2 and SD3 the mean centers under aggregate inference lie above the center under cohort inference, and under SD3 above the identified set under cohort inference.

\begin{figure}[!h]
\centering
\subfigure[SD1, $\lambda = 1$]{\includegraphics{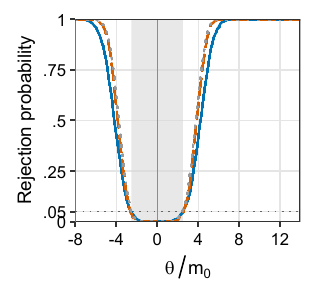}}\hfill
\subfigure[SD2, $\lambda = 1$]{\includegraphics{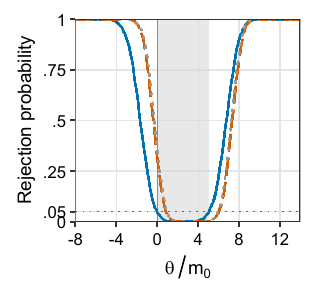}}\hfill
\subfigure[SD3, $\lambda = 1$]{\includegraphics{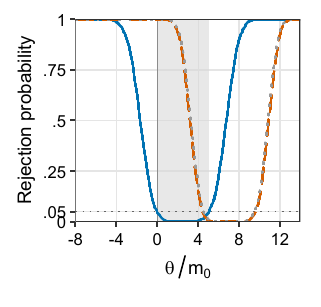}}\\[0.2em]
\subfigure[SD1, $\lambda = 0.25$]{\includegraphics{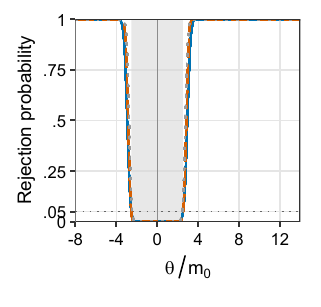}}\hfill
\subfigure[SD2, $\lambda = 0.25$]{\includegraphics{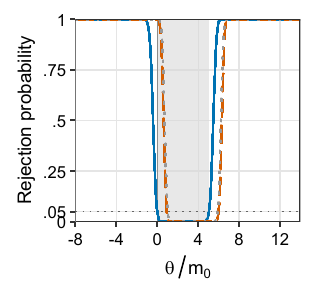}}\hfill
\subfigure[SD3, $\lambda = 0.25$]{\includegraphics{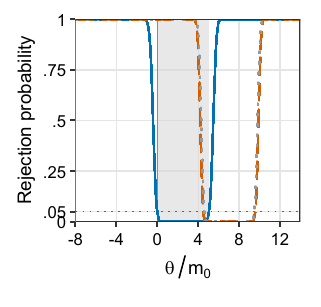}}\\[0.2em]
\includegraphics{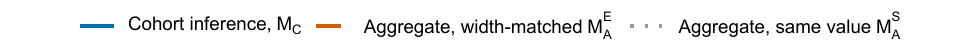}
\caption{Design I: C-LF rejection curves}
\label{fig:sim-profiles}
\vspace{0.5em}
\parbox{\linewidth}{\footnotesize
    \textit{Notes:} Empirical rejection probabilities of the C-LF confidence sets for a hypothesized value $\theta$, from $R = 1{,}000$ replications of Design~I. Cohort inference uses $M_C = m_0$ (blue solid); aggregate inference uses the width-matched bound $M_A^E = 1.04\,m_0$ (vermilion dashed) or the same-value bound $M_A^S = M_C$ (gray dotted). The gray band is the population identified set under cohort inference, the dotted horizontal line marks 5 percent, and the horizontal axis is $\theta/m_0$ with the true effect at zero.
}
\end{figure}

\begin{figure}[!h]
\centering
\subfigure[Width, all DGPs]{\includegraphics{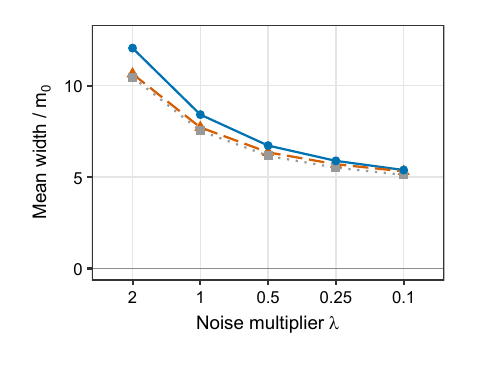}}\hfill
\subfigure[Center, SD1]{\includegraphics{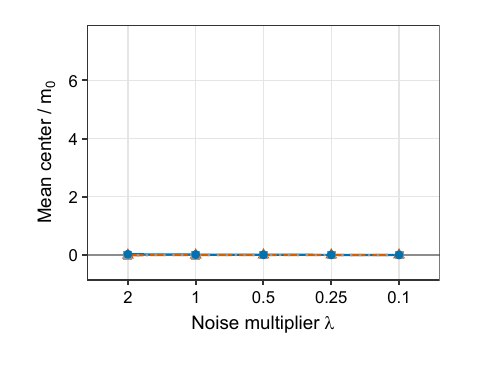}}\\[0.2em]
\subfigure[Center, SD2]{\includegraphics{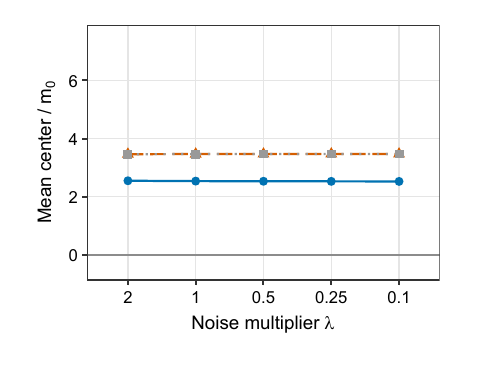}}\hfill
\subfigure[Center, SD3]{\includegraphics{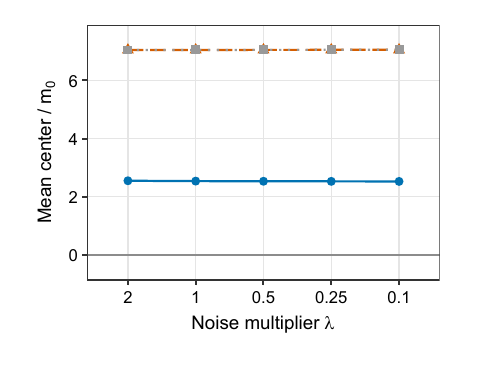}}\\[0.2em]
\includegraphics{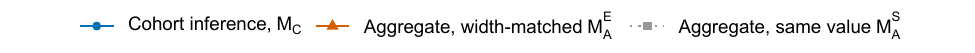}
\caption{Design I: C-LF confidence-set width and center}
\label{fig:app-clf-width}
\vspace{0.5em}
\parbox{\linewidth}{\footnotesize
    \textit{Notes:} Mean total width and mean center of the C-LF confidence sets in $m_0$ units, by noise multiplier $\lambda$ from high to low noise; series as in Figure~\ref{fig:sim-profiles}. The width is the same for the three DGPs. The two aggregate calibrations have the same center.
}
\end{figure}

\begin{figure}[!h]
\centering
\subfigure[Width, all DGPs]{\includegraphics{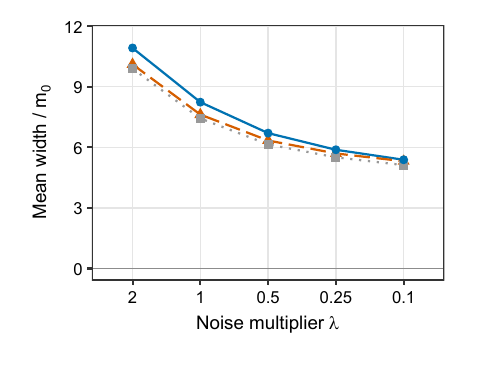}}\hfill
\subfigure[Center, SD1]{\includegraphics{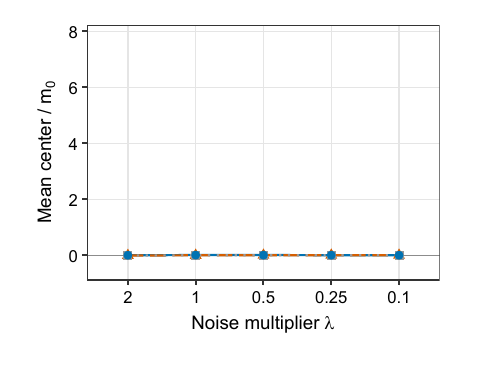}}\\[0.2em]
\subfigure[Center, SD2]{\includegraphics{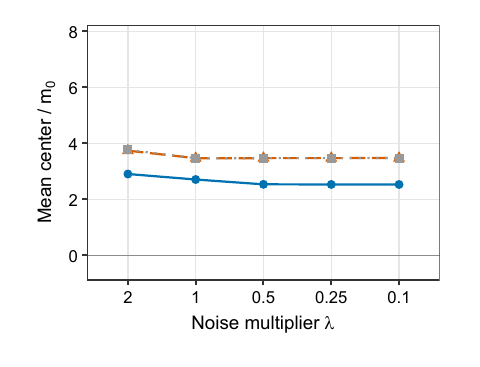}}\hfill
\subfigure[Center, SD3]{\includegraphics{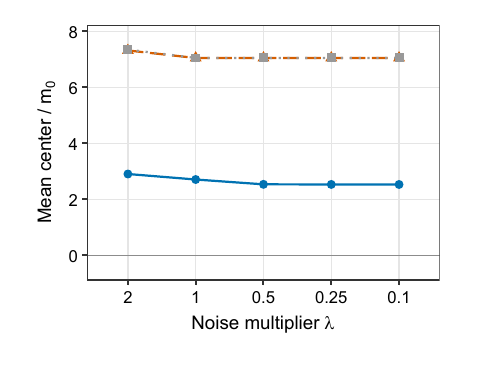}}\\[0.2em]
\includegraphics{figures/pdf/figB_design1_width_legend.pdf}
\caption{Design I: FLCI width and center}
\label{fig:app-flci-width}
\vspace{0.5em}
\parbox{\linewidth}{\footnotesize
    \textit{Notes:} As Figure~\ref{fig:app-clf-width}, for the full-path FLCI.
}
\end{figure}

Table~\ref{tab:app-design1-rm} and Figure~\ref{fig:app-rm-profiles} report C-LF results under RM. The RM confidence sets can be disconnected, because the restriction is a union over the location and sign of the binding pre-treatment change, so I do not assign them a center.

\begin{table}[!h]
\centering
\caption{Design I: minimum coverage of the C-LF confidence sets under RM, percent}
\label{tab:app-design1-rm}
\footnotesize
\begin{tabular}{@{}llrrr@{}}
\toprule
DGP & $\lambda$ & Cohort, $\bar M_C = 1$ & Aggregate, $\bar M_A^S = 1$ & Aggregate, $\bar M_A^E$ \\
\midrule
\multirow{5}{*}{RM1 ($\bar M_A^E = 1.29$)} & 2 & 100.0 & 99.1$^{\dagger}$ & 100.0$^{\dagger}$
\\
 & 1 & 100.0 & 95.9$^{\dagger}$ & 99.9$^{\dagger}$
\\
 & 0.5 & 99.8 & 60.5$^{\dagger}$ & 99.6$^{\dagger}$
\\
 & 0.25 & 99.7 & 2.7$^{\dagger}$ & 98.3$^{\dagger}$
\\
 & 0.1 & 99.7 & 0.0$^{\dagger}$ & 89.1$^{\dagger}$
\\
\midrule
\multirow{5}{*}{RM2 ($\bar M_A^E = 2.52$)} & 2 & 100.0 & 89.4$^{\dagger}$ & 100.0$^{\dagger}$
\\
 & 1 & 100.0 & 20.1$^{\dagger}$ & 99.5$^{\dagger}$
\\
 & 0.5 & 99.8 & 0.0$^{\dagger}$ & 92.2$^{\dagger}$
\\
 & 0.25 & 99.8 & 0.0$^{\dagger}$ & 85.6$^{\dagger}$
\\
 & 0.1 & 99.8 & 0.0$^{\dagger}$ & 58.4$^{\dagger}$
\\
\midrule
\multirow{5}{*}{RM3 ($\bar M_A^E = 0.49$)} & 2 & 100.0 & 100.0 & 94.5$^{\dagger}$
\\
 & 1 & 100.0 & 100.0 & 83.8$^{\dagger}$
\\
 & 0.5 & 100.0 & 100.0 & 50.6$^{\dagger}$
\\
 & 0.25 & 99.8 & 100.0 & 6.1$^{\dagger}$
\\
 & 0.1 & 99.7 & 100.0 & 0.0$^{\dagger}$
\\
\bottomrule
\end{tabular}

\vspace{0.5em}
\parbox{\linewidth}{\footnotesize
    \textit{Notes:} Minimum pointwise coverage over the population identified set under cohort inference, in percent, for C-LF under RM. Rows: DGP (with its width-matched $\bar M_A^E$) and noise multiplier $\lambda$. Columns: cohort inference at $\bar M_C = 1$, aggregate inference at $\bar M_A^S = 1$, and aggregate inference at $\bar M_A^E$. A dagger marks cells in which the aggregate bias path of the DGP violates the aggregate restriction at that bound. Nominal coverage 95 percent.
}
\end{table}

\begin{figure}[!h]
\centering
\subfigure[RM1, $\lambda = 1$]{\includegraphics{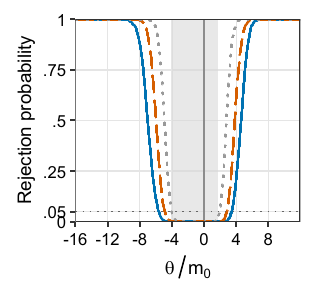}}\hfill
\subfigure[RM2, $\lambda = 1$]{\includegraphics{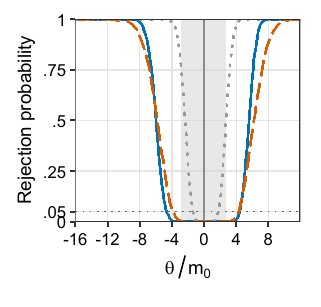}}\hfill
\subfigure[RM3, $\lambda = 1$]{\includegraphics{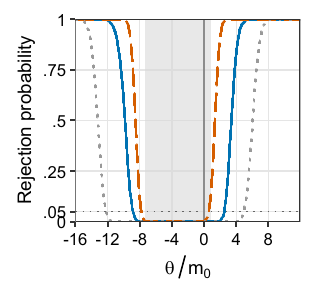}}\\[0.2em]
\subfigure[RM1, $\lambda = 0.1$]{\includegraphics{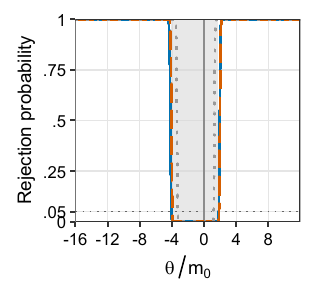}}\hfill
\subfigure[RM2, $\lambda = 0.1$]{\includegraphics{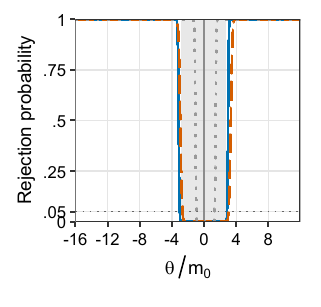}}\hfill
\subfigure[RM3, $\lambda = 0.1$]{\includegraphics{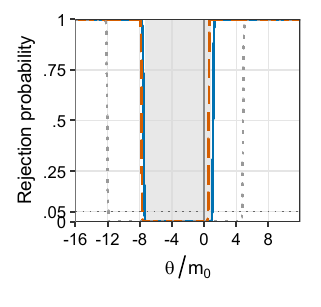}}\\[0.2em]
\includegraphics{figures/pdf/fig6_design1_profiles_legend.pdf}
\caption{Design I: C-LF rejection curves under RM}
\label{fig:app-rm-profiles}
\vspace{0.5em}
\parbox{\linewidth}{\footnotesize
    \textit{Notes:} Empirical rejection probabilities of the C-LF confidence sets for a hypothesized $\theta$, from $R = 1{,}000$ replications. Cohort inference uses $\bar M_C = 1$ (blue solid); aggregate inference uses $\bar M_A^E$ (vermilion dashed) or $\bar M_A^S = 1$ (gray dotted). The gray band is the population identified set under cohort inference; the horizontal axis is $\theta/m_0$.
}
\end{figure}

\subsection{Design II: balanced aggregation}
\label{app:sim-design2}

Design~II extends the Design~I panel to $T = 10$ with the same units and adoption dates. The balanced target averages the post-treatment relative periods $s = 1, \dots, 4$ observed for all three cohorts and forms the aggregate path on the window $s = -2, \dots, 4$ with cohort weights fixed across relative periods. Under SD3 the identified set under aggregate inference lies $1.7\,m_0$ above the one under cohort inference for the average treatment effect of all post-treatment periods and $0.7\,m_0$ below it for the balanced target. Under SD2 the gaps are $+0.4\,m_0$ and $-0.1\,m_0$.

\subsection{Design III: few treated clusters}
\label{app:sim-design3}

Design~III returns to the $T = 7$ panel of Design~I without the cohort adopting in period 4 and keeps the block-bias paths of the other two cohorts. The dropped cohort's 100 units all lie in one state, so a state-clustered covariance would omit its sampling variation (Appendix~\ref{app:covariance}). The remaining $2{,}184$ units keep their states.

Table~\ref{tab:app-design3} reports the minimum coverage for all DGPs and noise levels. Figure~\ref{fig:app-oracle-clf} shows that under CR1 the cohort C-LF confidence sets are slightly too narrow, rejecting at the boundary of the identified set above the nominal rate.

\begin{table}[!h]
\centering
\caption{Design III: minimum coverage with the estimated and the oracle covariances, percent}
\label{tab:app-design3}
\scriptsize
\setlength{\tabcolsep}{3.2pt}
\begin{tabular}{@{}llrrrrrrrrrrrr@{}}
\toprule
 & & \multicolumn{3}{c}{Cohort C-LF} & \multicolumn{3}{c}{Cohort FLCI} & \multicolumn{3}{c}{Aggr.\ C-LF, $M_A^E$} & \multicolumn{3}{c}{Aggr.\ FLCI, $M_A^E$} \\
\cmidrule(lr){3-5} \cmidrule(lr){6-8} \cmidrule(lr){9-11} \cmidrule(lr){12-14}
DGP & $\lambda$ & CR1 & CR2 & Oracle & CR1 & CR2 & Oracle & CR1 & CR2 & Oracle & CR1 & CR2 & Oracle \\
\midrule
\multirow{5}{*}{SD1} & 2 & 92.4 & 94.4 & 97.1 & 86.3 & 87.2 & 96.8 & 92.9 & 94.3 & 96.1 & 92.2 & 93.2 & 96.3
\\
 & 1 & 92.0 & 94.1 & 96.5 & 89.7 & 91.2 & 96.4 & 93.0 & 94.3 & 95.7 & 93.3 & 94.0 & 96.0
\\
 & 0.5 & 90.8 & 93.3 & 95.7 & 90.9 & 92.6 & 95.5 & 92.3 & 93.5 & 95.0 & 91.8 & 93.4 & 94.6
\\
 & 0.25 & 90.8 & 93.0 & 95.4 & 90.5 & 92.6 & 95.1 & 92.1 & 93.5 & 94.9 & 91.8 & 93.3 & 94.6
\\
 & 0.1 & 90.8 & 92.9 & 95.4 & 90.5 & 92.5 & 95.1 & 92.1 & 93.5 & 94.9 & 91.8 & 93.3 & 94.6
\\
\midrule
\multirow{5}{*}{SD2} & 2 & 92.4 & 94.4 & 97.1 & 81.8 & 82.2 & 94.2 & 91.6 & 93.5 & 95.2 & 86.4 & 87.5 & 93.0
\\
 & 1 & 92.0 & 94.1 & 96.5 & 86.6 & 87.6 & 94.7 & 89.2 & 91.2 & 93.3 & 86.4 & 87.6 & 91.8
\\
 & 0.5 & 90.8 & 93.3 & 95.7 & 89.5 & 90.9 & 95.0 & 83.0 & 85.5 & 87.8 & 81.8 & 83.8 & 87.4
\\
 & 0.25 & 90.8 & 93.0 & 95.4 & 90.5 & 92.1 & 95.1 & 65.8 & 69.8 & 73.0 & 65.0 & 68.6 & 72.6
\\
 & 0.1 & 90.8 & 92.9 & 95.4 & 90.5 & 92.5 & 95.1 & 12.4 & 14.7 & 14.4 & 11.9 & 14.3 & 13.9
\\
\midrule
\multirow{5}{*}{SD3} & 2 & 92.4 & 94.4 & 97.1 & 81.8 & 82.2 & 94.2 & 80.5 & 83.8 & 87.1 & 66.1 & 68.6 & 76.9
\\
 & 1 & 92.0 & 94.1 & 96.5 & 86.6 & 87.6 & 94.7 & 52.8 & 58.3 & 61.3 & 38.6 & 41.6 & 47.0
\\
 & 0.5 & 90.8 & 93.3 & 95.7 & 89.5 & 90.9 & 95.0 & 7.2 & 9.2 & 8.3 & 6.2 & 7.6 & 7.5
\\
 & 0.25 & 90.8 & 93.0 & 95.4 & 90.5 & 92.1 & 95.1 & 0.0 & 0.0 & 0.0 & 0.0 & 0.0 & 0.0
\\
 & 0.1 & 90.8 & 92.9 & 95.4 & 90.5 & 92.5 & 95.1 & 0.0 & 0.0 & 0.0 & 0.0 & 0.0 & 0.0
\\
\bottomrule
\end{tabular}

\vspace{0.5em}
\parbox{\linewidth}{\footnotesize
    \textit{Notes:} Minimum pointwise coverage over the population identified set under cohort inference, in percent, in Design~III. Rows: DGP and noise multiplier $\lambda$. Columns: C-LF confidence sets and FLCIs under cohort inference and under width-matched aggregate inference ($M_A^E = 1.045\,M_C$), with the covariances CR1, CR2, and Oracle of Table~\ref{tab:sim-design34}. Under SD2 and SD3 the aggregate bias path of the DGP violates the aggregate restriction. Nominal coverage 95 percent.
}
\end{table}

\begin{figure}[!h]
\centering
\subfigure[$\lambda = 1$]{\includegraphics{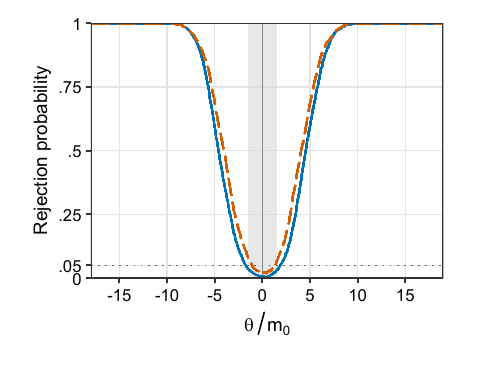}}\hfill
\subfigure[$\lambda = 0.1$]{\includegraphics{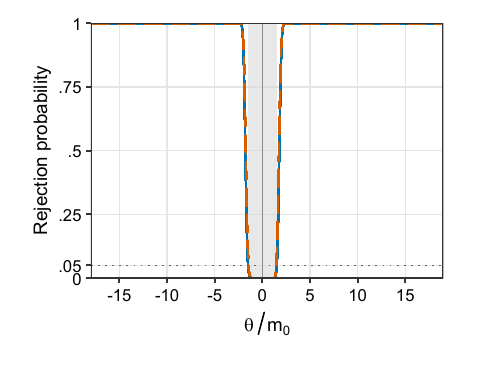}}\\[0.2em]
\includegraphics{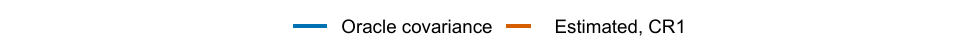}
\caption{Design III: rejection curves of the cohort C-LF confidence sets with the estimated and the oracle covariance}
\label{fig:app-oracle-clf}
\vspace{0.5em}
\parbox{\linewidth}{\footnotesize
    \textit{Notes:} Empirical rejection probabilities of the cohort C-LF confidence sets under SD1 from $R = 1{,}000$ replications, using the estimated state-clustered covariance CR1 (vermilion dashed) or the oracle covariance (blue solid). The gray band is the population identified set under cohort inference at $M_C = m_0$; the dotted line marks 5 percent; the horizontal axis is $\theta/m_0$.
}
\end{figure}

\subsection{Design IV: more and smaller cohorts}
\label{app:sim-design4}

In Design~IV the cohorts adopt in periods $\{5, 12, 19\}$ for $G = 3$, $\{5, 8, 12, 16, 19\}$ for $G = 5$, $\{5, 6, 8, 10, 11, 13, 14, 16, 18, 19\}$ for $G = 10$, and $5$ through $19$ for $G = 15$. The balanced target averages $s = 1, 2, 3$ and forms the aggregate path on $s = -3, \dots, 3$. For this target, aggregate inference is width-matched at $M_A^E = 1.000, 1.000, 1.015$, and $1.021$ times $M_C$ for $G = 3, 5, 10$, and $15$.

Table~\ref{tab:app-design4} reports minimum coverage, calibrations, and length ratios for the average treatment effect of all post-treatment periods under the three covariances. Figure~\ref{fig:app-c-clf} shows the rejection curves for the balanced target under CR1 for the largest and the smallest cohorts. For C-LF the curves under cohort and aggregate inference stay close as cohorts shrink. With 2 units per cohort, the cohort FLCI rejects far more at the boundary of the identified set than the aggregate FLCI.

\begin{figure}[!h]
\centering
\subfigure[C-LF, $G = 3$ (10 units per cohort)]{\includegraphics{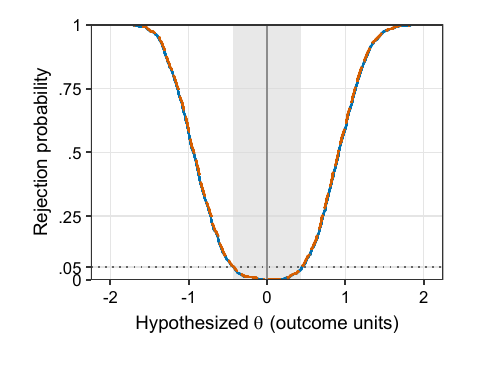}}\hfill
\subfigure[C-LF, $G = 15$ (2 units per cohort)]{\includegraphics{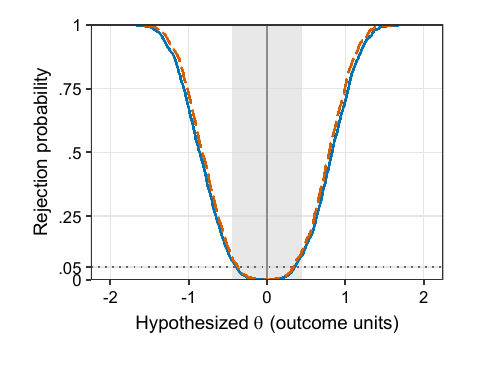}}\\[0.2em]
\subfigure[FLCI, $G = 3$ (10 units per cohort)]{\includegraphics{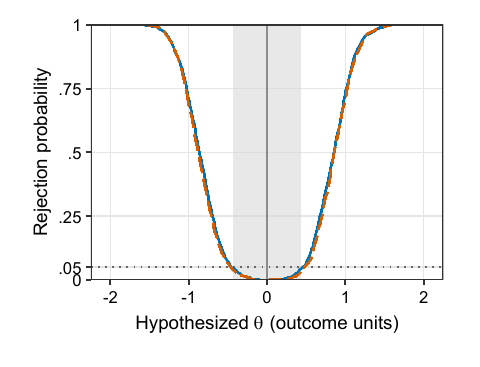}}\hfill
\subfigure[FLCI, $G = 15$ (2 units per cohort)]{\includegraphics{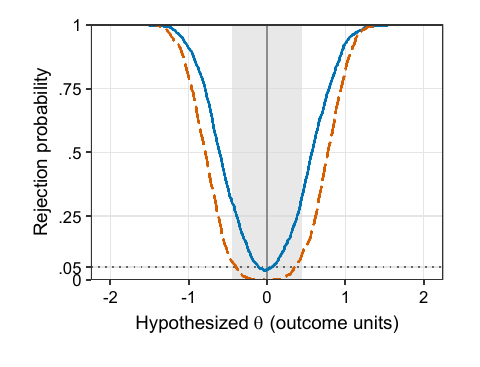}}\\[0.2em]
\includegraphics{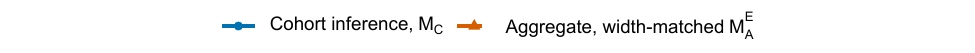}
\caption{Design IV: rejection curves for the balanced target with the largest and the smallest cohorts}
\label{fig:app-c-clf}
\vspace{0.5em}
\parbox{\linewidth}{\footnotesize
    \textit{Notes:} Empirical rejection probabilities for the balanced target under SD1, from $R = 1{,}000$ replications with the CR1 covariance. Cohort inference uses $M_C$ (blue solid) and aggregate inference the width-matched $M_A^E$ (vermilion dashed). The gray band is the population identified set under cohort inference and the dotted line marks 5 percent; the horizontal axis is in the outcome units of the simulated panel.
}
\end{figure}

\begin{table}[!h]
\centering
\caption{Design IV: coverage, calibration, and length for the average treatment effect of all post-treatment periods}
\label{tab:app-design4}
\footnotesize
\setlength{\tabcolsep}{4pt}
\begin{tabular}{@{}lllrrrrrr@{}}
\toprule
 & & & \multicolumn{2}{c}{C-LF coverage} & \multicolumn{2}{c}{FLCI coverage} & \multicolumn{2}{c}{Length ratio aggregate/cohort} \\
\cmidrule(lr){4-5} \cmidrule(lr){6-7} \cmidrule(lr){8-9}
$G$ (units) & $M_A^E/M_C$ & Covariance & Cohort & Aggregate & Cohort & Aggregate & C-LF & FLCI \\
\midrule
\multirow{3}{*}{3 (10)} & \multirow{3}{*}{1.041} & CR1 & 94.1 & 94.0 & 97.0 & 97.5 & 0.942 & 0.972
\\
 &  & CR2 & 94.4 & 94.2 & 97.0 & 97.5 & 0.942 & 0.972
\\
 &  & Oracle & 95.1 & 95.1 & 98.4 & 98.3 & 0.939 & 0.963
\\
\addlinespace[3pt]
\multirow{3}{*}{5 (6)} & \multirow{3}{*}{1.081} & CR1 & 93.4 & 93.7 & 94.2 & 96.6 & 0.941 & 0.982
\\
 &  & CR2 & 93.6 & 94.0 & 94.4 & 97.0 & 0.939 & 0.982
\\
 &  & Oracle & 94.7 & 94.1 & 98.2 & 98.8 & 0.937 & 0.960
\\
\addlinespace[3pt]
\multirow{3}{*}{10 (3)} & \multirow{3}{*}{1.119} & CR1 & 91.9 & 92.8 & 82.4 & 95.2 & 0.942 & 1.023
\\
 &  & CR2 & 93.7 & 94.1 & 82.1 & 96.7 & 0.937 & 1.032
\\
 &  & Oracle & 94.2 & 95.5 & 98.3 & 98.6 & 0.935 & 0.958
\\
\addlinespace[3pt]
\multirow{3}{*}{15 (2)} & \multirow{3}{*}{1.135} & CR1 & 89.1 & 89.0 & 67.7 & 90.4 & 0.946 & 1.077
\\
 &  & CR2 & 92.5 & 92.8 & 67.5 & 94.4 & 0.936 & 1.102
\\
 &  & Oracle & 95.0 & 94.4 & 98.4 & 98.5 & 0.936 & 0.957
\\
\bottomrule
\end{tabular}

\vspace{0.5em}
\parbox{\linewidth}{\footnotesize
    \textit{Notes:} Design~IV under SD1. Rows: number of cohorts $G$ (units per cohort in parentheses) and covariance (CR1, CR2, Oracle, as in Table~\ref{tab:sim-design34}). Columns: the width-matched bound $M_A^E/M_C$, the minimum coverage of the population identified set under cohort inference, in percent, for the C-LF confidence sets and FLCIs under cohort and aggregate inference, and the ratio of the mean length under aggregate inference to the mean length under cohort inference for each procedure. Nominal coverage 95 percent.
}
\end{table}

\clearpage
\section{Additional empirical results}
\label{app:empirical}
\setcounter{equation}{0}
\setcounter{figure}{0}
\setcounter{table}{0}

\subsection{Minimum wage and teen employment}
\label{app:emp-minwage}

Figure~\ref{fig:app-minwage-bycohort} shows the confidence sets for each cohort's average effect over its own post-treatment cells, at $M = 0$ and at $M = 0.01$. The post-only restriction extrapolates each cohort's own pre-treatment slope. It lowers the estimates for $\mathcal G_4$ and $\mathcal G_6$, whose paths rise before adoption, and raises the estimate for $\mathcal G_7$, whose path falls, so the corrected negative average comes from $\mathcal G_4$ and $\mathcal G_6$. At $M = 0.01$ only the confidence sets of $\mathcal G_6$ exclude zero, and the widest belong to $\mathcal G_4$, which has four post-treatment cells and the longest extrapolation.

\begin{figure}[!h]
\centering
\subfigure[$M = 0$]{\includegraphics{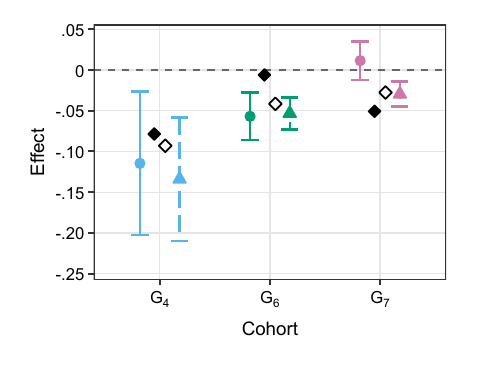}}\hfill
\subfigure[$M = 0.01$]{\includegraphics{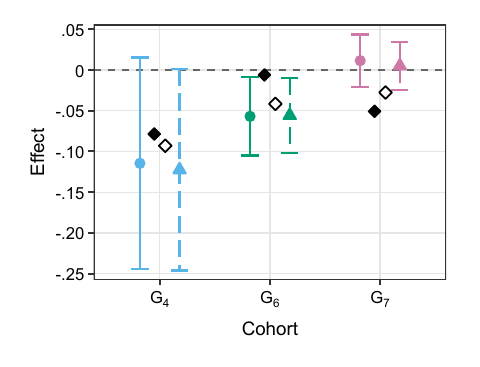}}\\[0.2em]
\includegraphics{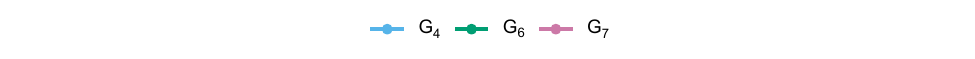}
\caption{Minimum wage: cohort-specific effects under cohort inference}
\label{fig:app-minwage-bycohort}
\vspace{0.5em}
\parbox{\linewidth}{\footnotesize
    \textit{Notes:} 95 percent confidence sets for the average effect of each cohort over its own post-treatment cells, in log points, under cohort inference with the county-clustered covariance; the confidence sets are the same for the two estimators. C-LF: solid whisker and circle; FLCI: dashed whisker and triangle; the marker is the center of the confidence set. Diamonds are the conventional cohort estimates, filled for imputation and hollow for \CSNYT{}; colors follow the cohort palette of Section~\ref{sec:illustration}.
}
\end{figure}

Figure~\ref{fig:minwage-rm} shows the C-LF confidence sets under RM for $\bar M \in \{0, 0.2, \dots, 1\}$. The confidence set under cohort inference excludes zero through $\bar M = 0.4$ and includes it from $\bar M = 0.6$. Under aggregate inference the imputation confidence sets are narrower and exclude zero through $\bar M = 0.6$, because the smoother aggregate path yields a smaller benchmark. The \CSNYT{} aggregate path is as volatile as the cohort paths, and its confidence sets are as wide as those under cohort inference.

\begin{figure}[!h]
\centering
\subfigure[Aggregate inference, imputation]{\includegraphics{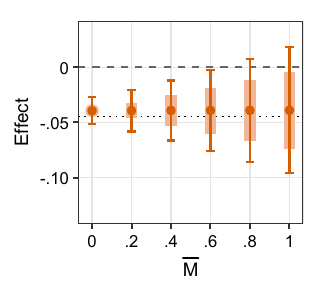}}\hfill
\subfigure[Aggregate inference, \CSNYT]{\includegraphics{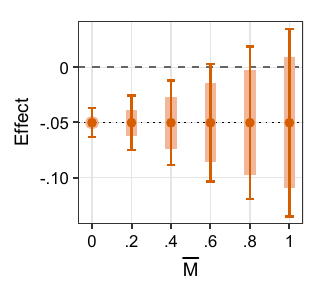}}\hfill
\subfigure[Cohort inference, both estimators]{\includegraphics{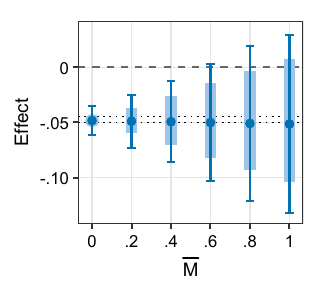}}\\[0.2em]
\includegraphics{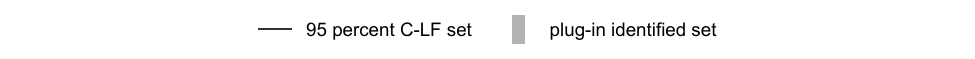}
\caption{Minimum wage: C-LF confidence sets under the RM restriction}
\label{fig:minwage-rm}
\vspace{0.5em}
\parbox{\linewidth}{\footnotesize
    \textit{Notes:} 95 percent C-LF confidence sets for the average effect over treated county-years, in log points, against the bound $\bar M$ of the RM restriction, with the covariance of Appendix~\ref{app:covariance} clustered by county. The light band is the plug-in identified set, a dot at $\bar M = 0$; the dashed line is zero, and the dotted lines are the conventional estimates, one per panel in (a) and (b) and both in (c). The marker is the center of the confidence set.
}
\end{figure}

Table~\ref{tab:app-minwage-curvature} reports the eight pure pre-treatment second differences of the cohort paths with their standard errors. Two benchmarks summarize them. The plug-in threshold is the smallest bound at which the fitted paths satisfy the full-path restriction, that is, the largest absolute second difference, $0.069$ log points per period for the cohort paths, the $\mathcal G_6$ second difference at $s = -2$. The test-based benchmark accounts for sampling error. It is the smallest $M$ at which a least-favorable max test of $|z_j| \leq M$ for every element $z_j$ of the pure pre-treatment second differences $z = \mathbf D_{\pre,\pre}\vec\beta_{\pre}$ does not reject at the 5 percent level. The statistic is $\max_j (|\hat z_j| - M)/\mathrm{se}_j$, where $\mathrm{se}_j$ is the standard error of $\hat z_j$, and the critical value is the $0.95$ quantile of $\max_j |Z_j|/\mathrm{se}_j$ for $Z \sim N(0, \hat V)$, with $\hat V$ the estimated covariance of $\hat z$. For the cohort paths the test-based benchmark is $0.024$.

\begin{table}[!h]
\centering
\caption{Minimum wage: pure pre-treatment second differences of the cohort block-bias paths}
\label{tab:app-minwage-curvature}
\vspace{0.3em}
\begin{tabular}{lcrr}
\toprule
Cohort & Period $s$ & Second difference & Standard error \\
\midrule
$\mathcal G_4$ & $0$ & $-0.015$ & $0.021$ \\
$\mathcal G_6$ & $-2$ & $0.069$ & $0.017$ \\
$\mathcal G_6$ & $-1$ & $-0.031$ & $0.017$ \\
$\mathcal G_6$ & $0$ & $0.006$ & $0.013$ \\
$\mathcal G_7$ & $-3$ & $0.041$ & $0.013$ \\
$\mathcal G_7$ & $-2$ & $-0.001$ & $0.013$ \\
$\mathcal G_7$ & $-1$ & $-0.017$ & $0.013$ \\
$\mathcal G_7$ & $0$ & $-0.038$ & $0.012$ \\
\bottomrule
\end{tabular}

\vspace{0.5em}
\parbox{\linewidth}{\footnotesize
    \textit{Notes:} Second differences $\hat\Delta_{g,s} - 2\hat\Delta_{g,s-1} + \hat\Delta_{g,s-2}$ of the fitted block-bias paths over pre-treatment cells only, in log points per period, with standard errors from the county-clustered covariance of Appendix~\ref{app:covariance}; the values are the same for the two estimators. The period $s$ is the last of the three cells.
}
\end{table}

Figure~\ref{fig:app-minwage-gamma} shows the weights $\gamma$ of the full-path FLCI on the eight second differences against $M$, for the average effect and for the first post-treatment effect. All weights are zero from $M = 0.015$ for the average effect and from $M = 0.011$ for the first post-treatment effect, the thresholds of Section~\ref{sec:flci-center} beyond which the FLCI equals the post-only FLCI. For the average effect, the FLCI center moves from $-0.066$ at $M = 0$ to $-0.048$ at $M = 0.005$, as the weights on $\mathcal G_6$ and $\mathcal G_7$ vanish.

\begin{figure}[!h]
\centering
\subfigure[Average effect]{\includegraphics{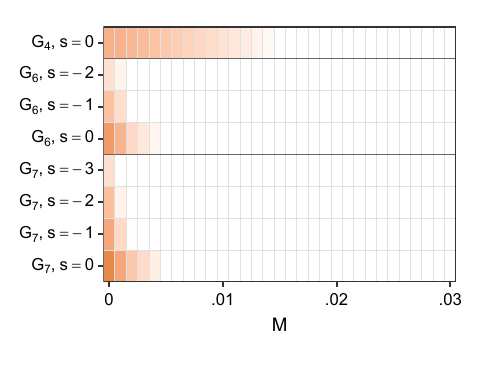}}\hfill
\subfigure[First post-treatment effect]{\includegraphics{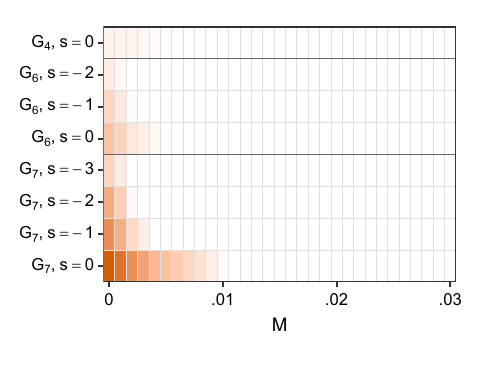}}\\[0.2em]
\includegraphics{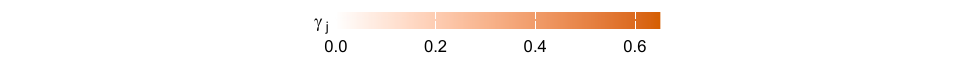}
\caption{Minimum wage: FLCI weights on the pure pre-treatment second differences}
\label{fig:app-minwage-gamma}
\vspace{0.5em}
\parbox{\linewidth}{\footnotesize
    \textit{Notes:} The weight $\gamma_j$ of the full-path FLCI on each pure pre-treatment second difference, by cohort and the relative period $s$ at which the second difference ends (rows), against the bound $M$ in log points per period (columns), under cohort inference with the county-clustered covariance. White cells are zero weights.
}
\end{figure}

Figure~\ref{fig:app-minwage-state} shows the confidence sets with a state-clustered covariance. Without $\mathcal G_4$, the conventional estimates are $-0.031$ for imputation and $-0.034$ for \CSNYT{}, and every C-LF confidence set includes zero under either covariance. State clustering multiplies the conventional standard errors of the target by $4.2$ for imputation and $2.0$ for \CSNYT{}. Yet the cohort C-LF confidence sets are at most 5 percent longer than under county clustering on the same sample, and the aggregate C-LF confidence sets change by at most a fifth. The state-level variation that inflates the conventional standard error is largely common to a cohort's pre-treatment and post-treatment cells and drops out of the corrected estimator. Under imputation, the post-only center has a standard error of $0.010$ under both clusterings, against $0.007$ and $0.031$ for the conventional estimator. The FLCI responds more, because its weights depend on the covariance. At $M = 0$ its confidence set under cohort inference is $[-0.022, -0.004]$ under state clustering, against $[-0.053, -0.027]$ under county clustering.

\begin{figure}[!h]
\centering
\subfigure[Aggregate inference, imputation]{\includegraphics{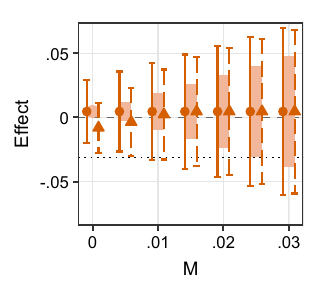}}\hfill
\subfigure[Aggregate inference, \CSNYT]{\includegraphics{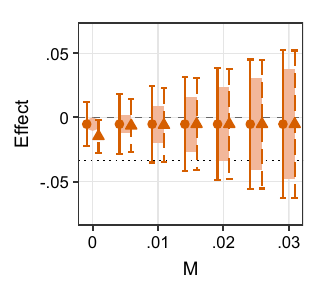}}\hfill
\subfigure[Cohort inference, both estimators]{\includegraphics{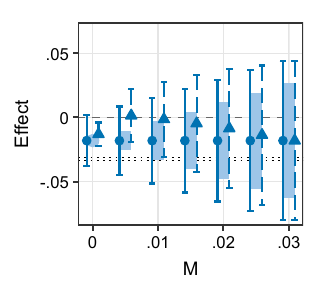}}\\[0.2em]
\includegraphics{figures/pdf/fig7_minwage_sd_legend.pdf}
\caption{Minimum wage without the single-state cohort: SD confidence sets with the state-clustered covariance}
\label{fig:app-minwage-state}
\vspace{0.5em}
\parbox{\linewidth}{\footnotesize
    \textit{Notes:} As Figure~\ref{fig:minwage-sd}, for the sample without $\mathcal G_4$ (cohorts $\mathcal G_6$ and $\mathcal G_7$ and the never-treated counties, 2,184 counties) and with the covariance of Appendix~\ref{app:covariance} clustered by the 28 states. The target is the average effect over the treated county-years of this sample. The dotted lines are the conventional estimates of this sample.
}
\end{figure}

\subsection{Community health centers}
\label{app:emp-chc}

The panel comes from the county file of \citet{bailey2015publicmedicine} in the estimation sample of the original study's preferred specification, 3,044 counties observed annually from 1959 to 1988. The sample excludes fifteen never-treated counties missing a 1959 or 1960 covariate, although the estimators here use no covariates. Figure~\ref{fig:app-chc-timing} shows the nine adoption cohorts and the never-treated counties. Each county is weighted by its 1960 population, the weight of the original study, held constant over the panel. I replace the cell means of Section~\ref{sec:estimators}, the cohort counts in $\mathbf W$, and the target weights by their population-weighted counterparts, which preserves the identities of Section~\ref{sec:biasdecomposition}. For both estimators the covariance follows \eqref{eq:A-vcov} with weighted cell-mean residuals and county clusters. Cohort inference retains each cohort's full history, up to fourteen pre-treatment changes for the 1974 cohort.

\begin{figure}[!h]
\centering
\includegraphics{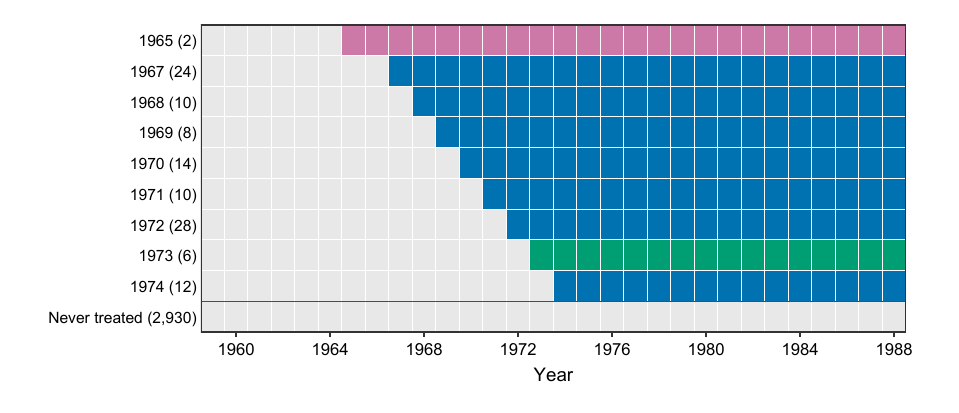}
\caption{Community health centers: treatment timing}
\label{fig:app-chc-timing}
\vspace{0.5em}
\parbox{\linewidth}{\footnotesize
    \textit{Notes:} Treatment status by cohort and year in the balanced panel of 3,044 counties. Rows are the nine adoption cohorts and the never-treated group, columns are years, as in Figure~\ref{fig:minwage}(a) but with equal row heights; the number of counties is printed after each cohort's adoption year. Shaded cells are treated county-years, with the 1965 and 1973 cohorts highlighted.
}
\end{figure}

In Figure~\ref{fig:chc}, the aggregate C-LF confidence sets for the balanced target are 7 to 12 percent shorter than the cohort C-LF confidence sets under imputation. At the same bound, the narrower plug-in identified set under aggregate inference accounts for $2$ of the $13$ deaths per 100,000 of length difference at $M = 0.33$ and $22$ of the $31$ at $M = 3.32$. The rest is the length added for sampling uncertainty, which is smaller under aggregate inference at both bounds.

The bias-reduced CR2 covariance of Section~\ref{sec:sim-designs} is computed for the working model in which a county's disturbance variance is inversely proportional to its population weight. Each county's score is divided by $\sqrt{1 - p_c}$, with $p_c$ its population share of its cohort group, which reaches $0.83$ for one county of the 1973 cohort. Under CR2 the center of the cohort FLCI still moves to $30.1$ at $M = 0.33$ and back to $2.2$ at $M = 0.83$.

Figure~\ref{fig:app-chc-mechanism} isolates the two smallest cohorts behind the center of the cohort FLCI. Panel (a) gives each cohort's plug-in threshold, the smallest bound at which its fitted pre-treatment path satisfies full-path SD. It is $85$ deaths per 100,000 for the two-county 1965 cohort and $139$ for the six-county 1973 cohort, against $8$ for the imputation aggregate path, and every fitted path's threshold exceeds the largest bound of Figure~\ref{fig:chc}, $3.32$. Panel (b) decomposes the adjustment $\gamma'\hat z$ of the FLCI center by cohort, and panel (c) shows the weights $\gamma$. At $M = 0.33$ the two cohorts hold 4.6 percent of $\|\gamma\|_1$, about their share of the treated population, but contribute $17.1$ deaths per 100,000, against $-6.4$ at $M = 0.83$. This swing is most of the change in the center from $28.8$ to $2.7$.

At $M = 0$ several adjustments have exactly the same precision. I select the one with the smallest Euclidean norm of $\gamma$, the same for both estimators as Proposition~\ref{prop:ci-equivalence} requires. Because that center is not the limit of the FLCI centers as $M$ decreases to zero, Figure~\ref{fig:chc} marks it with a cross and omits its interval.

\begin{figure}[!h]
\centering
\begin{minipage}[t]{0.49\linewidth}
\vspace{0pt}\centering
\subfigure[Plug-in thresholds by cohort]{\includegraphics{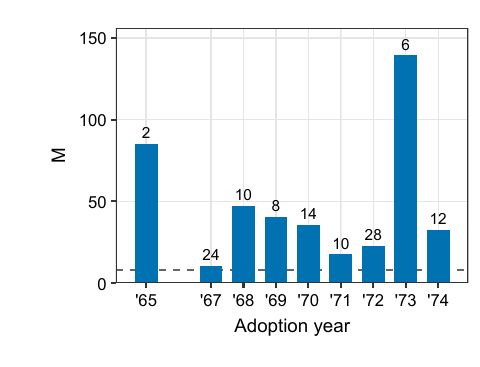}}\\
\subfigure[Adjustment $\gamma'\hat z$ by cohort]{\includegraphics{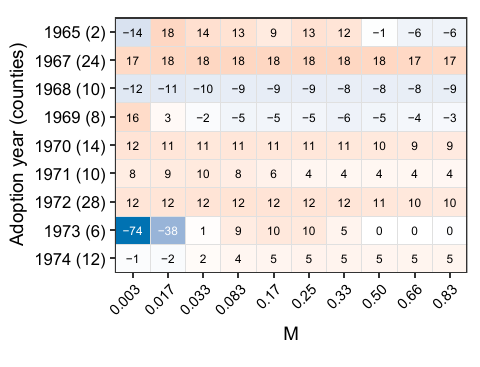}}\\[0.2em]
\includegraphics{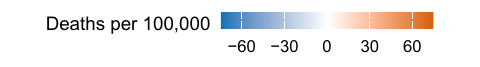}
\end{minipage}\hfill
\begin{minipage}[t]{0.49\linewidth}
\vspace{0pt}\centering
\subfigure[FLCI weights $\gamma$]{\includegraphics{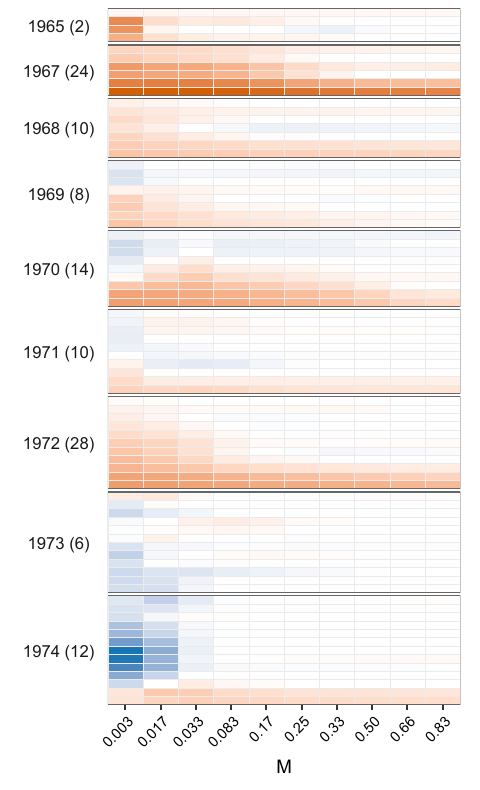}}\\[0.2em]
\includegraphics{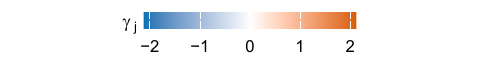}
\end{minipage}
\caption{Community health centers: the cohort FLCI adjustment}
\label{fig:app-chc-mechanism}
\vspace{0.5em}
\parbox{\linewidth}{\footnotesize
    \textit{Notes:} Panel (a): the smallest bound $M$ (deaths per 100,000 per period) at which each cohort's fitted pre-treatment path satisfies the full-path SD restriction, by adoption year, with the number of counties above each bar; the dashed line is the value for the aggregate path of the balanced target under imputation. Panel (b): each cohort's part of the adjustment $\gamma'\hat z$ of the center of the cohort FLCI for the balanced target, computed with the imputation estimator, in deaths per 100,000, by cohort (rows, number of counties in parentheses) and bound $M$ (columns). Panel (c): the weight $\gamma_j$ of the same FLCI on each pure pre-treatment second difference, by cohort and, within a cohort, by the relative period at which the second difference ends, earliest at the top.
}
\end{figure}

Figure~\ref{fig:app-chc-byperiod} shows the C-LF confidence sets for the effect at each period of the balanced window at $M = 0.33$. They include zero at every period and widen with the horizon. Under aggregate inference the plug-in identified set has width $M s(s+1)$, which at $s = 15$ accounts for $80$ of the roughly $200$ deaths per 100,000 of C-LF width. The \CSNYT{} confidence sets, which the figure omits, are nearly the same under aggregate inference, and the centers under cohort and aggregate inference differ by at most $3$ deaths per 100,000. The full-path FLCI is omitted because every fitted path violates its restriction at every bound of Figure~\ref{fig:chc}.

\begin{figure}[!h]
\centering
\subfigure[Aggregate inference, imputation]{\includegraphics{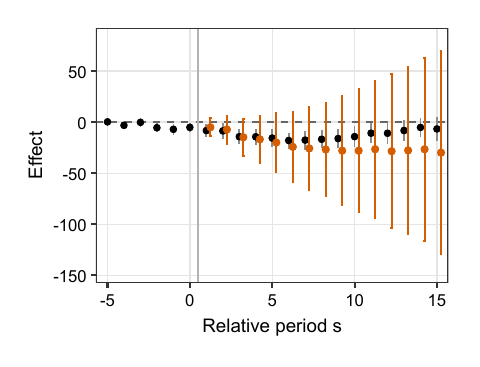}}\hfill
\subfigure[Cohort inference, imputation path]{\includegraphics{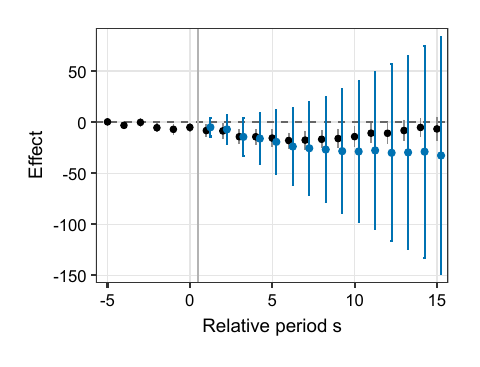}}
\caption{Community health centers: C-LF confidence sets by relative period at $M = 0.33$}
\label{fig:app-chc-byperiod}
\vspace{0.5em}
\parbox{\linewidth}{\footnotesize
    \textit{Notes:} 95 percent C-LF confidence sets for the effect on the age-adjusted mortality rate at each relative period $s = 1, \dots, 15$ of the balanced window, in deaths per 100,000, under the post-only SD restriction with $M = 0.33$ deaths per 100,000 per period. Colored whiskers are the confidence sets, with a circle at the center of each. Black points are the aggregate event-study coefficients over the balanced window, with 95 percent pointwise intervals; the vertical line separates pre-treatment from post-treatment periods. Confidence sets are pointwise.
}
\end{figure}

\clearpage
\section{Sensitivity analysis with baseline covariates}
\label{app:covariates}
\setcounter{equation}{0}
\setcounter{figure}{0}
\setcounter{table}{0}

Many DiD applications compare treated and control units after adjusting for baseline covariates. This appendix extends cohort inference to the estimators of \citet{callaway2021-did} with time-invariant covariates and applies it to the minimum-wage data under the SD restriction.

With covariates, the identifying assumption changes from unconditional PT to PT conditioning on baseline, time-invariant covariates. I therefore define the restriction over the violations of PT conditional on these covariates. The estimators of \citet{callaway2021-did} report cohort-period coefficients averaged over the covariate distribution of the treated cohort. For \CSNT{}, the restriction on the conditional violations carries over directly to the averaged coefficients. For \CSNYT{}, differences in covariate distributions across cohorts complicate the bias decomposition of the main text. A target cohort's estimation bias also depends on adjustment cohorts' block bias averaged under the target cohort's covariate distribution, which I call \emph{transported block bias}.

To address this complication, I work with the restrictions that the conditional SD restriction implies for the average and the transported block bias. They are necessary for the conditional restriction but not sufficient, so the confidence sets are valid under it and can be conservative. Under these implied restrictions, the pre-treatment portions of the transported block bias, which are estimable, bound the post-treatment portions. The post-treatment portions enter inference as additional nuisance parameters.

I consider the outcome-regression (OR), inverse-probability-weighting (IPW), and doubly robust (DR) estimators of \citet{callaway2021-did}, with never-treated (\CSNT{}) or not-yet-treated (\CSNYT{}) controls. I leave two extensions to future work. The RM restriction requires a maximum of absolute pre-treatment changes, which does not commute with averaging over covariates. The imputation estimator adjusts for covariates through a model of untreated outcomes, in which a covariate matters only if it or its coefficient varies over time. A covariate-adjusted version would therefore also require restrictions on how the covariate contribution to untreated outcomes evolves over time.

\subsection{Covariate-adjusted cohort-period coefficients}
\label{app:cov-setup}

Cohorts, adoption times, and relative periods are as in Section~\ref{sec:setup}: cohort $\mathcal G_g$ adopts in period $t_g$, $\mathcal G_\infty$ is the never-treated cohort, and $t_g - 1$ is the reference period. Relative period $s$ corresponds to calendar period $t = t_g + s - 1$, so that $s = 1$ is the first treated period and $s = 0$ the last untreated period. The covariate vector $X_i$ is measured before treatment and does not vary over time; $F_g$ denotes its distribution within cohort $\mathcal G_g$.

For a cohort $\mathcal G_a$ and for a set of units $A$ formed as a union of cohorts, define the untreated conditional means
\begin{align*}
\mu^0_{a,t}(x) &:= \E[Y_{it}(\infty) \mid i \in \mathcal G_a, X_i = x],\\
\mu^0_{A,t}(x) &:= \E[Y_{it}(\infty) \mid i \in A, X_i = x]
= \sum_{a:\, \mathcal G_a \subseteq A} \mathbb{P}(i \in \mathcal G_a \mid X_i = x,\, i \in A)\, \mu^0_{a,t}(x).
\end{align*}
The conditional treatment effect and the cohort-period effect from Section~\ref{sec:setup} are
\begin{equation*}
\tau_{\mathcal G_g,s}(x) := \E[Y_{i,t_g+s-1}(g) - Y_{i,t_g+s-1}(\infty) \mid i \in \mathcal G_g, X_i = x],
\qquad
\bar{\tau}_{\mathcal G_g,s} := \E[\tau_{\mathcal G_g,s}(X_i) \mid i \in \mathcal G_g].
\end{equation*}

The control groups are as in Section~\ref{sec:estimators}. The initial control group of cohort $\mathcal G_g$ and its not-yet-treated control group at post-treatment period $s$ are
\begin{equation*}
\mathcal C_{g,1} = \Bigl(\bigcup_{k: t_k > t_g} \mathcal G_k\Bigr) \cup \mathcal G_\infty,
\qquad
\mathcal C_{g,s} = \Bigl(\bigcup_{k: t_k > t} \mathcal G_k\Bigr) \cup \mathcal G_\infty,
\qquad t = t_g + s - 1,
\end{equation*}
so that $\mathcal C_{g,s}$ shrinks as later cohorts adopt. The adjustment cohorts at calendar time $t$ are $\mathcal K_g(t) = \{k : t_g < t_k \leq t\}$, and $s_k(t) = t - t_k + 1$ is the relative period of cohort $\mathcal G_k$ at time $t$.

For $s \geq 1$, the conditional cohort-period coefficients of the two estimators are
\begin{equation*}
\begin{aligned}
\beta^{\CSNT}_{\mathcal G_g,s}(x)
&:= \{\mu^g_{g,t}(x) - \mu^0_{g,t_g-1}(x)\} - \{\mu^0_{\infty,t}(x) - \mu^0_{\infty,t_g-1}(x)\},\\
\beta^{\CSNYT}_{\mathcal G_g,s}(x)
&:= \{\mu^g_{g,t}(x) - \mu^0_{g,t_g-1}(x)\} - \{\mu^0_{\mathcal C_{g,s},t}(x) - \mu^0_{\mathcal C_{g,s},t_g-1}(x)\},
\qquad t = t_g + s - 1,
\end{aligned}
\end{equation*}
where $\mu^g_{g,t}(x) := \E[Y_{it}(g) \mid i \in \mathcal G_g, X_i = x]$ is the treated conditional mean. Their estimation biases are $\delta_{\mathcal G_g,s}(x) := \beta_{\mathcal G_g,s}(x) - \tau_{\mathcal G_g,s}(x)$, with the estimator superscript suppressed when no confusion arises. For any conditional object I write a bar for its average over the cohort's covariate distribution, $\bar\beta_{\mathcal G_g,s} := \E[\beta_{\mathcal G_g,s}(X_i) \mid i \in \mathcal G_g]$ and likewise $\bar\delta_{\mathcal G_g,s}$ and, for the block biases defined below, $\bar\Delta_{\mathcal G_g,s}$. Then $\bar\beta_{\mathcal G_g,s} = \bar{\tau}_{\mathcal G_g,s} + \bar\delta_{\mathcal G_g,s}$.

Under overlap and the nuisance-model conditions of each estimator, the OR, IPW, and DR estimators are consistent for $\bar\beta_{\mathcal G_g,s}$ whether or not conditional PT holds \citep{callaway2021-did}. Conditional PT is the special case in which $\delta_{\mathcal G_g,s}(x) = 0$ for every $x$.

Conditional on $x$, the block biases are the objects of Section~\ref{sec:blockbias}:
\begin{equation*}
\begin{aligned}
\Delta^{\CSNT}_{\mathcal G_g,s}(x)
&:= \{\mu^0_{g,t}(x) - \mu^0_{g,t_g-1}(x)\} - \{\mu^0_{\infty,t}(x) - \mu^0_{\infty,t_g-1}(x)\},\\
\Delta^{\CSNYT}_{\mathcal G_g,s}(x)
&:= \{\mu^0_{g,t}(x) - \mu^0_{g,t_g-1}(x)\} - \{\mu^0_{\mathcal C_{g,1},t}(x) - \mu^0_{\mathcal C_{g,1},t_g-1}(x)\},
\qquad t = t_g + s - 1.
\end{aligned}
\end{equation*}
At every covariate value, each compares the untreated trend of cohort $\mathcal G_g$ with that of a fixed control group, the never-treated cohort for \CSNT{} and the initial control group for \CSNYT{}. Both are zero at $s = 0$.

\subsection{The bias decomposition conditional on covariates}
\label{app:cov-decomposition}

For a later cohort $\mathcal G_k$, let
\begin{equation}
q_k(x) := \mathbb{P}\bigl(i \in \mathcal G_k \,\big|\, X_i = x,\ i \in \mathcal G_k \cup \mathcal C_{k,1}\bigr)
\label{eq:cov-q}
\end{equation}
be the conditional share of cohort $\mathcal G_k$ among the units untreated immediately before it adopts, a set that includes $\mathcal G_k$ itself. Its unconditional counterpart is the adjustment weight $w_k = \frac{N_k}{\sum_{j=k}^{G} N_j + N_\infty}$ of Propositions~\ref{prop:imputation} and~\ref{prop:csnyt}.

\begin{proposition}[Conditional bias decomposition]
\label{prop:cov-pointwise}
Fix a covariate value $x$ with $\mathbb{P}(i \in \mathcal G_g \mid X_i = x) > 0$. Assume that $\mathbb{P}(i \in \mathcal G_\infty \mid X_i = x) > 0$ for \CSNT{}, that $\mathbb{P}(i \in \mathcal C_{g,s} \mid X_i = x) > 0$ for a \CSNYT{} post-treatment cell, and that $\mathbb{P}(i \in \mathcal G_k \mid X_i = x) > 0$ and $\mathbb{P}(i \in \mathcal C_{k,1} \mid X_i = x) > 0$ for every $k \in \mathcal K_g(t)$, $t = t_g + s - 1$. Then, for every $s \geq 1$,
\begin{equation*}
\delta^{\CSNT}_{\mathcal G_g,s}(x) = \Delta^{\CSNT}_{\mathcal G_g,s}(x)
\end{equation*}
and
\begin{equation}
\delta^{\CSNYT}_{\mathcal G_g,s}(x)
= \Delta^{\CSNYT}_{\mathcal G_g,s}(x)
+ \sum_{k \in \mathcal K_g(t)} q_k(x)
\bigl[\Delta^{\CSNYT}_{\mathcal G_k,s_k(t)}(x) - \Delta^{\CSNYT}_{\mathcal G_k,s_k(t_g-1)}(x)\bigr].
\label{eq:cov-pointwise}
\end{equation}
\end{proposition}

\begin{proof}
For \CSNT{}, the treated conditional mean cancels in $\beta^{\CSNT}_{\mathcal G_g,s}(x) - \tau_{\mathcal G_g,s}(x)$, leaving $\Delta^{\CSNT}_{\mathcal G_g,s}(x)$. For \CSNYT{}, the same cancellation gives $\delta_{\mathcal G_g,s}(x) - \Delta_{\mathcal G_g,s}(x) = D_t(\mathcal C_{g,1}) - D_t(\mathcal C_{g,s})$, where $D_t(A) := \mu^0_{A,t}(x) - \mu^0_{A,t_g-1}(x)$ is the untreated trend of the unit set $A$ from $t_g - 1$ to $t$. Order the adjustment cohorts in $\mathcal K_g(t)$ as $k_1 < \dots < k_m$ by adoption time. Because adoption times are strictly ordered, $\mathcal C_{g,1} = \mathcal G_{k_1} \cup \mathcal C_{k_1,1}$, $\mathcal C_{k_j,1} = \mathcal G_{k_{j+1}} \cup \mathcal C_{k_{j+1},1}$, and $\mathcal C_{k_m,1} = \mathcal C_{g,s}$. The mixture representation of a group mean gives, for any date $u$,
\begin{equation*}
\mu^0_{\mathcal G_k \cup \mathcal C_{k,1},u}(x) - \mu^0_{\mathcal C_{k,1},u}(x)
= q_k(x)\bigl[\mu^0_{k,u}(x) - \mu^0_{\mathcal C_{k,1},u}(x)\bigr].
\end{equation*}
Apply this at $u = t$ and $u = t_g - 1$ to each link of the chain and add. The trend difference between $\mathcal C_{g,1}$ and $\mathcal C_{g,s}$ then telescopes into $\sum_k q_k(x)\{[\mu^0_{k,t} - \mu^0_{\mathcal C_{k,1},t}] - [\mu^0_{k,t_g-1} - \mu^0_{\mathcal C_{k,1},t_g-1}]\}(x)$. Each bracket is cohort $\mathcal G_k$'s untreated gap relative to its own initial control group at one date. Adding and subtracting the same gap at $t_k - 1$ turns the difference into $\Delta_{\mathcal G_k,s_k(t)}(x) - \Delta_{\mathcal G_k,s_k(t_g-1)}(x)$.
\end{proof}

Proposition~\ref{prop:cov-pointwise} has the same structure as Proposition~\ref{prop:csnyt}: the target cohort's block bias enters directly, and each adjustment cohort contributes its current block bias minus its block bias at the target cohort's reference date. Covariate conditioning replaces $w_k$ with $q_k(x)$. 

\subsection{Averaging over covariates: transported block biases}
\label{app:cov-averaging}

Proposition~\ref{prop:cov-pointwise} shows the bias decomposition for cohort $\mathcal G_g$ conditional on each covariate value. However, the reported coefficients average over cohort $\mathcal G_g$'s covariate distribution. Averaging \eqref{eq:cov-pointwise} under $F_g$ gives, with $t = t_g + s - 1$,
\begin{equation*}
\bar\delta_{\mathcal G_g,s}
= \bar\Delta_{\mathcal G_g,s}
+ \sum_{k \in \mathcal K_g(t)}
\E\bigl[q_k(X_i)\{\Delta_{\mathcal G_k,s_k(t)}(X_i) - \Delta_{\mathcal G_k,s_k(t_g-1)}(X_i)\} \,\big|\, i \in \mathcal G_g\bigr].
\end{equation*}
The terms inside the sum are averages of cohort $\mathcal G_k$'s conditional block bias, but taken under cohort $\mathcal G_g$'s covariate distribution and weighted by $q_k(X_i)$. For each adjustment cohort $\mathcal G_k$, target cohort $\mathcal G_g$, and relative period $s$ of the adjustment cohort, I define the \emph{transported block bias} and the \emph{transported share} as
\begin{equation}
U_{k \mid g,s} := \E\bigl[q_k(X_i)\, \Delta_{\mathcal G_k,s}(X_i) \,\big|\, i \in \mathcal G_g\bigr],
\qquad
\omega_{k \mid g} := \E\bigl[q_k(X_i) \,\big|\, i \in \mathcal G_g\bigr].
\label{eq:cov-U}
\end{equation}
The notation $k \mid g$ reads ``cohort $\mathcal G_k$'s block bias, averaged for cohort $\mathcal G_g$.'' The estimation bias of a post-treatment \CSNYT{} coefficient is then
\begin{equation}
\;
\bar\delta_{\mathcal G_g,s}
= \bar\Delta_{\mathcal G_g,s}
+ \sum_{k \in \mathcal K_g(t)}
\bigl(U_{k \mid g,s_k(t)} - U_{k \mid g,s_k(t_g-1)}\bigr),
\qquad t = t_g + s - 1,\ s \geq 1.\;
\label{eq:cov-exact}
\end{equation}
For $s \leq 0$ the coefficient compares $\mathcal G_g$ with $\mathcal C_{g,1}$, as in Section~\ref{sec:coefvector}, and $\bar\tau_{\mathcal G_g,s} = 0$, so $\bar\delta_{\mathcal G_g,s} = \bar\Delta_{\mathcal G_g,s}$. For the adjustment cohort's relative period $s = 0$, $U_{k \mid g,0} = 0$ because $\Delta_{\mathcal G_k,0}(x) = 0$. The transported share $\omega_{k \mid g}$ is the average of the share $q_k$ over cohort $\mathcal G_g$'s covariate distribution, equal to $w_k$ when $q_k$ is constant. It does not enter \eqref{eq:cov-exact} but will be useful when we discuss the SD restriction.

Proposition~\ref{prop:csnyt} can be understood as a special case of Equation \eqref{eq:cov-exact}. If the target cohort $\mathcal G_g$ and every cohort in $\mathcal C_{g,1}$ have the same covariate distribution, $q_k(x)$ is the constant $w_k$ and $U_{k \mid g,s} = w_k \bar\Delta_{\mathcal G_k,s}$. However, the transported block bias in general cannot be expressed in terms of the average block bias $\bar\Delta_{\mathcal G_k,s}$. $U_{k \mid g,s}$ averages the same conditional block bias $\Delta_{\mathcal G_k,s}(x)$ as $\bar\Delta_{\mathcal G_k,s}$ but with a different measure, namely $q_k(x)\, dF_g(x)$ rather than $dF_k(x)$.

For $s \leq 0$ the conditional block bias $\Delta_{\mathcal G_k,s}(x)$ compares untreated outcomes at two untreated periods, so $U_{k \mid g,s}$ can be estimated from the data. The transported share $\omega_{k \mid g}$ can also be estimated from the joint distribution of cohort membership and covariates. For $s \geq 1$ the transported block biases are unobservable. The robust inference framework I will discuss later therefore relies on restrictions on both $\bar\Delta_{\mathcal G_k,s}$ and $U_{k \mid g,s}$.

\subsection{The SD restriction with covariates}
\label{app:cov-sd}

I derive these restrictions from a natural and interpretable primitive, the SD restriction on each conditional block bias for post-treatment periods $s \geq 1$. With $\nabla^2\Delta_{\mathcal G_k,s} := \Delta_{\mathcal G_k,s} - 2\Delta_{\mathcal G_k,s-1} + \Delta_{\mathcal G_k,s-2}$ as in Section~\ref{sec:restrictions}, this restriction requires
\begin{equation}
\bigl|\nabla^2\Delta_{\mathcal G_k,s}(x)\bigr| \leq M
\qquad \text{for every supported } x.
\label{eq:cov-pointwise-sd}
\end{equation}

The bound $M$ has the same meaning and units as in the main text: the slope of a cohort's untreated block-bias path may change by at most $M$ between adjacent periods, now at every covariate value. With these conditional restrictions, we can derive the restrictions on the average and transported block biases. 

\begin{proposition}[Inherited SD restrictions]
\label{prop:cov-inherited}
If the conditional block-bias paths are integrable and satisfy \eqref{eq:cov-pointwise-sd}, then
\begin{equation}
\bigl|\nabla^2\bar\Delta_{\mathcal G_k,s}\bigr| \leq M,
\qquad
\bigl|\nabla^2 U_{k \mid g,s}\bigr| \leq M\, \omega_{k \mid g}.
\label{eq:cov-inherited}
\end{equation}
\end{proposition}

\begin{proof}
Second differencing and integration commute, so $\nabla^2\bar\Delta_{\mathcal G_k,s} = \int \nabla^2\Delta_{\mathcal G_k,s}(x)\, dF_k(x)$ and $\nabla^2 U_{k \mid g,s} = \int q_k(x)\, \nabla^2\Delta_{\mathcal G_k,s}(x)\, dF_g(x)$. The triangle inequality, \eqref{eq:cov-pointwise-sd}, and $q_k(x) \geq 0$ give the two bounds.
\end{proof}

$U_{k \mid g,s}$ weights the conditional block bias by $q_k$. With the conditional restriction, its curvature bound is $M\omega_{k \mid g}$, that is, a smaller transported share permits less curvature in that transported block bias.  For fixed $M$, this restriction is linear in $(U_{k \mid g,\cdot}, \omega_{k \mid g})$, so the restricted set remains a polyhedron when the transported share is treated as unknown. This lets me include the transported share as an estimated nuisance parameter and account for its sampling uncertainty.

The implied restriction \eqref{eq:cov-inherited} is an outer relaxation of the conditional restriction \eqref{eq:cov-pointwise-sd}: a pair $(\bar\Delta_{\mathcal G_k,\cdot}, U_{k \mid g,\cdot})$ can satisfy the two bounds in \eqref{eq:cov-inherited} without coming from a common conditional block bias that satisfies the conditional restriction. The conditional restriction also implies other bounds, for example $|\nabla^2 U_{k \mid g,s} - \omega_{k \mid g}\nabla^2\bar\Delta_{\mathcal G_k,s}| \leq M\, \mathrm{TV}_{k \mid g}$. Here $\mathrm{TV}_{k \mid g} := \int |q_k(x)\, dF_g(x) - \omega_{k \mid g}\, dF_k(x)|$ is the distance between the two averaging measures after $dF_k$ is scaled to the mass $\omega_{k \mid g}$ of $q_k\, dF_g$. I impose only the two bounds in \eqref{eq:cov-inherited} for simplicity.

The confidence sets below are built using the implied restriction and are valid whenever \eqref{eq:cov-inherited} holds. Because the conditional restriction \eqref{eq:cov-pointwise-sd} implies \eqref{eq:cov-inherited}, the confidence sets are also valid under the conditional restriction, though they may be conservative. Section~\ref{app:cov-inference} states the coverage guarantee behind this claim.

\subsection{Inference}
\label{app:cov-inference}

Stack the average block biases, the transported block biases, and the transported shares into an augmented block-bias vector $\vec\Delta^{+} := (\vec{\bar\Delta}', \vec U', \vec\omega')'$. Stack the cohort-period coefficients, the pre-treatment transported block biases, and the transported shares into an augmented coefficient vector $\vec\beta^{+} := (\vec{\bar\beta}', \vec U_{\pre}', \vec\omega')'$. Here $\vec{\bar\beta}$ collects $\bar\beta_{\mathcal G_g,s}$ over the cohorts and the non-reference periods, $\vec U_{\pre}$ collects $U_{k \mid g,s}$ over the pairs $k \mid g$ and the pre-treatment periods $s \leq 0$, and $\vec\omega$ collects $\omega_{k \mid g}$. The population vector satisfies
\begin{equation}
\vec\beta^{+} = \vec\tau^{+} + \mathbf B \vec\Delta^{+},
\label{eq:cov-B}
\end{equation}
where $\vec\tau^{+}$ holds the post-treatment treatment effects in the post-treatment coefficient rows and zeros elsewhere, and $\mathbf B$ contains only zeros, ones, and minus ones. Equation \eqref{eq:cov-B} is the analogue of $\vec\beta = \vec\tau + \mathbf W\vec\Delta$ in Section~\ref{sec:model}. As in Section~\ref{sec:inference}, we have an observed vector whose mean is linear in $\vec\tau^{+}$ and in a nuisance vector that lies in a polyhedral restriction set. The pre-treatment transported block biases are part of the observed vector as we need to restrict the potential transported block bias in post-treatment periods using its observed history.

The post-only SD class in Section~\ref{sec:restrictions} extends to the augmented vector as
\begin{equation}
\Lambda^{\mathrm{SD}}_{\Delta^{+}}(M) := \bigl\{\vec\Delta^{+} :
\ |\nabla^2\bar\Delta_{\mathcal G_k,s}| \leq M,\ \ |\nabla^2 U_{k \mid g,s}| \leq M\omega_{k \mid g}
\ \text{for all } s \geq 1,\ \ 0 \leq \omega_{k \mid g} \leq 1\bigr\}.
\label{eq:cov-class-new}
\end{equation}
It collects the two bounds in \eqref{eq:cov-inherited} for every cohort, every pair $k \mid g$, and every post-treatment period, together with the range of the transported shares.

In the Gaussian model of \eqref{eq:model} with the augmented vector, $\vec{\hat\beta}^{+} \sim N(\vec\beta^{+}, \Sigma_N)$ with $\Sigma_N$ known, the inference procedures of \citet{andrews2023inference} apply. Their guarantee is the analogue of \eqref{eq:coverage}, with the identified set $\mathcal S(\vec\beta^{+}, \Lambda^{\mathrm{SD}}_{\Delta^{+}}(M))$ defined as in \eqref{eq:identified-set} with \eqref{eq:cov-B} in place of \eqref{eq:model}:
\begin{equation}
\inf_{\vec\Delta^{+} \in \Lambda^{\mathrm{SD}}_{\Delta^{+}}(M),\ \vec\tau^{+}}\ \ \inf_{\theta \in \mathcal S(\vec\tau^{+} + \mathbf B\vec\Delta^{+},\ \Lambda^{\mathrm{SD}}_{\Delta^{+}}(M))}
\mathbb{P}\bigl(\theta \in \mathcal C_N(\vec{\hat\beta}^{+}, \Sigma_N)\bigr) \geq 1 - \alpha.
\label{eq:cov-coverage}
\end{equation}
The outer infimum runs over every DGP whose average and transported block biases satisfy the implied restriction, and the inner one over every value of $\theta$ that such a DGP leaves compatible with the coefficients. The confidence sets are thus valid under the implied restriction in the Gaussian model. With $\hat\Sigma_N$ in place of $\Sigma_N$, they are also valid uniformly in large samples under the uniform normal approximation, covariance consistency, and regularity conditions of \RR{} (their Section~3.3), as in Section~\ref{sec:confidence}. Under the conditional restriction, the admissible DGPs are those whose $\vec\Delta^{+}$ is generated by conditional block biases that satisfy it. These $\vec\Delta^{+}$ form a subset of $\Lambda^{\mathrm{SD}}_{\Delta^{+}}(M)$ by Proposition~\ref{prop:cov-inherited}, and the compatible values of $\theta$ therefore form a subset of $\mathcal S(\vec\beta^{+}, \Lambda^{\mathrm{SD}}_{\Delta^{+}}(M))$. Both infima in \eqref{eq:cov-coverage} are then taken over smaller sets, so the bound $1 - \alpha$ still holds and the confidence sets remain valid under the conditional restriction. They can nevertheless be conservative under it, because they are built to cover every value in $\mathcal S(\vec\beta^{+}, \Lambda^{\mathrm{SD}}_{\Delta^{+}}(M))$, which can be larger than the identified set under the conditional restriction.

The one difference from the main text is that the rectangular map $\mathbf B$ cannot be inverted as $\mathbf W$ was. However, for any given value of $\vec b = \mathbf B\vec\Delta^{+}$, we can write all compatible augmented block-bias vectors as $\vec\Delta^{+} = \mathbf R\vec b + \mathbf N\vec\eta$. Here $\mathbf R$ is a right inverse of $\mathbf B$, $\mathbf N$ spans its null space, and $\vec\eta$ is unrestricted. A right inverse exists because $\mathbf B$ has full row rank. Substituting into \eqref{eq:cov-class-new} gives linear inequalities in $(\vec b, \vec\eta)$. In the Conditional and C-LF tests of \citet{andrews2023inference}, we can profile over $\vec\eta$ together with the nuisance treatment effects. The nuisance treatment effects are the post-treatment effects in $\vec\tau^{+}$, which the null restricts only through $\theta$. In the application below, the test inversion uses a $0.002$ log-point grid, 5{,}000 least-favorable critical-value draws, and first-stage size $\kappa = \alpha/10$.

Alongside each confidence set I report the plug-in identified set $\mathcal S(\vec{\hat\beta}^{+}, \Lambda^{\mathrm{SD}}_{\Delta^{+}}(M))$, obtained by fixing the pre-treatment coordinates and the transported shares at their estimates. The plug-in set isolates identification uncertainty, and the confidence set adds statistical uncertainty.

\subsection{Estimating the transported block biases and transported shares}
\label{app:cov-estimation}

The pre-treatment transported block biases and the transported shares are estimated from untreated data only. Fix a target cohort $\mathcal G_g$, an adjustment cohort $\mathcal G_k$, and a pre-treatment relative period $s \leq 0$ of the adjustment cohort. Define the observed outcome change as $V_{i,k,s} := Y_{i,t_k+s-1} - Y_{i,t_k-1}$. Cohort $\mathcal G_k$ and every unit in $\mathcal C_{k,1}$ are untreated at both periods, so with
\begin{equation*}
m_{k,s}(x) := \E[V_{i,k,s} \mid i \in \mathcal G_k, X_i = x],
\qquad
m_{C_k,s}(x) := \E[V_{i,k,s} \mid i \in \mathcal C_{k,1}, X_i = x],
\end{equation*}
the conditional block bias is the difference $\Delta_{\mathcal G_k,s}(x) = m_{k,s}(x) - m_{C_k,s}(x)$. The two quantities to estimate are the averages over the target cohort in \eqref{eq:cov-U}, $U_{k \mid g,s} = \E[q_k(X_i)\, \Delta_{\mathcal G_k,s}(X_i) \mid i \in \mathcal G_g]$ and $\omega_{k \mid g} = \E[q_k(X_i) \mid i \in \mathcal G_g]$.

The functions $q_k$ and $\Delta_{\mathcal G_k,s}$ are identified in the adjustment sample $\mathcal R_k := \mathcal G_k \cup \mathcal C_{k,1}$, whose covariate distribution differs from that of the target cohort. Under the overlap conditions of Proposition~\ref{prop:cov-pointwise}, two facts connect the two samples. Write $I_{k,i} := \mathbf 1\{i \in \mathcal G_k\}$, $I_{C_k,i} := \mathbf 1\{i \in \mathcal C_{k,1}\}$, and $o_k(x) := q_k(x)/\{1 - q_k(x)\}$. First, within $\mathcal R_k$, the product $q_k \Delta_{\mathcal G_k,s}$ is the conditional mean of a signed outcome change:
\begin{equation*}
\E\bigl[\{I_{k,i} - o_k(X_i)\, I_{C_k,i}\}\, V_{i,k,s} \,\big|\, X_i = x,\ i \in \mathcal R_k\bigr] = q_k(x)\, \Delta_{\mathcal G_k,s}(x).
\end{equation*}
Adjustment-cohort units enter with weight one. Control units enter with weight $-o_k$, which scales the control group to the size of the adjustment cohort at each covariate value. Second, weighting the units of $\mathcal R_k$ by the transport ratio
\begin{equation*}
\rho_{k \mid g}(x) := \frac{\mathbb{P}(i \in \mathcal G_g \mid X_i = x)}{\mathbb{P}(i \in \mathcal R_k \mid X_i = x)}
\end{equation*}
gives them the covariate distribution of the target cohort: for any function $f$, the $\rho_{k \mid g}$-weighted mean of $f(X_i)$ over $\mathcal R_k$ equals $\E[f(X_i) \mid i \in \mathcal G_g]$.

The estimators replace these two population averages by sample averages: the mean over the $N_g$ units of the target cohort, and the mean over the units of the adjustment sample with weights $\hat\rho_i := \hat\rho_{k \mid g}(X_i)$. Hats on functions denote fitted versions, with $\hat\Delta_{\mathcal G_k,s} := \hat m_{k,s} - \hat m_{C_k,s}$ and $\hat o_k := \hat q_k/(1 - \hat q_k)$. The three estimators mirror the three CS estimators. The outcome regression (OR) versions average fitted values over the target cohort:
\begin{equation}
\widehat U^{\mathrm{OR}}_{k \mid g,s} = \frac{1}{N_g}\sum_{i \in \mathcal G_g} \hat q_k(X_i)\, \hat\Delta_{\mathcal G_k,s}(X_i),
\qquad
\widehat\omega^{\mathrm{OR}}_{k \mid g} = \frac{1}{N_g}\sum_{i \in \mathcal G_g} \hat q_k(X_i).
\label{eq:cov-or}
\end{equation}
The inverse-probability-weighting (IPW) versions use no outcome model. They average the signed outcome changes of the adjustment sample, reweighted to the target cohort:
\begin{equation}
\widehat U^{\mathrm{IPW}}_{k \mid g,s} = \frac{\sum_{i \in \mathcal R_k} \hat\rho_i\, \{I_{k,i} - \hat o_k(X_i)\, I_{C_k,i}\}\, V_{i,k,s}}{\sum_{i \in \mathcal R_k} \hat\rho_i},
\qquad
\widehat\omega^{\mathrm{IPW}}_{k \mid g} = \frac{\sum_{i \in \mathcal R_k} \hat\rho_i\, I_{k,i}}{\sum_{i \in \mathcal R_k} \hat\rho_i}.
\label{eq:cov-ipw}
\end{equation}
The augmented (AUG) versions add to the OR values the reweighted residuals of the adjustment sample:
\begin{equation}
\widehat U^{\mathrm{AUG}}_{k \mid g,s} = \widehat U^{\mathrm{OR}}_{k \mid g,s} + \frac{\sum_{i \in \mathcal R_k} \hat\rho_i\, \hat\psi_{i,k,s}}{\sum_{i \in \mathcal R_k} \hat\rho_i},
\qquad
\widehat\omega^{\mathrm{AUG}}_{k \mid g} = \widehat\omega^{\mathrm{OR}}_{k \mid g} + \frac{\sum_{i \in \mathcal R_k} \hat\rho_i\, \{I_{k,i} - \hat q_k(X_i)\}}{\sum_{i \in \mathcal R_k} \hat\rho_i},
\label{eq:cov-aug}
\end{equation}
with the residual
\begin{equation*}
\hat\psi_{i,k,s} := I_{k,i}\bigl\{V_{i,k,s} - \hat m_{k,s}(X_i)\bigr\} - \hat o_k(X_i)\, I_{C_k,i}\bigl\{V_{i,k,s} - \hat m_{C_k,s}(X_i)\bigr\} + \bigl\{I_{k,i} - \hat q_k(X_i)\bigr\}\, \hat\Delta_{\mathcal G_k,s}(X_i).
\end{equation*}
Conditional on the covariates, the first term has mean zero when the outcome regression of the adjustment cohort is correct, the second when the outcome regression of the control group is correct, and the third when the share $q_k$ is correct.

In the application, I fit the share $q_k$ by a logit model on the adjustment sample $\mathcal R_k$, the transport ratio $\rho_{k \mid g}$ by the odds from a logit model of membership in $\mathcal G_g$, fitted on $\mathcal G_g \cup \mathcal R_k$, and the outcome regressions $m_{k,s}$ and $m_{C_k,s}$ by linear regressions. All of these models use the same covariates as the cohort-period coefficients. Proposition~\ref{prop:cov-robust} gives combinations of correct models under which the three estimators are consistent.

\begin{proposition}[Consistency of the OR, IPW, and AUG estimators]
\label{prop:cov-robust}
Let $q^*_k$, $m^*_{k,s}$, $m^*_{C_k,s}$, and $\rho^*_{k \mid g}$ denote the probability limits of the fitted functions, with $q^*_k < 1$ and $\E[\rho^*_{k \mid g}(X_i) \mid i \in \mathcal R_k] > 0$. Call a model correct when its limit equals the true function. Suppose that, for every $x$ in the support of $F_g$, the overlap condition $\mathbb{P}(i \in \mathcal G_k \mid X_i = x) > 0$ and $\mathbb{P}(i \in \mathcal C_{k,1} \mid X_i = x) > 0$ holds. Suppose also that the sample averages in \eqref{eq:cov-or} to \eqref{eq:cov-aug} converge in probability to their population counterparts, with the limits in place of the fitted functions. Consider three conditions: (i) both outcome regressions $m_{k,s}$ and $m_{C_k,s}$ are correct; (ii) the share $q_k$ is correct; (iii) the transport ratio $\rho_{k \mid g}$ is correct.
\begin{enumerate}[nosep,label=(\alph*)]
\item $\widehat U^{\mathrm{OR}}_{k \mid g,s}$ is consistent for the pre-treatment transported block bias $U_{k \mid g,s}$ under (i) and (ii), and $\widehat\omega^{\mathrm{OR}}_{k \mid g}$ is consistent for the transported share $\omega_{k \mid g}$ under (ii).
\item $\widehat U^{\mathrm{IPW}}_{k \mid g,s}$ is consistent under (ii) and (iii), and $\widehat\omega^{\mathrm{IPW}}_{k \mid g}$ is consistent under (iii).
\item $\widehat U^{\mathrm{AUG}}_{k \mid g,s}$ is consistent if at least two of (i), (ii), and (iii) hold. When (iii) holds, (i) can be weakened to a correct outcome regression of the control group, $m_{C_k,s}$, alone. $\widehat\omega^{\mathrm{AUG}}_{k \mid g}$ is consistent under (ii) or (iii).
\end{enumerate}
\end{proposition}

\begin{proof}
Let $\Delta^*_{\mathcal G_k,s} := m^*_{k,s} - m^*_{C_k,s}$ and $o^*_k := q^*_k/(1 - q^*_k)$ be the limits of $\hat\Delta_{\mathcal G_k,s}$ and $\hat o_k$. Under the assumed convergence, each estimator converges to its population counterpart with the limits in place of the fitted functions. Two facts are used throughout. Conditional on $X_i = x$ and $i \in \mathcal R_k$, the indicator $I_{k,i}$ has mean $q_k(x)$ by \eqref{eq:cov-q}, and $V_{i,k,s}$ has mean $m_{k,s}(x)$ in the adjustment cohort and $m_{C_k,s}(x)$ in the control group. Under (iii), the $\rho^*_{k \mid g}$-weighted mean over $\mathcal R_k$ of any function of $X_i$ equals its mean over the target cohort, by the overlap condition.

(a) The limit of $\widehat U^{\mathrm{OR}}_{k \mid g,s}$ is $\E[q^*_k(X_i)\, \Delta^*_{\mathcal G_k,s}(X_i) \mid i \in \mathcal G_g]$, which equals $U_{k \mid g,s}$ under (i) and (ii). The limit of $\widehat\omega^{\mathrm{OR}}_{k \mid g}$ is $\E[q^*_k(X_i) \mid i \in \mathcal G_g]$, which equals $\omega_{k \mid g}$ under (ii).

(b) Conditional on $X_i = x$ and $i \in \mathcal R_k$, the signed change $\{I_{k,i} - o^*_k(X_i)\, I_{C_k,i}\}\, V_{i,k,s}$ has mean $q_k m_{k,s} - o^*_k (1 - q_k)\, m_{C_k,s}$, evaluated at $x$. Under (ii), $o^*_k (1 - q_k) = q_k$, so this mean is $q_k \Delta_{\mathcal G_k,s}$. The limit of $\widehat U^{\mathrm{IPW}}_{k \mid g,s}$ is the $\rho^*_{k \mid g}$-weighted mean of this conditional mean over $\mathcal R_k$, which under (iii) equals $\E[q_k(X_i)\, \Delta_{\mathcal G_k,s}(X_i) \mid i \in \mathcal G_g] = U_{k \mid g,s}$. Likewise, $\widehat\omega^{\mathrm{IPW}}_{k \mid g}$ converges to the $\rho^*_{k \mid g}$-weighted mean of $q_k(X_i)$ over $\mathcal R_k$, which equals $\omega_{k \mid g}$ under (iii).

(c) Let $\psi^*_{i,k,s}$ be the residual with the limits in place of the fitted functions. The limit of $\widehat U^{\mathrm{AUG}}_{k \mid g,s}$ is
\begin{equation*}
U^*_{k \mid g,s} := \E\bigl[q^*_k(X_i)\, \Delta^*_{\mathcal G_k,s}(X_i) \,\big|\, i \in \mathcal G_g\bigr]
+ \frac{\E\bigl[\rho^*_{k \mid g}(X_i)\, \psi^*_{i,k,s} \,\big|\, i \in \mathcal R_k\bigr]}{\E\bigl[\rho^*_{k \mid g}(X_i) \,\big|\, i \in \mathcal R_k\bigr]}.
\end{equation*}
The conditional mean of $\psi^*_{i,k,s}$ given $X_i = x$ and $i \in \mathcal R_k$ is
\begin{equation*}
\zeta(x) := q_k\,(m_{k,s} - m^*_{k,s}) - o^*_k\,(1 - q_k)\,(m_{C_k,s} - m^*_{C_k,s}) + (q_k - q^*_k)\, \Delta^*_{\mathcal G_k,s},
\end{equation*}
with every function evaluated at $x$ and a term read as zero where its weight $q_k$ or $1 - q_k$ vanishes. By iterated expectations, $\zeta(X_i)$ can replace $\psi^*_{i,k,s}$ in $U^*_{k \mid g,s}$. Add and subtract $\E[\zeta(X_i) \mid i \in \mathcal G_g]$ in $U^*_{k \mid g,s}$, so that the target-cohort terms combine into the mean of $q^*_k \Delta^*_{\mathcal G_k,s} + \zeta$. In this sum, the terms in $q^*_k \Delta^*_{\mathcal G_k,s}$ cancel and, because $\Delta^*_{\mathcal G_k,s} = m^*_{k,s} - m^*_{C_k,s}$, so do the terms in $m^*_{k,s}$. This gives the first equality below. The second equality substitutes $m_{k,s} = \Delta_{\mathcal G_k,s} + m_{C_k,s}$ and uses $q_k - o^*_k(1 - q_k) = (q_k - q^*_k)/(1 - q^*_k)$:
\begin{equation*}
\begin{aligned}
q^*_k \Delta^*_{\mathcal G_k,s} + \zeta
&= q_k\,(m_{k,s} - m^*_{C_k,s}) - o^*_k\,(1 - q_k)\,(m_{C_k,s} - m^*_{C_k,s})\\
&= q_k \Delta_{\mathcal G_k,s} + \frac{(q_k - q^*_k)\,(m_{C_k,s} - m^*_{C_k,s})}{1 - q^*_k}.
\end{aligned}
\end{equation*}
By \eqref{eq:cov-U}, the limit therefore satisfies
\begin{equation}
\begin{aligned}
U^*_{k \mid g,s} - U_{k \mid g,s}
={}& \E\biggl[\frac{\{q_k(X_i) - q^*_k(X_i)\}\{m_{C_k,s}(X_i) - m^*_{C_k,s}(X_i)\}}{1 - q^*_k(X_i)} \,\bigg|\, i \in \mathcal G_g\biggr]\\
&+ \frac{\E\bigl[\rho^*_{k \mid g}(X_i)\, \zeta(X_i) \,\big|\, i \in \mathcal R_k\bigr]}{\E\bigl[\rho^*_{k \mid g}(X_i) \,\big|\, i \in \mathcal R_k\bigr]} - \E\bigl[\zeta(X_i) \,\big|\, i \in \mathcal G_g\bigr].
\end{aligned}
\label{eq:cov-aug-error}
\end{equation}
The first line of \eqref{eq:cov-aug-error} is zero under (ii) or when the outcome regression of the control group is correct, which (i) implies. The second line is zero when $\zeta = 0$, which holds under (i) and (ii) together. It is also zero under (iii), because the two averages of $\zeta$ then coincide. Any two of (i), (ii), and (iii) therefore make both lines zero, so that $U^*_{k \mid g,s} = U_{k \mid g,s}$. Under (iii), the first line uses only the outcome regression of the control group from (i). For the transported share, $I_{k,i} - q^*_k(X_i)$ has conditional mean $q_k(x) - q^*_k(x)$, so the limit of $\widehat\omega^{\mathrm{AUG}}_{k \mid g}$ is
\begin{equation*}
\E\bigl[q^*_k(X_i) \,\big|\, i \in \mathcal G_g\bigr]
+ \frac{\E\bigl[\rho^*_{k \mid g}(X_i)\,\{q_k(X_i) - q^*_k(X_i)\} \,\big|\, i \in \mathcal R_k\bigr]}{\E\bigl[\rho^*_{k \mid g}(X_i) \,\big|\, i \in \mathcal R_k\bigr]}.
\end{equation*}
Under (ii) the second term is zero and the first is $\omega_{k \mid g}$. Under (iii) the second term equals $\E[q_k(X_i) - q^*_k(X_i) \mid i \in \mathcal G_g]$, and the sum is again $\omega_{k \mid g}$.
\end{proof}

Part (c) of Proposition~\ref{prop:cov-robust} shows that the AUG estimator's double robustness is weaker than the double robustness of the DR cohort-period coefficient in \citet{callaway2021-did}, which requires only one of its two models to be correct. The share $q_k$ is part of the estimand $U_{k \mid g,s}$, so correct outcome regressions alone are not enough. With correct outcome regressions, the first line of \eqref{eq:cov-aug-error} is zero, but the second is zero, in general, only when the share or the transport ratio is also correct. The specification that combines the DR coefficients with the AUG estimators in \eqref{eq:cov-aug} is therefore labeled AUG rather than DR. The OR and IPW specifications combine the OR and IPW coefficients with the estimators in \eqref{eq:cov-or} and \eqref{eq:cov-ipw}, respectively.

$\vec\beta^{+}$ can be estimated by any of the three methods. For the cohort-period coefficients I take the influence functions from \texttt{did}. For the transported block biases and transported shares I derive them from a stacked M-estimation system that includes the nuisance-model estimating equations, so that first-step estimation error is accounted for. The cross-product of the stacked influence functions, divided by the squared county count, is $\hat\Sigma_N$. 

\subsection{Minimum-wage application}
\label{app:cov-application}

The data are the balanced county panel of Section~\ref{sec:illustration}: 2{,}284 counties observed from 2001 to 2007. Of these, cohort $\mathcal G_4$ (100 counties) first raises the minimum wage in 2004, $\mathcal G_6$ (223) in 2006, $\mathcal G_7$ (584) in 2007, and 1{,}377 never do. The rich specification of \citet{callaway2021-did} adjusts for region, population, median income, white share, high-school share, poverty share, and quadratic terms in population and income, all fixed within county. The target is $\theta = \ell'\vec{\bar\tau}_{\post}$, the cohort-size-weighted average of the averaged effects $\bar\tau_{\mathcal G_g,s}$ over the seven post-treatment cohort-period cells. Since each $\bar\tau_{\mathcal G_g,s}$ is the cell's unconditional effect, this is the target of Section~\ref{sec:minwage}. I report four specifications: \CSNYT{} without covariates, and \CSNYT{} with the rich covariates estimated by OR, IPW, and AUG. The no-covariate specification repeats the \CSNYT{} analysis of Section~\ref{sec:minwage} with the covariance construction of Section~\ref{app:cov-estimation} and a grid step of $0.002$. These two changes move its C-LF confidence set at $M = 0$ from $[-0.074, -0.017]$ in Figure~\ref{fig:minwage-sd} to $[-0.072, -0.018]$. 

Table~\ref{tab:cov-masses} reports the estimated transported shares. They exceed the count shares $w_k$ by between about 20 and 60 percent. The two are averages of the same conditional share under different covariate distributions. $w_k = \E[q_k(X_i) \mid i \in \mathcal R_k]$ averages $q_k$ over the adjustment sample $\mathcal R_k$, which consists mostly of never-treated counties, whereas $\omega_{k \mid g}$ averages it over the target cohort $\mathcal{G}_{g}$. 

\begin{table}[!h]
\centering
\caption{Transported shares in the minimum-wage design}
\label{tab:cov-masses}
\footnotesize
\begin{tabular}{lcccc}
\toprule
Adjustment $\mid$ target & Count share $w_k$ & OR & IPW & AUG \\
\midrule
$\mathcal G_6 \mid \mathcal G_4$ & 0.102 & 0.163 (0.018) & 0.168 (0.018) & 0.159 (0.018) \\
$\mathcal G_7 \mid \mathcal G_4$ & 0.298 & 0.391 (0.027) & 0.402 (0.026) & 0.391 (0.027) \\
$\mathcal G_7 \mid \mathcal G_6$ & 0.298 & 0.362 (0.018) & 0.349 (0.018) & 0.349 (0.018) \\
\bottomrule
\end{tabular}
\vspace{0.5em}
\parbox{\linewidth}{\footnotesize
    \textit{Notes:} The count share is $w_k = N_k/(\sum_{j=k}^{G} N_j + N_\infty)$, the adjustment weight of the no-covariate analysis. The remaining columns report $\widehat\omega_{k \mid g}$ from \eqref{eq:cov-or}, \eqref{eq:cov-ipw}, and \eqref{eq:cov-aug}, respectively, with standard errors in parentheses. Cohorts are named by adoption period as in the main text: $\mathcal G_4$, $\mathcal G_6$, and $\mathcal G_7$ adopt in 2004, 2006, and 2007.
}
\end{table}

Figure~\ref{fig:cov-average} shows the confidence sets for the average effect at $M = 0$, $0.005$, $0.010$, $0.015$, and $0.020$. Covariate adjustment moves the conventional estimate about a third of the way toward zero, from $-0.050$ without covariates to $-0.034$, $-0.033$, and $-0.032$ with OR, IPW, and AUG. The sensitivity sets move with it. At $M = 0$ the plug-in identified set is a point, the effect implied by extending each cohort's block-bias path linearly from its last two pre-treatment periods. That point is $-0.045$ without covariates and $-0.031$, $-0.020$, and $-0.022$ with covariates. The C-LF confidence set at $M = 0$ excludes zero without covariates ($[-0.072, -0.018]$) and with OR ($[-0.060, -0.002]$). It includes zero using IPW ($[-0.052, 0.012]$) and AUG ($[-0.056, 0.012]$). At $M = 0.005$ only the no-covariate set still excludes zero, and by $M = 0.010$ every set includes zero. After covariate adjustment, point estimates move toward zero, and every adjusted set includes zero when the allowed slope change reaches $0.005$ log points per period. 

\begin{figure}[!h]
\centering
\subfigure[No covariates]{\includegraphics{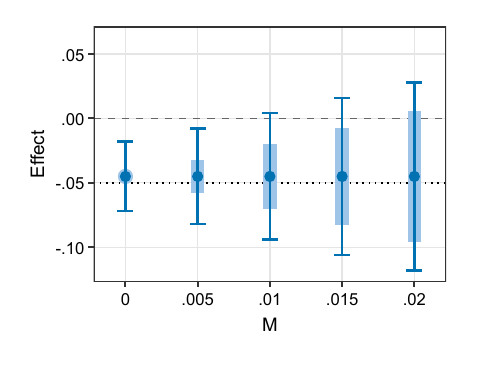}}\hfill
\subfigure[OR]{\includegraphics{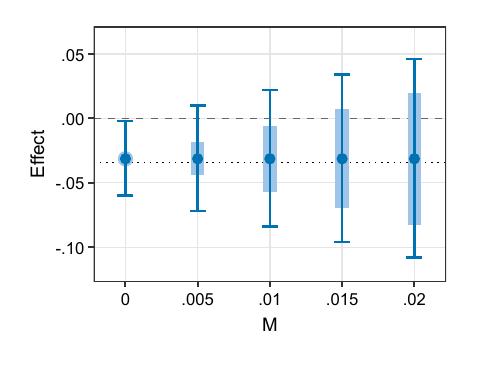}}\\[0.2em]
\subfigure[IPW]{\includegraphics{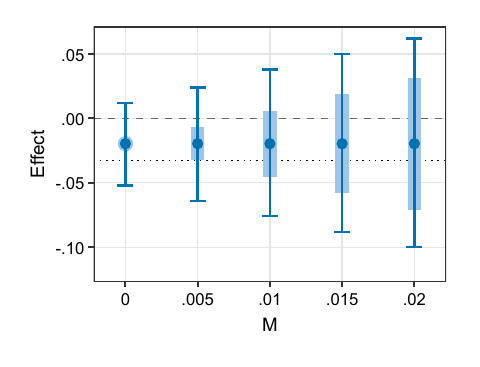}}\hfill
\subfigure[AUG]{\includegraphics{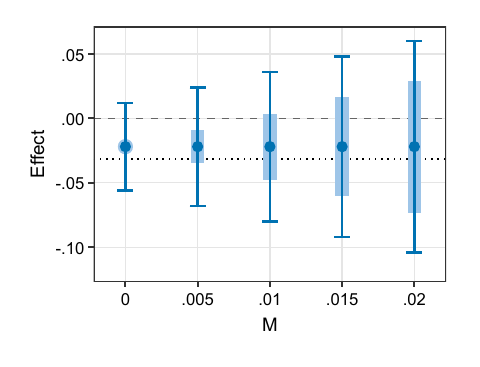}}\\[0.2em]
\includegraphics{figures/pdf/fig7_minwage_rm_legend.pdf}
\caption{Covariates: confidence sets for the average effect}
\label{fig:cov-average}
\vspace{0.5em}
\parbox{\linewidth}{\footnotesize
    \textit{Notes:} 95 percent C-LF confidence sets under cohort inference for the average effect over the seven post-treatment cohort-period cells, in log points, against the bound $M$ in log points per period, for \CSNYT{}, with the rich covariates in (b) to (d). The light band is the plug-in identified set under the restriction \eqref{eq:cov-class-new}, a dot at $M = 0$; the marker is its center; the dashed line is zero and the dotted line the panel's conventional estimate.
}
\end{figure}

Figure~\ref{fig:cov-event} reports cohort-size-weighted effects by relative period at $M = 0.010$. The covariate-adjusted estimated treatment effects become more negative with time since treatment, from $-0.021$ in the first post-treatment period to $-0.071$ in the fourth under AUG. The sensitivity sets widen quickly with the relative period. The first-period set is narrow because no transported block bias enters and one period of extrapolation is allowed. The second-period set excludes zero without covariates and with OR but not with IPW or AUG. The third- and fourth-period sets are wide because only $\mathcal G_4$ remains and its allowed bias path has more periods over which to diverge.

\begin{figure}[!h]
\centering
\subfigure[No covariates]{\includegraphics{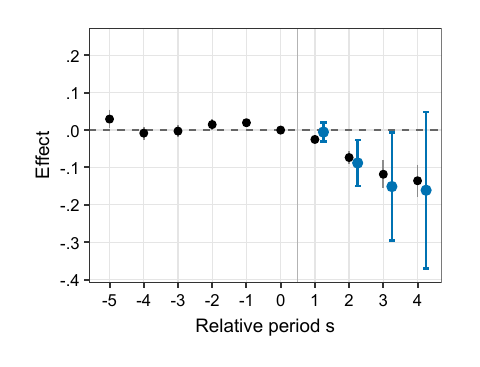}}\hfill
\subfigure[OR]{\includegraphics{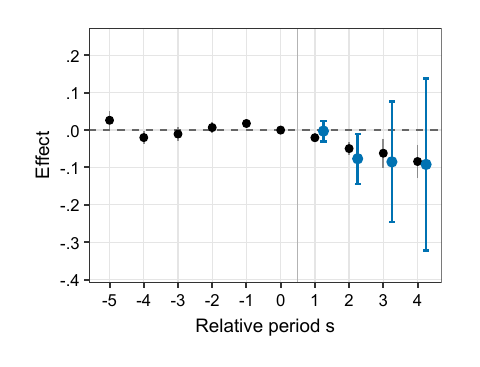}}\\[0.2em]
\subfigure[IPW]{\includegraphics{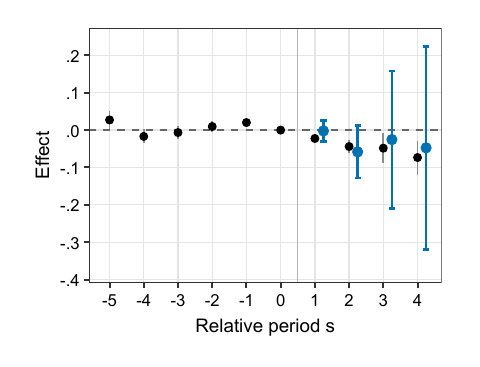}}\hfill
\subfigure[AUG]{\includegraphics{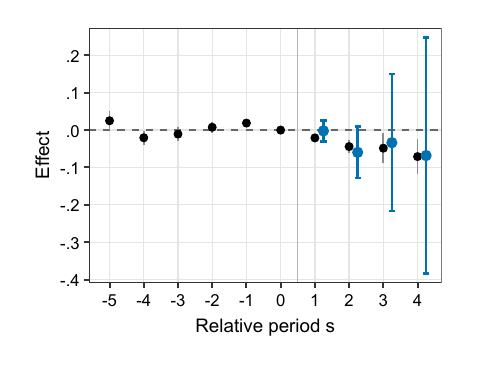}}\\[0.2em]
\includegraphics{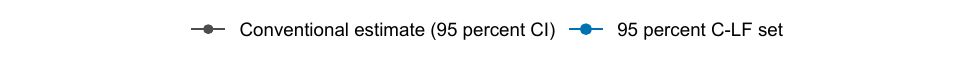}
\caption{Covariates: confidence sets by relative period at $M = 0.01$}
\label{fig:cov-event}
\vspace{0.5em}
\parbox{\linewidth}{\footnotesize
    \textit{Notes:} Black points and intervals are the conventional cohort-size-weighted estimates by relative period with 95 percent pointwise intervals; the pre-treatment points are cohort-size-weighted averages of the covariate-adjusted block biases, and the vertical line separates the pre-treatment periods. Whiskers are the 95 percent C-LF confidence sets under cohort inference for $s = 1, \dots, 4$ at $M = 0.01$ log points per period, with the marker at the center of the plug-in identified set, for \CSNYT{}, with the rich covariates in (b) to (d). Confidence sets are pointwise.
}
\end{figure}

Figure~\ref{fig:cov-cohort} shows where the imprecision comes from, one panel per specification with the three cohorts' post-treatment sets at $M = 0.010$. Cohort $\mathcal G_4$ contributes four post-treatment periods. Its sets widen from about $\pm 0.06$ in the first period to between $\pm 0.2$ and $\pm 0.3$ in the fourth. Under IPW and AUG the later sets are visibly wider than under OR. Part of the reason is that these two estimators weight the never-treated counties by the fitted odds $\hat o_7(X_i)$, and for the pair $\mathcal G_7 \mid \mathcal G_4$ a few counties carry most of that weight. The pre-treatment $\mathcal G_7 \mid \mathcal G_4$ transported block biases therefore have standard errors of 0.07 to 0.15 under IPW and 0.02 to 0.09 under AUG, against 0.005 to 0.009 under OR. This transported term enters only at $s = 4$, so the wider sets at $s = 2$ and $s = 3$ come from the cohort's own coefficients and, at $s = 3$, from $U_{6 \mid 4}$. 

\begin{figure}[!h]
\centering
\subfigure[No covariates]{\includegraphics{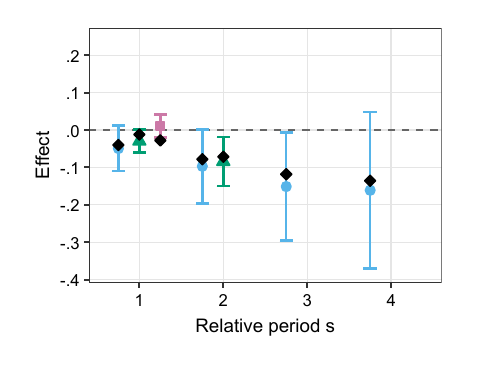}}\hfill
\subfigure[OR]{\includegraphics{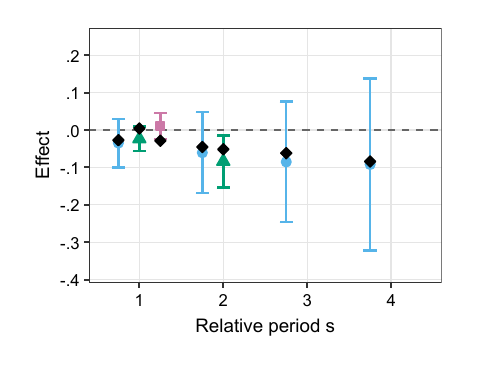}}\\[0.2em]
\subfigure[IPW]{\includegraphics{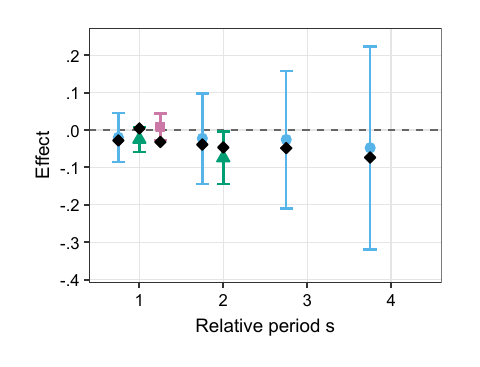}}\hfill
\subfigure[AUG]{\includegraphics{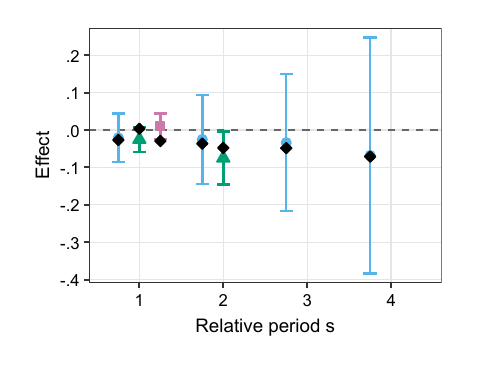}}\\[0.2em]
\includegraphics{figures/pdf/figC_minwage_cohort_legend.pdf}
\caption{Covariates: cohort-period effects at $M = 0.01$}
\label{fig:cov-cohort}
\vspace{0.5em}
\parbox{\linewidth}{\footnotesize
    \textit{Notes:} 95 percent C-LF confidence sets under cohort inference for the cohort-period effects $\bar{\tau}_{\mathcal G_g,s}$, in log points, at $M = 0.01$ log points per period, for \CSNYT{}, with the rich covariates in (b) to (d); colors follow the cohort palette of Section~\ref{sec:illustration} ($\mathcal G_4$ circles, $\mathcal G_6$ triangles, $\mathcal G_7$ squares), the marker sits at the center of the plug-in identified set, and black diamonds are the conventional estimates. Within each relative period the cohorts are offset horizontally. All panels share one vertical scale. Confidence sets are pointwise.
}
\end{figure}

\clearpage
\section{Notation}
\label{app:notation}
\setcounter{equation}{0}
\setcounter{figure}{0}
\setcounter{table}{0}

The table gives the paper's symbols and their first-use subsections. Vectors indexed by the cohort-period cells carry arrows, matrices are bold, and hats indicate estimates. The subscripts $\pre$ and $\post$ select pre-treatment and post-treatment cells from a stacked vector.

{\footnotesize
\setlength{\tabcolsep}{3pt}
\begin{longtable}{p{0.32\linewidth} p{0.56\linewidth} l}
\toprule
Symbol & Meaning & First use \\
\midrule
\endfirsthead
\toprule
Symbol & Meaning & First use \\
\midrule
\endhead
\bottomrule
\endlastfoot
\multicolumn{3}{l}{\textit{Indices, cohorts, and sets}} \\
$i$, $t$, $s$, $g$, $k$ & unit; calendar period; relative period, $t = t_g + s - 1$, with $s = 1$ the first treated period and $s = 0$ the last untreated period; cohort indices & 3.1 \\
$N$, $T$, $G$ & numbers of units, periods, and treated cohorts & 3.1 \\
$\mathcal G_g$, $t_g$; $\mathcal G_\infty$ & cohort adopting in period $t_g$; never-treated cohort & 3.1 \\
$N_g$, $N_\infty$ & cohort sizes, population-weighted in the community-health-center application & 3.1 \\
$\pre_g$ & pre-treatment periods of cohort $g$, $\{1, \dots, t_g - 1\}$ & 3.1 \\
$\mathcal C_{g,1}$, $\mathcal C_{g,s}$ & initial control group of cohort $g$; its not-yet-treated control group at post-treatment period $s$ & 3.2 \\
$\mathcal K_g(t)$, $s_k(t)$ & adjustment cohorts $\{k : t_g < t_k \leq t\}$; relative period of cohort $k$ at calendar period $t$ & 3.4 \\
$w_k$ & adjustment weight $N_k / (\sum_{j \geq k} N_j + N_\infty)$ & 3.4 \\
\addlinespace
\multicolumn{3}{l}{\textit{Outcomes, effects, and coefficients}} \\
$Y_{it}(g)$, $Y_{it}(\infty)$, $Y_{it}$ & potential outcomes under adoption in $t_g$ and never; observed outcome & 3.1 \\
$\alpha_i$, $\phi_t$ & unit and period effects of the imputation fit & 3.2 \\
$\tau_{\mathcal G_g,s}$, $\vec\tau_{\post}$ & cohort-period average treatment effect; stacked post-treatment effects & 3.1 \\
$\ell$, $\theta = \ell'\vec\tau_{\post}$ & target weights on the post-treatment cells; target & 3.1 \\
$\beta_{\mathcal G_g,s}$, $\vec\beta$, $\vec{\hat\beta}$ & population cohort-period coefficient; stacked; estimated & 3.5 \\
$\Sigma_N$, $\hat\Sigma_N$ & covariance of $\vec{\hat\beta}$, exact in the Gaussian model and asymptotic otherwise; its estimate & 3.5 \\
\addlinespace
\multicolumn{3}{l}{\textit{Bias objects and the map}} \\
$\delta_{\mathcal G_g,s}$, $\vec\delta$ & estimation bias of a cohort-period coefficient; stacked & 3.2 \\
$\Delta_{\mathcal G_g,s}$, $\vec\Delta$ & block bias: the cohort's untreated trend relative to its initial control group; stacked & 3.3 \\
$\mathbf W$ & map from block biases to estimation biases, $\vec\delta = \mathbf W\vec\Delta$, with blocks $\mathbf W_{\post,\pre}$ and $\mathbf W_{\post,\post}$ & 3.4; blocks 5.2 \\
$\nabla^2\Delta_{\mathcal G_g,s}$ & second difference along a cohort's path & 4.2 \\
$\mathbf A$, $\vec\beta_A$, $\vec\tau_A$, $\Sigma_A$, $\ell_A$; $\vec\delta_A$ & aggregation matrix; aggregate coefficients, effects, covariance, and target weights, $\mathbf A'\ell_A = (\vec 0', \ell')'$; aggregate bias path $\mathbf A\mathbf W\vec\Delta$ & 4.5 \\
\addlinespace
\multicolumn{3}{l}{\textit{Restrictions and bounds}} \\
$\Lambda_\Delta$, $\{\vec\Delta : \mathbf P\vec\Delta \leq \vec d\}$; $\Lambda_\delta = \mathbf W\Lambda_\Delta$ & restriction set on the block biases; a polyhedron; the implied restriction on the estimation biases & 4.1 \\
$\Lambda^{\mathrm{RM}}_\Delta(\bar M)$ & relative-magnitudes restriction, each cohort benchmarked by its largest pre-treatment change & 4.2 \\
$\Lambda^{\mathrm{SD}}_\Delta(M)$, $\Lambda^{\mathrm{SD,F}}_\Delta(M)$ & second-differences restriction, post-only and full path & 4.2, 5.2 \\
$M$, $\bar M$ & SD bound, in outcome units per period; RM bound, a dimensionless multiple of the benchmark & 4.2 \\
$\Lambda^{\mathrm{SD}}_A(M_A)$, $\mathcal S_A$, $M_A$ & post-only SD restriction on the aggregate bias path; the identified set under aggregate inference; the aggregate bound & 4.5 \\
$M_A^C$; $M_A^{\theta}$ & covering bound, the smallest aggregate bound at which the aggregate restriction accommodates the cohort restriction; target-specific bound, from which the identified set under aggregate inference contains the one under cohort inference & 4.5; A.9 \\
$M_C$; $M_A^S$, $M_A^E$ & cohort bound; aggregate bound at the same value and width-matched & 6.1 \\
$r$ & width-matching factor, the ratio of the cohort to the aggregate plug-in width at the same bound & 7 \\
\addlinespace
\multicolumn{3}{l}{\textit{Identified sets and confidence sets}} \\
$\mathcal S(\vec\beta, \Lambda_\Delta)$, $\mathcal S(\vec{\hat\beta}, \Lambda_\Delta)$ & identified set for $\theta$; plug-in identified set & 4.1 \\
$\mathcal C_N(\vec{\hat\beta}, \Sigma_N)$ & confidence set, computed with $\hat\Sigma_N$ in place of $\Sigma_N$ & 4.3 \\
$\theta_0$; $\alpha$, $\kappa = \alpha/10$ & candidate value of the target in each test; level; first-stage size of the C-LF hybrid & 4.3 \\
\addlinespace
\multicolumn{3}{l}{\textit{Fixed-length confidence intervals}} \\
$v = (v_{\pre}, \ell)$, $\hat\theta_v$ & affine weights with the post-treatment weights fixed at $\ell$; affine estimator $v'\vec{\hat\beta}$ & 5.1 \\
$\sigma(v)$, $B(v; \Lambda_\Delta)$ & standard error of $\hat\theta_v$; its worst-case bias over $\Lambda_\Delta$ & 5.1 \\
$\cvq_\alpha(\cdot)$, $\chi(v)$ & folded-normal quantile; half-length of the interval & 5.1 \\
$\mathbf D_{\pre,\pre}$, $\mathbf D_{\post,\pre}$, $\mathbf D_{\post,\post}$ & row blocks of the second-difference operator & 5.2 \\
$z$, $\hat z$; $u$; $\xi$ & pure pre-treatment second differences of the block-bias path and their estimate; post-only second differences; target exposure to $u$ & 5.2 \\
$v^{\mathrm P}$, $c = v^{\mathrm P\prime}\vec{\hat\beta}$, $B^{\mathrm P} = M\|\xi\|_1$ & post-only weights, post-only center, and worst-case bias & 5.2 \\
$\gamma$, $c + \gamma'\hat z$ & full-path weights on the pure pre-treatment second differences; center of the interval & 5.2 \\
\addlinespace
\multicolumn{3}{l}{\textit{Simulations}} \\
$\lambda$ & noise multiplier: multiplies the disturbances of the panel designs, hence the coefficient covariance by $\lambda^2$, and the covariance $\Sigma_0$ of the calibrated FLCI design by $\lambda^2$ & 6.1 \\
$m_0$ & reference curvature unit of the panel simulation designs, $0.131$ outcome units, $3.2$ standard errors of the first post-treatment effect at $\lambda = 1$ in Design I & 6.1 \\
$R$ & Monte Carlo replications & 6.1 \\
$\sigma_0$ & reference scale: standard error at $\lambda = 1$ of the cohort-size-weighted first post-treatment effect under the imputation estimator in the calibrated FLCI design & B \\
$\widehat{\mathrm{Cov}}_{\min}$, $\mathcal S_C$ & minimum empirical coverage over the population identified set under cohort inference $\mathcal S_C$ & B \\
$\nu$ & excess pure pre-treatment curvature in the stress design, in units of $\sigma_0$ & B.1 \\
$\pi_S$ & state share of the disturbance variance in Design III & B.2 \\
\addlinespace
\multicolumn{3}{l}{\textit{Covariates (Appendix~\ref{app:covariates})}} \\
$X_i$, $F_g$ & baseline covariates; their distribution in cohort $g$ & D.1 \\
$\mu^0_{a,t}(x)$, $\mu^g_{g,t}(x)$ & untreated conditional mean of cohort $a$ or of a unit set; treated conditional mean & D.1 \\
$q_k(x)$ & conditional share of cohort $k$ among the units untreated immediately before it adopts & D.2 \\
$U_{k \mid g,s}$, $\omega_{k \mid g}$ & transported block bias and transported share & D.3 \\
$\mathrm{TV}_{k \mid g}$ & total-variation norm of the difference between the two averaging measures & D.4 \\
$\vec\Delta^{+}$, $\vec\beta^{+}$, $\mathbf B$ & augmented block-bias and coefficient vectors; their map & D.5 \\
$\mathbf R$, $\mathbf N$, $\vec\eta$ & right inverse of $\mathbf B$; null-space basis; profiled nuisance & D.5 \\
\addlinespace
\multicolumn{3}{l}{\textit{Proofs (Appendix~\ref{app:proofs})}} \\
$\mathcal U_t$, $N_A$ & units not yet treated in period $t$; number of units in a set $A$ & A \\
$m_A(t)$, $\bar m_A(P)$, $\mathrm{TD}_{A,B}(t; P)$ & mean untreated potential outcome of a unit set; its average over periods $P$; trend difference between two unit sets from baseline $P$ & A \\
$\varkappa_g$, $\vec\varkappa$ & cohort constant $\Delta^{\Imp}_{\mathcal G_g,0}$ by which the two estimators' block-bias paths differ; stacked & A.4 \\
$\mathbf L$, $\mathbf L_{\mathsf c}$, $\vec{\hat s}_{\mathsf c}$, $\hat e_{it}$ & linear operator of the coefficients on the outcomes; its columns for cluster $\mathsf c$; cluster score; cell-mean residual & 3.5, A.3 \\
$G_{\mathsf c}$, $n$, $p$, $a_N$ & number of clusters; observations; outcome cells; finite-sample factor of the clustered covariance & 3.5, A.3 \\
$\mathbf J$ & map from the reduced coordinates, which drop each cohort's normalized cell, to the full cells & A.6 \\
$a$, $\mathbf K$, $\bar\ell$, $\tilde\tau$ & data part, nuisance matrix, and nuisance parametrization $\theta_0\bar\ell + \tilde\tau$ of the moments $m(\theta_0, \tilde\tau)$ in the proof of the post-only center & A.7 \\
$\vec\Delta^0$, $\mathbf E$, $\mathbf D_A$; $\vec a$, $\mathbf C$; $c_C$, $c_A$, $H_C$, $H_A$ & block-bias vector at $M_C = 0$, map from the post-only second differences to the block biases, post-only second differences of the aggregate path; aggregate second differences of the continued paths and their exposure to the cohort second differences; centers and half-widths per unit bound of the two identified sets & A.9 \\
$r$ (Appendix~\ref{app:proofs}) & a period in Sections~A and~A.2; the argument $B/\sigma$ of $\cvq_\alpha$ in Section~A.8, equal to $B^{\mathrm P}/\sigma(v^{\mathrm P})$ at $\gamma = 0$ & A \\
$\phi$, $\Phi$, $z_{1-\alpha}$; $J$ & standard normal density, distribution function, and $1 - \alpha$ quantile; number of pure pre-treatment second differences & A.8; $J$ in 5.2 \\
\end{longtable}
}

\clearpage
\bibliographystyle{ecta}
\bibliography{general_honestdid,introduction_references}

\end{document}